\documentclass[a4paper]{article}
\usepackage{a4wide}
\usepackage{graphicx,xcolor}
\usepackage{caption}
\usepackage{subcaption}
\usepackage{empheq}
\usepackage{amsthm}
\usepackage{amssymb, latexsym, mathrsfs}
\usepackage{amsmath}
\usepackage{dsfont}
\usepackage{enumerate}
\usepackage{algorithm}
\usepackage{color}
\usepackage{transparent}
\usepackage{subfig}
\usepackage{url}
\usepackage[affil-it]{authblk}
\usepackage{booktabs}

\usepackage[american,cuteinductors,smartlabels]{circuitikz}
\usepackage{tikz}
\usepackage{schemabloc}
\usetikzlibrary{shapes,arrows}

\theoremstyle{plain}

\newtheorem{assum}{Assumption}
\newtheorem{prbl}{Problem}

\newtheorem{rmk}{Remark}
\newtheorem{prop}{Proposition}
\newtheorem{lem}{Lemma}

\newcommand{\Rset}{\mathbb{R}}

\newcommand{\hd}{{\hat{d}}}

\newcommand{\bd}{{\bar{d}}}

\newcommand{\CC}{{\mathcal{C}}}
\newcommand{\DD}{{\mathcal{D}}}

\newcommand{\NN}{{\mathcal{N}}}

\newcommand{\PP}{{\mathcal{P}}}

\newcommand{\diag}{{\mbox{diag}}}                  
\newcommand{\norme}[2]{||{#1}||_{#2}}            
\newcommand{\rank}{\mbox{rank}}                    

\newcommand{\mbf}[1]{\mathbf{#1}}                  

\newcommand{\Zero}{\textbf{0}}

\newcommand{\subss}[2]{{#1}_{[#2]}}

\newcommand{\dx}{{\dot x}}

\newcommand{\matr}[1]{
\begin{bmatrix}
    #1
\end{bmatrix}
}

\begin{document}
	\title{\LARGE \bf A decentralized scalable approach to voltage control of DC islanded microgrids}
          \author[1]{Michele Tucci%
       \thanks{Electronic address:
         \texttt{michele.tucci02@universitadipavia.it}; Corresponding author}}
        \author[2]{Stefano Riverso%
       \thanks{Electronic address: \texttt{riverss@utrc.utc.com}}} 
\author[3]{Juan C. Vasquez%
       \thanks{Electronic address: \texttt{joz@et.aau.dk}} }
\author[3]{Josep M. Guerrero%
       \thanks{Electronic address: \texttt{juq@et.aau.dk}} }
\author[1]{Giancarlo Ferrari-Trecate%
       \thanks{Electronic address: \texttt{giancarlo.ferrari@unipv.it}} }

     \affil[1]{Dipartimento di Ingegneria Industriale e
       dell'Informazione\\Universit\`a degli Studi di Pavia}
     \affil[2]{United Technologies Research Center Ireland}  
     \affil[3]{Institute of Energy Technology, Aalborg University}
     \date{\textbf{Technical Report}\\ March, 2015}

     \maketitle
     \begin{abstract}
     We propose a new decentralized control scheme for
DC Islanded microGrids (ImGs) composed by several
Distributed Generation Units (DGUs) with a general interconnection
topology.
Each local controller regulates
to a reference value the voltage of the Point of Common Coupling (PCC)
of the corresponding DGU. Notably, off-line control design is
conducted in a Plug-and-Play (PnP) fashion meaning that (i) the
possibility of adding/removing a DGU without spoiling stability of the
overall ImG is checked through an optimization problem; (ii) when a DGU is plugged in or
out at most neighbouring DGUs have to update their controllers and
(iii) the synthesis of a local controller uses only information on the
corresponding DGU and lines connected to it. This guarantee total
scalability of control synthesis as the ImG size grows or DGU gets
replaced. Yes, under mild approximations of line dynamics, we formally
guarantee stability of the overall closed-loop ImG. The performance of
the proposed controllers is analyzed simulating different scenarios in
PSCAD.

     \end{abstract}

\newpage
     \section{Introduction}
In the recent years, the increasing penetration of renewable energy
sources has motivated a growing interest 
for microgrids, energy networks composed by the interconnection of
DGUs and loads \cite{lasseter2002certs}. Microgrids are
self-sustained electric systems that can supply local loads even in
islanded mode, i.e. 
disconnected from the main grid \cite{guerrero2013advanced}.  
Besides their use for electrifying remote areas, islands, or large
buildings, microgrids can be used for improving resilience to faults
and power quality in power networks \cite{guerrero2013advanced2}.
So far, research mainly focused on AC microgrids
\cite{lasseter2002certs,guerrero2013advanced, guerrero2013advanced2, riverso2014plug,Riverso2014c}.  
However, technological advances in power electronics
converters have considerably facilitated the operation of DC power
systems. This, together with the increasing use of DC
renewables (e.g. PV panels), batteries and loads
(e.g. electronic appliances, LEDs and electric vehicles),
has triggered a major interest in 
DC microgrids \cite{kwasinski2011quantitative, shafiee2014hierarchical, elsayed2015dc}. DC microgrids have also several advantages over
their AC counterparts. For instance, control of reactive power
or unbalanced electric signals are not an issue. On the
other hand, protection of DC systems is still a challenging problem \cite{elsayed2015dc}.

For AC ImGs a
key issue is to guarantee voltage and frequency stability by
controlling inverters interfacing energy sources with lines and
loads. This problem has received great attention and
several decentralized control schemes have been proposed, ranging from classic droop
control \cite{guerrero2013advanced, simpson2013synchronization}, to decentralized control
\cite{etemadi2012decentralized, riverso2014plug, Riverso2014c}. Some control design approaches are scalable, meaning
that the design of a local controller for a DGU is not based on
the knowledge of the whole ImG and the complexity of
local control design is independent of the ImG size. 
In addition, the method proposed in \cite{riverso2014plug, Riverso2014c} 
allows for the seamless plugging-in, unplugging and replacement of DGUs without
spoiling ImG stability.  Control design procedure with these features
have been termed PnP \cite{Riverso2013c,Riverso2014a,Riverso2014,stoustrup2009plug}.

Voltage stability is critical also in DC microgrids as
they cannot be directly coupled to an ``infinite-power'' source, such as the AC main grid, and therefore
they always operate in islanded mode. 
Existing controllers for the stabilization of DC ImGs are mainly based on
droop control \cite{shafiee2014hierarchical,lu2014improved}. So far, however, 
stability of the closed-loop systems has been analyzed only
for specific ImGs \cite{shafiee2014hierarchical,lu2014improved}.


In this paper we develop a totally scalable method for the synthesis
of 
decentralized controllers for DC ImGs. We propose a PnP design procedure 
where the synthesis of a local controller requires only the model of
the corresponding DGU and the parameters of transmission lines connected to it. 
Importantly, no specific information about any other DGU is needed. 
Moreover, when a DGU is plugged in or out, only DGUs
physically connected to it have to retune their local controllers.
As in \cite{riverso2014plug}, we exploit Quasi-Stationary Line
(QSL) approximations of line dynamics \cite{Venkatasubramanian1995}
and use structured Lyapunov functions for mapping control design
into a Linear Matrix Inequality (LMI) problem.  This also allows to
automatically deny plugging-in/out requests if these operations
spoil the stability of the ImG.

         In order to validate our results, we run several
         simulations in PSCAD using realistic models of Buck
         converters and associated filters. As a first test, we
         consider two radially connected DGUs
         \cite{shafiee2014modeling} and we show that, in spite of QSL
         approximations, PnP controllers lead to very good
         performances in terms of voltage tracking and robustness to
         unknown load dynamics. We also show how to
         embed PnP controllers in a bumpless transfer scheme \cite{aastrom2006advanced} so as to
         avoid abrupt changes of the control variables due to controller
         switching. Then, we consider an ImG with 5 DGUs
         arranged in a meshed topology including loops and discuss the
         real-time plugging-in and out of a DGU.

          The paper is organized as follows. In Section \ref{sec:Model} we present dynamical models of ImGs and the adopted line approximation. In Section \ref{sec:PnPctrl}, the procedure for performing PnP operations is described. In Section \ref{sec:Simresults} we assess performance of PnP controllers through simulation case studies. Section \ref{sec:conclusions} is devoted to some conclusions.
    \section{Model of a DC Microgrid}
          \label{sec:Model}      

This section discusses dynamic models of ImGs. For clarity, we start
by introducing an ImG consisting of two parallel
DGUs, then we generalize the model to ImGs composed of $N$ DGUs. 
Consider the scheme depicted in Figure \ref{fig:schemairandist_DC}
comprising two DGUs denoted with $i$ and $j$ and connected through a
DC line with an impedance specified by parameters $R_{ij}>0$ and
$L_{ij}>0$. At each DGU level, a DC voltage source represents
a generic renewable resource and a Buck converter is present in order
to supply a local DC
load connected to the PCC through a series $LC$ filter. For instance, the DC load
can be a combination of resistive electronic loads and negative
resistance of constant power loads. Furthermore, we assume that
loads are unknown and we treat them as current disturbances ($I_L$)
\cite{riverso2014plug,Babazadeh2013}. 
	      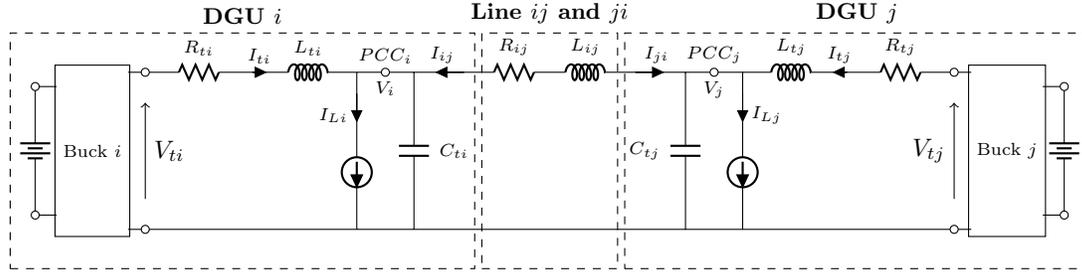
\begin{figure*}[!htb]
            \centering
           \ctikzset{bipoles/length=0.7cm}
\begin{circuitikz}[scale=.95,transform shape]
\ctikzset{current/distance=1}
\draw

node[] (Ti) at (0,0) {}
node[] (Tj) at ($(5.4,0)$) {}

node[ocirc] (Aibattery) at ([xshift=-4.5cm,yshift=0.9cm]Ti) {}
node[ocirc] (Bibattery) at ([xshift=-4.5cm,yshift=-0.9cm]Ti) {}

(Bibattery) to [battery] (Aibattery) {}
node [rectangle,draw,minimum width=1cm,minimum height=2.4cm] (bucki) at ($0.5*(Aibattery)+0.5*(Bibattery)+(0.8,0)$) {\scriptsize{Buck $i$}}
(Aibattery) to [short] ([xshift=0.3cm]Aibattery)
(Bibattery) to [short] ([xshift=0.3cm]Bibattery) 

node[ocirc] (Ai) at ($(Aibattery)+(1.54,0.2)$) {}
node[ocirc] (Bi) at ($(Bibattery)+(1.54,-0.2)$) {}
(Ai) to [short] ([xshift=-0.24cm]Ai)
(Bi) to [short] ([xshift=-0.24cm]Bi)
(Ai) to [R, l=\scriptsize{$R_{ti}$}] ($(Ai)+(1.5,0)$) {}
to [short,i=\scriptsize{$I_{ti}$}]($(Ai)+(1.6,0)$){}
to [L, l=\scriptsize{$L_{ti}$}]($(Ti)+(0,1.1)$){}
(Bi) to [short] ($(Ti)+(0,-1.1)$);
\begin{scope}[shorten >= 10pt,shorten <= 10pt,]
\draw[<-] (Ai) -- node[right] {$V_{ti}$} (Bi);
\end{scope};

\draw
($(Ti)+(0.4,1.1)$) node[anchor=north]{\scriptsize{$V_i$}}
($(Ti)+(0.4,1.1)$) node[anchor=south]{\scriptsize{$PCC_i$}}
($(Ti)+(0.4,1.1)$) node[ocirc](PCCi){}
($(Ti)+(0,1.1)$)--($(Ti)+(0,0.8)$) to [short,i>_=\scriptsize{$I_{Li}$}]($(Ti)+(0,0.5)$)
to [I]($(Ti)+(0,-1.1)$)
($(Ti)+(0.8,1.1)$) to [C, l=\scriptsize{$C_{ti}$}] ($(Ti)+(0.8,-1.1)$)

($(Ti)+(0,1.1)$) to [short] ($(Ti)+(0.8,1.1)$)
($(Ti)+(1.5,1.1)$) to [short,i_=\scriptsize{$I_{ij}$}] ($(Ti)+(1.2,1.1)$)--($(Ti)+(0.8,1.1)$)
($(Ti)+(1.5,1.1)$)--($(Ti)+(1.9,1.1)$) to [R, l=\scriptsize{$R_{ij}$}] ($(Ti)+(2.5,1.1)$) {}
to [L, l=\scriptsize{$L_{ij}$}]($(Tj)+(-1.5,1.1)$){}
($(Tj)+(-1.5,1.1)$) to [short,i=\scriptsize{$I_{ji}$}] ($(Tj)+(-1.2,1.1)$)--($(Tj)+(-0.8,1.1)$)
($(Tj)+(-0.8,1.1)$) to [short] ($(Tj)+(0,1.1)$)
($(Ti)+(0,-1.1)$) to [short] ($(Tj)+(0,-1.1)$)

($(Tj)+(-0.4,1.1)$) node[anchor=north]{\scriptsize{$V_j$}}
($(Tj)+(-0.4,1.1)$) node[anchor=south]{\scriptsize{$PCC_j$}}
($(Tj)+(-0.4,1.1)$) node[ocirc](PCCj){}
($(Tj)+(0,1.1)$)--($(Tj)+(0,0.8)$) to [short,i>=\scriptsize{$I_{Lj}$}]($(Tj)+(0,0.5)$)
to [I]($(Tj)+(0,-1.1)$)
($(Tj)+(-0.8,1.1)$) to [C, l_=\scriptsize{$C_{tj}$}] ($(Tj)+(-0.8,-1.1)$)

node[ocirc] (Ajbattery) at ([xshift=4.5cm,yshift=0.9cm]Tj) {}
node[ocirc] (Bjbattery) at ([xshift=4.5cm,yshift=-0.9cm]Tj) {}
(Bjbattery) to [battery] (Ajbattery) {}
node [rectangle,draw,minimum width=1cm,minimum height=2.4cm] (vsci) at ($0.5*(Ajbattery)+0.5*(Bjbattery)-(0.8,0)$) {\scriptsize{Buck $j$}}
(Ajbattery) to [short] ([xshift=-0.3cm]Ajbattery)
(Bjbattery) to [short] ([xshift=-0.3cm]Bjbattery)

node[ocirc] (Aj) at ($(Ajbattery)+(-1.54,0.2)$) {}
node[ocirc] (Bj) at ($(Bjbattery)+(-1.54,-0.2)$) {}
(Aj) to [short] ([xshift=0.24cm]Aj)
(Bj) to [short] ([xshift=0.24cm]Bj)
(Bj) to [short] ($(Tj)+(0,-1.1)$)
(Aj) to [R, l_=\scriptsize{$R_{tj}$}] ($(Aj)+(-1.5,0)$) {}
to [short,i_=\scriptsize{$I_{tj}$}]($(Aj)+(-1.6,0)$){}
($(Tj)+(0,1.1)$) to [L, l=\scriptsize{$L_{tj}$}]($(Aj)+(-1.6,0)$){};
\begin{scope}[shorten >= 10pt,shorten <= 10pt,]
\draw[<-] (Aj) -- node[left] {$V_{tj}$} (Bj);
\end{scope};

\draw
node [rectangle,draw,minimum width=6.5cm,minimum height=3.3cm,dashed,label=\small\textbf{DGU $i$}] (DGUi) at ($0.5*(Aibattery)+0.5*(Bibattery)+(2.9,0)$) {}
node [rectangle,draw,minimum width=6.5cm,minimum height=3.3cm,dashed,label=\small\textbf{DGU $j$}] (DGUj) at ($0.5*(Ajbattery)+0.5*(Bjbattery)-(2.9,0)$) {}
node [rectangle,draw,minimum width=1.9cm,minimum height=3.3cm,dashed,label=\small\textbf{Line $ij$ and $ji$}] (Lineij) at ($0.5*(DGUi.center)+0.5*(DGUj.center)+(0,0)$){}

;\end{circuitikz}
           \caption{Electrical scheme of a DC ImG composed of two radially connected DGUs with unmodeled loads.}
           \label{fig:schemairandist_DC}
          \end{figure*}

Applying Kirchoff's voltage law and
Kirchoff's current law to the electrical scheme of Figure
\ref{fig:schemairandist_DC}, it is possible to write the following set
of equations:


          \begin{subequations}
            \label{eq:sysdist}            
            \begin{empheq}[left=$DGU \emph{i}:\quad$\empheqlbrace]{align}
              \label{eq:sysdistB}\frac{dV_{i}}{dt} &= \frac{1}{C_{ti}}I_{ti}+\frac{1}{C_{ti}}I_{{ij}}-\frac{1}{C_{ti}}I_{Li} \\
              \label{eq:sysdistA}\frac{dI_{ti}}{dt}&= -\frac{R_{ti}}{L_{ti}}I_{ti}-\frac{1}{L_{ti}}V_{i}+\frac{1}{L_{ti}}V_{ti}
            \end{empheq}
            \begin{empheq}[left=$Line \emph{ij}:\quad$\empheqlbrace]{align}
              \label{eq:sysdistC}\frac{dI_{ij}}{dt}&= \frac{1}{L_{ij}}V_{j}-\frac{R_{ij}}{L_{ij}}I_{{ij}}-\frac{1}{L_{ij}}V_{i}
            \end{empheq}
            \begin{empheq}[left=$Line \emph{ji}:\quad$\empheqlbrace]{align}
              \label{eq:sysdistCji}\frac{dI_{ji}}{dt}&= \frac{1}{L_{ji}}V_{i}-\frac{R_{ji}}{L_{ji}}I_{{ji}}-\frac{1}{L_{ji}}V_{j}
            \end{empheq}
            \begin{empheq}[left=$DGU \emph{j}:\quad$\empheqlbrace]{align}
              \label{eq:sysdistD}\frac{dV_{j}}{dt} &= \frac{1}{C_{tj}}I_{tj}+\frac{1}{C_{tj}}I_{{ji}}-\frac{1}{C_{tj}}I_{Lj} \\
              \label{eq:sysdistE}\frac{dI_{tj}}{dt}&= -\frac{R_{tj}}{L_{tj}}I_{tj}-\frac{1}{L_{tj}}V_{j}+\frac{1}{L_{tj}}V_{tj}
            \end{empheq}
          \end{subequations}

 As in \cite{riverso2014plug}, we notice that from
 \eqref{eq:sysdistC} and \eqref{eq:sysdistCji} one gets two opposite line
 currents $I_{ij}$ and $I_{ji}$. This is equivalent to have a
 reference current entering in each DGU. We exploit the following
 assumption to ensure that $I_{ij}(t)=-I_{ji}(t)$, $\forall t\geq 0$.
          \begin{assum}
            \label{ass:lines}
            Initial states for the line currents fulfill
            $I_{ij}(0)=-I_{ji}(0)$. Furthermore, we set $L_{ij}=L_{ji}$ and $R_{ij}=R_{ji}$.
          \end{assum}

          \begin{rmk}
According to the terminology in Section 3.4 of \cite{Lunze1992}, the system in \eqref{eq:sysdistC}, \eqref{eq:sysdistCji} represents an expansion of the line
model one obtains introducing only a single state variable. System \eqref{eq:sysdist}
can also be viewed as a system of differential-algebraic equations,
given by (\ref{eq:sysdistB})-(\ref{eq:sysdistC}), (\ref{eq:sysdistD}),
(\ref{eq:sysdistE}) and $I_{ij}(t)=-I_{ji}(t)$.
\end{rmk}
At this point, we notice that adopting the above notation for the
lines, both DGU models have the same structure. In
particular, by recalling that the load current $I_{L*},\mbox{ }*\in{i,j}$
is treated as a disturbance, (\ref{eq:sysdist}) is the
following linear system         
\begin{equation}
            \label{eq:sysdistABCDM}
            \begin{aligned}
              \dot{x}(t) &= Ax(t)+Bu(t)+Md(t)\\
              y(t)       &= Cx(t)
            \end{aligned}
          \end{equation}
          where
          $x=[V_{i},I_{ti},I_{{ij}},I_{{ji}},V_{j,},I_{tj}]^T$
          is the state, $u=[V_{ti},V_{tj}]^T$
          the input, $d=[I_{Li},I_{Lj}]^T$ the
          disturbance and $y=[V_{i},V_{j}]^T$ the
          output of the system.
All matrices in (\ref{eq:sysdistABCDM}), which are obtained from
(\ref{eq:sysdist}), are given in Appendix \ref{sec:AppMasterSlave}.

Next, we show how to describe each DGU as
a dynamical system affected directly by state of the other DGU
connected to it. An approximate model will be proposed so that there
will be no need of
using the line current in the DGU state equations.
          
          \subsection{QSL model}
               \label{sec:newmodel}               
               As in \cite{Venkatasubramanian1995} and
               \cite{Akagi2007}, we set $\frac{d I_{{ij}}}{dt}=0$ and
               $\frac{d I_{{ji}}}{dt}=0$. Consequently, from
               \eqref{eq:sysdistC} and \eqref{eq:sysdistCji}, one gets the QSL model
               \begin{equation}
                 \label{eq:staticline}
                 \begin{aligned}
                   \bar{I}_{ij} = \frac{V_{j}}{R_{ij}}
                   -\frac{V_{i}}{R_{ij}}\\
                   \bar{I}_{ji} = \frac{V_{i}}{R_{ji}} -\frac{V_{j}}{R_{ji}}
                 \end{aligned}
               \end{equation}
               By replacing variables $I_{ij}$ and $I_{ji}$ in
               \eqref{eq:sysdistB} and \eqref{eq:sysdistD} with the right-hand side of
               \eqref{eq:staticline}, we obtain the following model of
               DGU $i$ 
               \begin{equation}
                 \label{eq:newDGU}
                 \text{DGU}~i:\quad\left\lbrace
                   \begin{aligned}
                     \frac{dV_{i}}{dt} &= \frac{1}{C_{ti}}I_{ti}-\frac{1}{C_{ti}}I_{Li}+\frac{1}{C_{ti}}\bar{I}_{ij}\\
                     \frac{dI_{ti}}{dt} &= -\frac{1}{L_{ti}}V_{i}-\frac{R_{ti}}{L_{ti}}I_{ti}+\frac{1}{L_{ti}}V_{ti}\\
                   \end{aligned}
                 \right.
               \end{equation}
               Switching indexes $i$ and $j$ in \eqref{eq:newDGU} one obtains the model of DGU $i$
               \begin{equation}
                 \label{eq:subsysDGUi}
                 \subss{\Sigma}{i}^{DGU} :
                 \left\lbrace
                 \begin{aligned}
                   \subss{\dot{x}}{i}(t) &= A_{ii}\subss{x}{i}(t) + B_{i}\subss{u}{i}(t)+M_{i}\subss{d}{i}(t)+ \subss\xi i(t)\\
                   \subss{y}{i}(t)       &= C_{i}\subss{x}{i}(t)\\
                   \subss{z}{i}(t)       &= H_{i}\subss{y}{i}(t)\\
                 \end{aligned}
                 \right.
               \end{equation}
               where $\subss{x}{i}=[V_{i},I_{ti}]^T$ is the state,
               $\subss{u}{i} = V_{ti}$ the control input,
               $\subss{d}{i} = I_{Li}$ the exogenous input and
               $\subss{z}{i} = V_{i}$ the controlled variable of the
               system. Moreover, $\subss y i(t)$ is the measurable
               output and we assume $\subss{y}{i}=\subss{x}{i}$, while
               $\subss\xi i(t)=A_{ij}\subss x j$ represents
               the coupling with DGU $j$.

The matrices of
               $\subss{\Sigma}{i}^{DGU}$ are obtained from
               \eqref{eq:newDGU} and they are provided in Appendix \ref{sec:AppMasterMaster}. As
               regards the line, we obtain the subsystem 
               \begin{equation}
                 \label{eq:subsysLine}
                 \subss{\Sigma}{ij}^{Line} :
                 \left\lbrace
                   \subss{\dot{x}}{l,ij}(t) = A_{ll,ij}\subss{x}{l,ij}(t) + A_{li,ij}\subss{x}{i}(t) + A_{lj,ij}\subss{x}{j}(t)\\
                 \right.
               \end{equation}
                with $\subss{x}{l,ij}=I_{ij}$ as the state of the line. The matrices of \eqref{eq:subsysLine} are derived from \eqref{eq:sysdistC} and reported in Appendix \ref{sec:AppMasterMaster}.
We have now all the ingredients to write the model of the overall
microgrid depicted in Figure \ref{fig:schemairandist_DC}. In
particular, from equations \eqref{eq:subsysDGUi} and
\eqref{eq:subsysLine}, we get
        \begin{equation}
                  \label{eq:overallmodeltwoDGU}
                  \begin{aligned}
                    \begin{bmatrix}
                      \subss{\dx}{i} \\
                      \subss{\dx}{j} \\
                      \subss{\dx}{l,ij} \\
                      \subss{\dx}{l,ji}
                    \end{bmatrix} 
                    &= 
                    \begin{bmatrix}
                      A_{ii} & A_{ij} & 0 & 0 \\
                      A_{ji} & A_{jj} & 0 & 0 \\
                      A_{li,ij} & A_{lj,ij} & A_{ll,ij} & 0 \\
                      A_{li,ji} & A_{lj,ji} & 0 & A_{ll,ji}
                    \end{bmatrix}
                    \begin{bmatrix}
                      \subss{x}{i} \\
                      \subss{x}{j} \\
                      \subss{x}{l,ij} \\
                      \subss{x}{l,ji}
                    \end{bmatrix}
                    +
                    \begin{bmatrix}
                      B_{i} & 0\\
                      0 & B_{j} \\
                      0 & 0  \\
                      0 & 0
                    \end{bmatrix}
                    \begin{bmatrix}
                      \subss{u}{i} \\
                      \subss{u}{j}  
                    \end{bmatrix}
                    +
                    \begin{bmatrix}
                      M_{i} & 0\\
                      0 & M_{j} \\
                      0 & 0  \\
                      0 & 0
                    \end{bmatrix}
                    \begin{bmatrix}
                      \subss{d}{i} \\
                      \subss{d}{j}  
                    \end{bmatrix}
                    \\		
                    \begin{bmatrix}
                      \subss{y}{i}\\
                      \subss{y}{j}
                    \end{bmatrix}
                    &=
                    \begin{bmatrix}
                      C_{1} & 0 & 0 & 0 \\
                      0 & C_{2} & 0 & 0 \\
                    \end{bmatrix}
                    \begin{bmatrix}
                      \subss{x}{i} \\
                      \subss{x}{j} \\
                      \subss{x}{l,ij} \\
                      \subss{x}{l,ji} 
                    \end{bmatrix}
                    \\
                    \begin{bmatrix}
                      \subss{z}{i}\\
                      \subss{z}{j}
                    \end{bmatrix}
                    &=
                    \begin{bmatrix}
                      H_{i} & 0\\
                      0 & H_{j}
                    \end{bmatrix}
                    \begin{bmatrix}
                      \subss{y}{i} \\
                      \subss{y}{j} \\ 
                    \end{bmatrix}.
                  \end{aligned}
                \end{equation}
\begin{rmk}
Consider the structure of matrix A
  \begin{equation*}
                    \label{eq:blktriangmatrix}
                    A=\left[\begin{array}{cc|cc}
			A_{ii} & A_{ij} & 0 & 0  \\
			A_{ji} & A_{jj} & 0 & 0 \\ \hline	
			A_{li,ij} & A_{lj,ij} & A_{ll,ij} & 0 \\
                        A_{li,ji} & A_{lj,ji} & 0 & A_{ll,ji}
                      \end{array}\right]
                  \end{equation*}
We notice that A is block-triangular, therefore its eigenvalues are
given by the union of those of $\matr{ A_{ii} & A_{ij} \\ A_{ji} &
  A_{jj} }$, $A_{ll,ij}$ and $A_{ll,ji}$. Moreover, we have
$A_{ll,ij}=A_{ll,ji}$. By virtue of the positivity of the line
parameters, line dynamics is asymptotically stable. As a consequence,
stability of \eqref{eq:overallmodeltwoDGU}  the depends on the
stability of local DGUs connected through the QSL model \eqref{eq:staticline}. 
Hence, designing decentralized controllers $\subss{u}{*}= k_*(\subss y
*)$, $*\in\{i,j\}$, such that the connection of the DGUs is
asymptotically stable implies stability of the overall closed-loop
model of the microgrid. We refer to the resulting system as QSL-ImG model.
\end{rmk}
                
          \subsection{QSL model of a  microgrid composed of $N$ DGUs}
               In this section, a generalization of model \eqref{eq:subsysDGUi} to ImGs composed of $N$ DGUs is
               presented. Let
               $\DD=\{1,\ldots,N\}£$. First, we call
               two DGUs neighbours if there is a
               transmission line connecting them. Then, we denote with
               $\NN_i\subset\DD$ the subset of neighbours of DGU
               $i$. We highlight that the neighbouring relation is symmetric,
               consequently $j\in\NN_i$ implies $i\in\NN_j$. In order
               to describe the dynamics of DGU $i$, we use model
               \eqref{eq:subsysDGUi}, with $\subss\xi i
               =\sum_{j\in\NN_i} A_{ij}\subss{x}{j}(t)$. The new
               matrices of $\subss{\Sigma}{i}^{DGU}$ are given in Appendix
               \ref{sec:AppNDGunit} while the overall QSL-ImG
               model can be written as follows
  \begin{subequations}
                 \label{eq:stdform}
                 \begin{align}
                   \label{eq:stdformA}\mbf{\dot{x}}(t) &= \mbf{Ax}(t) + \mbf{Bu}(t)+ \mbf{Md}(t)\\
                   \label{eq:stdformB}\subss{\dot{x}}{l,ij}(t) &= A_{ll,ij}\subss{x}{l,ij}(t) + A_{li,ij}\subss{x}{i}(t) + A_{lj,ij}\subss{x}{j}(t),~\forall i\in\DD,~\forall j\in\NN_i
                 \end{align}
               \end{subequations}
               \begin{equation}
                 \begin{aligned}
                   \label{eq:stdformOut}
                   \mbf{y}(t)       &= \mbf{Cx}(t)\\
                   \mbf{z}(t)       &= \mbf{Hy}(t)
                 \end{aligned}
               \end{equation}
               where $\mbf x = (\subss x 1,\ldots,\subss x
               N)\in\Rset^{2N}$, $\mbf u = (\subss u 1,\ldots,\subss u
               N)\in\Rset^{N}$, $\mbf d = (\subss d 1,\ldots,\subss d
               N)\in\Rset^{N}$, $\mbf y = (\subss y 1,\ldots,\subss y
               N)\in\Rset^{2N}$, $\mbf z = (\subss z 1,\ldots,\subss z
               N)\in\Rset^{N}$. Matrices $\mbf{A}$, $A_{ll,ij}$,
               $A_{li,ij}$, $A_{lj,ij}$, $\mbf{B}$, $\mbf
               M$, $\mbf C$ and $\mbf H$ are reported in Appendix
               \ref{sec:AppMasterMaster} and \ref{sec:AppNDGunit}. 

Note that neither $\mbf y$ nor $\mbf z$ depend upon states $\subss x
{l,ij}$. Moreover, even $\subss x {l,ij}$ does not influence $\mbf x$. Hence, equations \eqref{eq:stdformB} will be omitted in the sequel.               
               
     \section{Plug-and-Play decentralized voltage control}
          \label{sec:PnPctrl}
	  \subsection{Decentralized control scheme with integrators}
               \label{sec:ctrlint}
               Let $\mbf{z_{ref}}(t)$ denote the constant desired
               reference trajectory for the output $\mbf{z}(t)$. In
               order to track asymptotically $\mbf{z_{ref}}(t)$ when
               $\mbf{d}(t)$ is constant, we consider the augmented ImG
               model with integrators \cite{Skogestad1996}. A necessary condition for having that the steady-state
               error $\mbf{e}(t)=\mbf{z_{ref}}(t)-\mbf{z}(t)$ tends to
               zero as $t\rightarrow\infty$, is that for arbitrary
               constant signals $\mbf{d}(t)=\mbf{\bar d}$ and $\mbf{z_{ref}}(t)=\mbf{\bar
               z_{ref}}$, there are equilibrium states and inputs $\mbf{\bar
               x}$ and $\mbf{\bar u}$ verifying
               \begin{equation}
                 \begin{aligned}
                   \Zero &= \mbf{A\bar{x}}+\mbf{B\bar{u}}+\mbf{M\bd}\\
                   \mbf{\bar z_{ref}} &= \mbf{HC\bar{x}}
                 \end{aligned}
               \end{equation}
               \begin{equation}
                 \label{eq:cond_integrators}
                 \Gamma\begin{bmatrix}
                   \mbf{\bar{x}}\\
                   \mbf{\bar{u}}
                 \end{bmatrix}
                 =\begin{bmatrix}
                   \Zero & \mbf{-M}\\
                   \mbf{I} & \Zero
                 \end{bmatrix}
                 \begin{bmatrix}
                   \mbf{\bar z_{ref}}\\
                   \mbf{\bd}
                 \end{bmatrix},\quad
                 \Gamma = \begin{bmatrix}
                   \mbf{A} & \mbf{B}\\
                   \mbf{HC} & \Zero
                 \end{bmatrix} \in \Rset^{3N\times 3N}
               \end{equation}
               \begin{prop}
                 \label{prop:target}
                 Given $\mbf{\bar z_{ref}}$ and $\mbf{\bd}$, vectors
                 $\mbf{\bar{x}}$ and $\mbf{\bar{u}}$ satisfying \eqref{eq:cond_integrators} always exist.
               \end{prop}
               \begin{proof}
                 From \cite{Skogestad1996}, we know that exists
                 $\mbf{\bar{x}},\mbf{\bar{u}}$ verifying
                 \eqref{eq:cond_integrators} if and only if the following two conditions are fulfilled:
                 
                 \begin{enumerate}[(i)]
                 \item\label{enu:cond_integrators_p1} 
                   The number of controlled variables is not greater than the number of control inputs.
                 \item\label{enu:cond_integrators_p2} The system under
                   control has no invariant zeros (i.e. $\rank{(\Gamma)}=3N$).
                 \end{enumerate}
                 Condition (\ref{enu:cond_integrators_p1}) is
                 fulfilled since from \eqref{eq:subsysDGUi} one
                 has that $\subss u i$ and $\subss z i$ have the same
                 size, $\forall i\in\DD$. In order to prove Condition
                 (\ref{enu:cond_integrators_p2}), we exploit the definition of matrices $\mbf A$, $\mbf B$, $\mbf C$ and $\mbf H$ and the fact that electrical parameters are positive.
               \end{proof}
              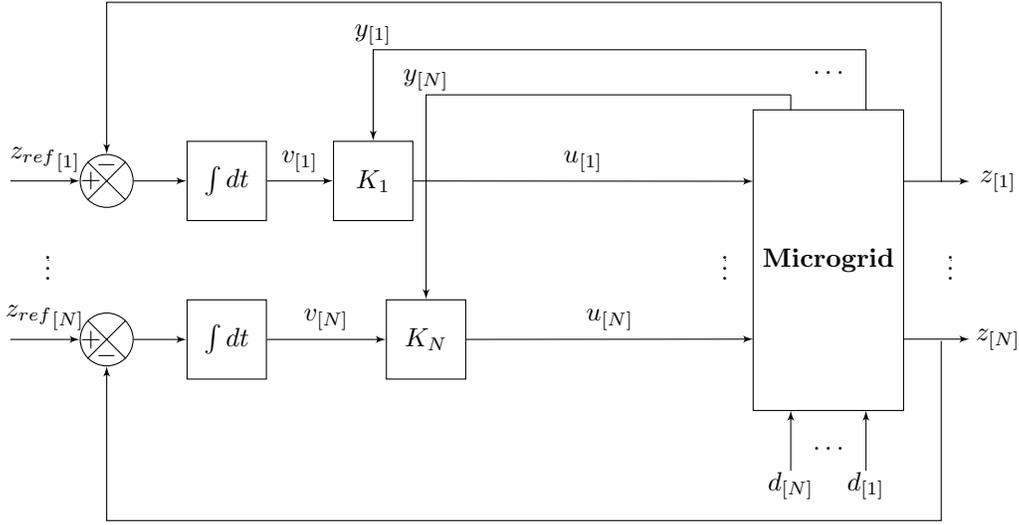
\begin{figure}[!htb]
                 \centering
                 \tikzstyle{input} = [coordinate]
\tikzstyle{output} = [coordinate]
\tikzstyle{guide} = []
\tikzstyle{block} = [draw, rectangle, minimum height=1cm]

\begin{tikzpicture}
  \sbEntree{zref1}
  \sbDecaleNoeudy[3]{zref1}{zrefj}
  \sbDecaleNoeudy[3]{zrefj}{zrefN}
  \node [block, right of=zrefj,node distance=11cm,minimum height=4cm, minimum width=2cm] (microgrid) {\textbf{Microgrid}};
  \node [guide, right of=zrefj,xshift=-0.5cm] (zrefjline) {};
  \draw [draw] (zrefjline) -| node{$\vdots$} (zrefjline);
  
  \sbComph{sumret1}{zref1}
  \sbBloc{integrator1}{$\int dt$}{sumret1}
  \sbBloc[2.5]{controller1}{$K_1$}{integrator1}
  \sbRelier[$\subss{z_{ref}}{1}$]{zref1}{sumret1}
  \sbRelier{sumret1}{integrator1}
  \sbRelier[$\subss v 1$]{integrator1}{controller1}
  \node [guide, left of=microgrid,yshift=1.05cm,xshift=0.125cm] (u1end) {};
  \sbRelier[$\subss u 1$]{controller1}{u1end}
  
  \sbComp{sumretN}{zrefN}
  \sbBloc{integratorN}{$\int dt$}{sumretN}
  \sbBloc[4.5]{controllerN}{$K_N$}{integratorN}
  \sbRelier[$\subss{z_{ref}}{N}$]{zrefN}{sumretN}
  \sbRelier{sumretN}{integratorN}
  \sbRelier[$\subss v N$]{integratorN}{controllerN}
  \node [guide, left of=microgrid,yshift=-1.05cm,xshift=0.125cm] (uNend) {};
  \sbRelier[$\subss u N$]{controllerN}{uNend}
  
  \node [output, below of=microgrid,yshift=-1.8cm,xshift=0.5cm] (d1out) {};
  \node [output, below of=microgrid,yshift=-1.0cm,xshift=0.5cm] (d1outstart) {};
  \draw [draw,->,>=latex'] (d1out) -| node[yshift=-0.2cm]{$\subss d 1$} (d1outstart);
  
  \node [output, below of=microgrid,yshift=-1.5cm] (djout) {};
  \draw [draw] (djout) -| node {$\ldots$} (djout);
  
  \node [output, below of=microgrid,yshift=-1.8cm,xshift=-0.5cm] (dNout) {};
  \node [output, below of=microgrid,yshift=-1.0cm,xshift=-0.5cm] (dNoutstart) {};
  \draw [draw,->,>=latex'] (dNout) -| node[yshift=-0.2cm]{$\subss d N$} (dNoutstart);

  \node [output, above of=microgrid,yshift=1.8cm,xshift=0.5cm] (y1out) {};
  \node [output, above of=microgrid,yshift=1.0cm,xshift=0.5cm] (y1outstart) {};
  \node [output, above of=microgrid,yshift=1.5cm] (yjout) {};
  \node [output, above of=microgrid,yshift=1.2cm,xshift=-0.5cm] (yNout) {};
  \node [output, above of=microgrid,yshift=1.0cm,xshift=-0.5cm] (yNoutstart) {};
  \draw [draw] (y1outstart) -- node {} (y1out);
  \draw [draw,->,>=latex'] (y1out) -| node[yshift=0.2cm]{$\subss y 1$} (controller1);
  \draw [draw] (yjout) -| node {$\ldots$} (yjout);
  \draw [draw] (yNoutstart) -- node {} (yNout);
  \draw [draw,->,>=latex'] (yNout) -| node[yshift=0.2cm]{$\subss y N$} (controllerN);
  
  \node [guide, right of=microgrid,yshift=1.05cm,xshift=-.125cm] (z1) {};
  \node [guide, right of=microgrid,yshift=1.165cm,xshift=0.5cm] (z1near) {};
  \node [guide, right of=microgrid,yshift=1.05cm,xshift=1cm] (z1end) {};
  \draw [draw,->,>=latex',near end,swap] (z1) -- node[xshift=0.6cm] {$\subss{z} 1$} (z1end);
  \sbRenvoi[-6.8]{z1near}{sumret1}{}

  \node [guide, right of=microgrid,yshift=-1.05cm,xshift=-.125cm] (zN) {};
  \node [guide, right of=microgrid,yshift=-0.95cm,xshift=0.5cm] (zNnear) {};
  \node [guide, right of=microgrid,yshift=-1.05cm,xshift=1cm] (zNend) {};
  \draw [draw,->,>=latex',near end,swap] (zN) -- node[xshift=0.6cm] {$\subss{z} N$} (zNend);
  \sbRenvoi[6.8]{zNnear}{sumretN}{}
  
  \node [guide, right of=microgrid,xshift=0.5cm] (zj) {};
  \draw [draw] (zj) -| node {$\vdots$} (zj);
  
  \node [guide, left of=microgrid,xshift=-0.5cm] (uj) {};
  \draw [draw] (uj) -| node {$\vdots$} (uj);
  
\end{tikzpicture}
                 \caption{Control scheme with integrators for the overall
                   augmented model.}
                 \label{fig:schemaint}
               \end{figure}
               The dynamics of the integrators is (see Figure \ref{fig:schemaint})
               \begin{equation}
                 \label{eq:integrators}
                 \begin{aligned}
                   \subss{\dot{v}}{i}(t) = \subss{e}{i}(t) &= \subss{z_{ref}}{i}(t)-\subss{z}{i}(t) \\
                   &= \subss{z_{ref}}{i}(t)-H_{i}C_{i}\subss{x}{i}(t),
                 \end{aligned}
               \end{equation}
               and hence, the DGU model augmented with integrators is
               \begin{equation}
                 \label{eq:modelDGUgen-aug}
  \subss{\hat{\Sigma}}{i}^{DGU} :
                 \left\lbrace
                 \begin{aligned}
                   \subss{\dot{\hat{x}}}{i}(t) &= \hat{A}_{ii}\subss{\hat{x}}{i}(t) + \hat{B}_{i}\subss{u}{i}(t)+\hat{M}_{i}\subss{\hat{d}}{i}(t)+\subss{\hat\xi}i(t)\\
                   \subss{\hat{y}}{i}(t)       &= \hat{C}_{i}\subss{\hat{x}}{i}(t)\\
                   \subss{z}{i}(t)       &= \hat{H}_{i}\subss{\hat{y}}{i}(t)
                 \end{aligned}
                 \right.
               \end{equation}
               where $\subss{\hat{x}}{i}=[\subss{x^T}
               i,v_{i,}]^T\in\Rset^3$ is the state,
               $\subss{\hat{y}}{i}=\subss{\hat{x}}{i}\in\Rset^3$ is
               the measurable output,
               $\subss{\hat{d}}{i}=[\subss{d}{i},\subss{z_{ref}}{i}]^T\in\Rset^2$
               collects the exogenous signals (both current of the
               load and reference signals) and
               $\subss{\hat\xi}i(t)=\sum_{j\in\NN_i}\hat{A}_{ij}\subss{\hat{x}}{j}(t)$. Matrices
               in \eqref{eq:modelDGUgen-aug} are defined as follows
  \begin{equation}
                 \label{eq:augith}
                 \hat{A}_{ii}=\begin{bmatrix}
                   A_{ii} & 0\\
                   -H_{i}C_{i} & 0
                 \end{bmatrix}
                 \hat{A}_{ij}=\begin{bmatrix}
                   A_{ij} &0\\
                   0&0
                 \end{bmatrix}
                 \hat{B}_{i}=\begin{bmatrix}
                   B_{i}\\
                   0
                 \end{bmatrix}
                 \hat{C}_{i}=\begin{bmatrix}
                   C_{i} & 0\\
                   0 & I
                 \end{bmatrix}
                 \hat{M}_{i}=\begin{bmatrix}
                   M_{i} & 0 \\
                   0 & 1
                 \end{bmatrix}
                 \hat{H}_{i}=\begin{bmatrix}
                   H_{i} & 0
                 \end{bmatrix}.
               \end{equation}             
               Through the following proposition we make sure that the pair $(\hat{A}_{ii},\hat{B}_{i})$ is controllable, thus system \eqref{eq:modelDGUgen-aug} can be stabilized.
               \begin{prop}
                 \label{prop:local_controllability}
                 The pair $(\hat{A}_{ii},\hat{B}_i)$ is controllable.
               \end{prop}
               \begin{proof}
                 Using the definition of controllability matrix, we get
                 \begin{small}
                   \begin{equation}
                     \label{eq:ctrb}
                     \begin{aligned}
                       \hat{M}^C_i &= \begin{bmatrix}
                         \hat B_{i} & \hat{A}_{ii}\hat{B}_i & \hat{A}_{ii}^2\hat{B}_i 
                       \end{bmatrix}\\
                       &= \underbrace{\begin{bmatrix}
                           A_{ii} & B_i \\ -H_iC_i & 0 
                         \end{bmatrix}}_{\hat{M}^C_{i,1}}\underbrace{\begin{bmatrix}
                           0 & B_i & A_{ii}B_i & A_{ii}^2B_i  \\ I & 0 & 0 & 0 
                         \end{bmatrix}}_{\hat{M}^C_{i,2}}.
                     \end{aligned}
                   \end{equation}
                 \end{small}
                 Matrices $\hat{M}^C_{i,1}$ and $\hat{M}^C_{i,2}$ have always full rank, since all electrical parameters are positive, hence $\rank(\hat{M}^C_i)=3$. Therefore the pair $(\hat{A}_{ii},\hat{B}_i)$ is controllable.
             \end{proof}
               
               The overall augmented system is obtained from \eqref{eq:modelDGUgen-aug} as
               \begin{equation}
                 \label{eq:sysaugoverall_1}
                 \left\lbrace
                   \begin{aligned}
                     \mbf{\dot{\hat{x}}}(t) &= \mbf{\hat{A}\hat{x}}(t) + \mbf{\hat{B}u}(t)+ \mbf{\hat{M}\hat{d}}(t)\\
                     \mbf{\hat{y}}(t)       &= \mbf{\hat{C}\hat{x}}(t)\\
                     \mbf{z}(t)       &= \mbf{\hat{H}\hat{y}}(t)
                   \end{aligned}
                 \right.
               \end{equation}
               where $\mbf{\hat{x}}$, $\mbf{\hat{y}}$ and $\mbf{\hat{d}}$ collect variables $\subss{\hat{x}}{i}$, $\subss{\hat{y}}{i}$ and $\subss{\hat{d}}{i}$ respectively, and matrices $\mbf{\hat{A}}, \mbf{\hat{B}}, \mbf{\hat{C}}, \mbf{\hat{M}}$ and $\mbf{\hat{H}}$ are obtained from systems \eqref{eq:modelDGUgen-aug}.

	  \subsection{Decentralized PnP control}
          \label{sec:riferimento}
               This section presents the adopted control approach that
               allows us to design local controlles while guaranteeing
               asymptotic stability for the augmented system
               \eqref{eq:sysaugoverall_1}. Local controllers are
               synthesized in a decentralized fashion permitting PnP
               operations. Let us equip each DGU $\subss{\hat{\Sigma}}{i}^{DGU}$ with the following state-feedback controller
               \begin{equation}
                 \label{eq:ctrldec}
                 \subss{\CC}{i}:\qquad \subss{u}{i}(t)=K_{i}\subss{\hat{y}}{i}(t)=K_{i}\subss{\hat{x}}{i}(t)
               \end{equation}
               where $K_{i}\in\Rset^{1\times3}$ and 
               controllers $\subss{\CC}{i}$, $i\in\DD$ are
               decentralized since the computation of
               $\subss{u}{i}(t)$ requires the state of
               $\subss{\hat{\Sigma}}{i}^{DGU}$ only. Let nominal
               subsystems be given by $\subss{\hat{\Sigma}}{i}^{DGU}$
               without coupling terms $\subss{\hat\xi}i(t)$. 
               We aim to design local controllers $\subss{\CC}{i}$
               such that the nominal closed-loop subsystem 
               \begin{equation}
                 \label{eq:modelDGUgen-aug-closed}
                 \left\lbrace
                   \begin{aligned}
                     \subss{\dot{\hat{x}}}{i}(t) &= (\hat{A}_{ii}+ \hat{B}_{i}K_{i})\subss{\hat{x}}{i}(t)+\hat{M}_{i}\subss{\hd}{i}(t)\\
                     \subss{\hat{y}}{i}(t)       &= \hat{C}_{i}\subss{\hat{x}}{i}(t)\\
                     \subss{z}{i}(t)       &= \hat{H}_{i}\subss{\hat{y}}{i}(t)
                   \end{aligned}
                 \right.
               \end{equation}
is asymptotically stable. From Lyapunov theory, we know
               that if there exists a symmetric matrix $P_{i}\in\Rset^{3\times3}, P_{i}>0$ such that 
               \begin{equation}
                 \label{eq:Lyapeqnith}
                 (\hat{A}_{ii}+\hat{B}_{i}K_{i})^T
                 P_{i}+P_{i}(\hat{A}_{ii}+\hat{B}_{i}K_{i})<0,
               \end{equation}
               then the nominal closed-loop subsystem equipped with
               controller $\subss{\CC}{i}$ is asymptotically stable.
               Similarly, the closed-loop QSL-ImG be given by
               \eqref{eq:sysaugoverall_1} and \eqref{eq:ctrldec} 
               \begin{equation}
                 \label{eq:sysaugoverallclosed}
                 \left\lbrace
                   \begin{aligned}
                     \mbf{\dot{\hat{x}}}(t) &= (\mbf{\hat{A}+\hat{B}K})\mbf{\hat{x}}(t)+ \mbf{\hat{M}\hd}(t)\\
                     \mbf{\hat{y}}(t)       &= \mbf{\hat{C}\hat{x}}(t)\\
                     \mbf{z}(t)       &= \mbf{\hat{H}\hat{y}}(t)
                   \end{aligned}
                 \right.
               \end{equation}
is asymptotically stable if matrix $\mbf{P}=\diag(P_{1},\dots,P_{N})$ satisfies 
               \begin{equation}
                 \label{eq:Lyapeqnoverall}
                 (\mbf{\hat{A}}+\mbf{\hat{B}K})^T \mbf{P}+\mbf{P}(\mbf{\hat{A}}+\mbf{\hat{B}K})<0
               \end{equation}
               where $\mbf{\hat{A}}$, $\mbf{\hat{B}}$ and $\mbf{K}$
               collect matrices $\hat{A}_{ij}$, $\hat{B}_i$ and $K_i$,
               for all $i,~j\in\DD$. We want to emphasize that, in
               general, \eqref{eq:Lyapeqnith} does not imply
               \eqref{eq:Lyapeqnoverall}, since one can show that
               decentralized design of local controllers can fail to
               guarantee voltage stability of the whole ImG, if
               coupling among DGUs is neglected (see Appendix B in
               \cite{Riverso2014c} for an example in the case of AC ImGs).
 In order to derive conditions such that \eqref{eq:Lyapeqnith} guarantees \eqref{eq:Lyapeqnoverall}, we first define $\mbf{\hat{A}_{D}}=\diag(\hat{A}_{ii},\dots,\hat{A}_{NN})$ and $\mbf{\hat{A}_{C}=\hat{A}-\hat{A}_{D}}$.
Then, we exploit the following assumptions to ensure asymptotic stability of the closed-loop QSL-ImG.
               \begin{assum} 
                 \label{ass:ctrl}
                 \begin{enumerate}[(i)]
                 \item\label{assum:pstruct} Decentralized controllers $\subss{\CC}{i}$, $i\in\DD$ are designed such that \eqref{eq:Lyapeqnith} holds with 
                   \begin{equation}
                     \small \label{eq:pstruct}
                     P_{i}=\left( \begin{array}{c|cc}
                         \eta_i & 0 & 0 \\
                         \hline
                         0 & \bullet & \bullet\\
                         0 & \bullet & \bullet\\
                         \end{array}\right)
                   \end{equation}
                   where $\bullet$ denotes an arbitrary entry and $\eta_i>0$ is a local parameter.
                 \item\label{assum:etasmallest} It holds $\frac{\eta_i } {R_{ij}C_{ti}}\approx 0$, $\forall i\in\DD$, $\forall j\in\NN_i$.
                 \end{enumerate}
               \end{assum}

               As regards Assumption
               \ref{ass:ctrl}-(\ref{assum:pstruct}), we will show later that checking the existence of $P_i$ as in \eqref{eq:pstruct} and $K_i$ fulfilling \eqref{eq:Lyapeqnith} leads to solving a convex optimization problem. 
On the other hand, there exist different ways to fulfill Assumption
\ref{ass:ctrl}-(\ref{assum:etasmallest}). In fact, when an upper bound
to all ratios $\frac{1} {R_{ij}C_{ti}}$ (which depend upon line
parameters only) is known, one can simply set the control design
parameter $\eta_i$ sufficiently small. However, if networks are spread
over a small area, the impedances are small and predominantly resistive. Therefore, one has $\frac{1} {R_{ij}C_{ti}}\approx 0$ by construction and bigger values of $\eta_i$ can be used for synthesizing local controllers.
               \begin{prop}
                 \label{prop:ctrldec}
                 Let Assumption \ref{ass:ctrl} holds. Then, the overall closed-loop QSL-ImG is asymptotically stable.
               \end{prop}
               \begin{proof}                 
                 We have to show that \eqref{eq:Lyapeqnoverall}
                 holds, which is equivalent to prove that
                 \begin{equation}
                   \label{eq:Lyap2part}
                   \underbrace{\mbf{(\hat{A}_{D}+\hat{B}K)}^T \mbf{P+P(\hat{A}_{D}+\hat{B}K)}}_{(a)}+\underbrace{\mbf{\hat{A}_{C}}^T \mbf{P+P\hat{A}_{C}}}_{(b)}<0.
                 \end{equation}
                 First, we highlight that term $(a)$ is a block diagonal matrix that collects on the diagonal all left hand sides of \eqref{eq:Lyapeqnith}. It follows that term $(a)$ is a negative definite matrix. Next, we show that term $(b)$ is zero. In particular, each block $(i,j)$ of term $(b)$ can be written as 
 \begin{displaymath}
\left\{ \begin{array}{ll}
 P_i\hat{A}_{ij}+\hat{A}_{ji}^TP_j & \hspace{7mm}\mbox{if } j\in\mathcal{N}_i \\
 0 & \hspace{7mm}\mbox{otherwise}
  \end{array} \right.
\end{displaymath}

Using Assumption \ref{ass:ctrl}-(\ref{assum:etasmallest}), we obtain
                 \begin{equation}
                   \label{eq:couplingslyap}
                   \begin{footnotesize}
                     \begin{aligned}
                       P_i\hat{A}_{ij}=
                       \renewcommand\arraystretch{2}
                       \begin{pmatrix}
                         \frac{\eta_i }{R_{ij}C_{ti}} & 0 & 0& \\
                         0 & 0 & 0 \\
                         0 & 0 & 0\\
                        \end{pmatrix}\approx
                       \begin{pmatrix}
                         0 & 0 & 0 \\
                         0 & 0 & 0 \\
                         0 & 0 & 0\\
                       \end{pmatrix}
                     \end{aligned}
                   \end{footnotesize}
                 \end{equation}
                 and
            \begin{equation}
                   \label{eq:couplingslyap}
                   \begin{footnotesize}
                     \begin{aligned}
                       \hat{A}_{ji}^TP_j=
                       \renewcommand\arraystretch{2}
                       \begin{pmatrix}
                         \frac{\eta_j}{ R_{ji}C_{tj}} & 0 & 0 \\
                         0 & 0 & 0 \\
                         0 & 0 & 0 \\
                       \end{pmatrix}\approx
                       \begin{pmatrix}
                         0 & 0 & 0 \\
                         0 & 0 & 0 \\
                         0 & 0 & 0 \\
                       \end{pmatrix},
                     \end{aligned}
                   \end{footnotesize}
                 \end{equation}
                 which proves that inequality \eqref{eq:Lyap2part} holds.
               \end{proof}
               At this point, in order to complete the design of the local
               controller $\subss{\CC}{i}$, we have to solve the following
               problem.
               \begin{prbl}
                 \label{prbl:designPrbl}
                 Compute a matrix $K_{i}$ such that the nominal closed-loop subsystem is asymptotically stable and Assumption \ref{ass:ctrl}-(\ref{assum:pstruct}) is verified, i.e. \eqref{eq:Lyapeqnith} holds for a matrix $P_i$ structured as in \eqref{eq:pstruct}.
               \end{prbl}
               Consider the following optimization problem
               \begin{subequations}
                 \label{eq:optproblem}
                 \begin{small}
                 \begin{align}
                  \mathcal{O}: &\min_{\substack{Y_{i},G_{i},\gamma_{i},\beta_{i},\delta_{i}}}\quad \alpha_{i1}\gamma_{i}+\alpha_{i2}\beta_{i}+\alpha_{i3}\delta_{i}\nonumber \\
                   \label{eq:Ystruct}&Y_{i}=\left[ \begin{smallmatrix}
                       \eta_i^{-1} & 0 & 0 \\
                       0 & \bullet & \bullet\\
                       0 & \bullet & \bullet\\
                     \end{smallmatrix}\right]>0\\
                   \label{eq:LMIstab}&\begin{bmatrix}
                     Y_{i}\hat{A}_{ii}^{T}+G_{i}^{T}\hat{B}_{i}^{T}+\hat{A}_{ii}Y_{i}+\hat{B}_{i}G_{i} & Y_{i} \\
                     Y_{i} & -\gamma_{i}I
                   \end{bmatrix}\le 0\\
                   \label{eq:Gcostr}&\begin{bmatrix}
                     -\beta_{i}I & G_{i}^T\\
                     G_{i} & -I
                   \end{bmatrix}<0\\
                   \label{eq:Ycostr}&\begin{bmatrix}
                     Y_{i} & I\\
                     I & \delta_{i}I
                   \end{bmatrix}>0\\
                   &\gamma_{i}>0,\quad\beta_{i}>0,\quad\delta_{i}>0
                 \end{align}
                 \end{small}
               \end{subequations} 
               where $\alpha_{i1}$, $\alpha_{i2}$ and $\alpha_{i3}$
               represent positive weights and $\bullet$ are arbitrary
               entries. Since all constraints in \eqref{eq:optproblem}
               are Linear Matrix Inequalities (LMI), the optimization
               problem is convex and can be solved with efficient
               (i.e. polynomial-time) LMI solvers \cite{Boyd1994}. 
               \begin{lem}
                 \label{lem:optProbl}
                 Problem $\mathcal{O}$ is feasible if and only if Problem \ref{prbl:designPrbl} has a solution. Moreover, $K_i$ and $P_i$ in \eqref{eq:Lyapeqnith} are given by $K_{i}=G_{i}Y_{i}^{-1}$, $P_{i}=Y_{i}^{-1}$ and $\norme{K_{i}}{2}<\sqrt{\beta_i}\delta_{i}$.
               \end{lem}
               \begin{proof}
                 Inequality \eqref{eq:Lyapeqnith} is equivalent to the existence of $\gamma_i>0$ such that
                 \begin{equation}
                   \label{eq:Lyapdecr}
                   (\hat{A}_{ii}+\hat{B}_{i}K_{i})^{T}P_{i}+P_{i}(\hat{A}_{ii}+\hat{B}_{i}K_{i})+\gamma_{i}^{-1}I\le 0
                 \end{equation} 
                 where $P_{i}$ is defined in \eqref{eq:pstruct}. By
                 applying the Schur lemma on \eqref{eq:Lyapdecr}, we
                 get the following inequality
                 \begin{equation}
                   \label{eq:LyapSchur}
                   \begin{bmatrix}
                     (\hat{A}_{ii}+\hat{B}_{i}K_{i})^{T}P_{i}+P_{i}(\hat{A}_{ii}+\hat{B}_{i}K_{i}) & I\\
                     I & -\gamma_{i}I
                   \end{bmatrix}\le 0
                 \end{equation}
                 which is nonlinear in $P_{i}$ and $K_{i}$. In order
                 to get rid of the nonlinear terms, we perform the
                 following parametrization trick \cite{Boyd1994}
                 \begin{equation}
                   \label{eq:eqschur}
                   \begin{aligned}
                     Y_{i}&=P_{i}^{-1}\\
                     G_{i}&=K_{i}Y_{i}.
                   \end{aligned}
                 \end{equation}
                 Notice that the structure of $Y_{i}$ is the same as
                 the structure of $P_{i}$. By pre- and post-multiplying \eqref{eq:LyapSchur} with $\left[\begin{smallmatrix}
                     Y_{i} & 0\\
                     0 & I 
                   \end{smallmatrix}\right]$ and exploiting \eqref{eq:eqschur} we obtain
                 \begin{equation}
                   \begin{bmatrix}
                     Y_{i}\hat{A}_{ii}^{T}+G_{i}^{T}\hat{B}_{i}^{T}+\hat{A}_{ii}Y_{i}+\hat{B}_{i}G_{i} & Y_{i} \\
                     Y_{i} & -\gamma_{i}I
                   \end{bmatrix}\le 0
                 \end{equation}
                 
                 Constraint
                 \eqref{eq:Ystruct} ensures that matrix $P_i$ has
                 the structure required by Assumption
                 \ref{ass:ctrl}-(\ref{assum:pstruct}). At the same time,
                 constraint \eqref{eq:LMIstab} guarantees stability of
                 the closed-loop subsystem. Further constraints appear
                 in Problem $\mathcal{O}$  with the aim of bounding
                 $\norme{K_{i}}{2}$. In particular, we add
                 $\norme{G_{i}}{2}<\sqrt{\beta_{i}}$ and
                 $\norme{Y_{i}^{-1}}{2}<\delta_{i}$ (which via Schur
                 complement, correspond to constraints
                 \eqref{eq:Gcostr} and \eqref{eq:Ycostr}) to prevent
                 $\norme{K_{i}}{2}$ from becoming too large. These bounds imply $\norme{K_{i}}{2}<\sqrt{\beta_i}\delta_{i}$ and then affect the magnitude of control variables.
               \end{proof}
               Next, we discuss the key feature of the proposed
               decentralized control approach. We first notice that
               constraints in
               \eqref{eq:optproblem} depend upon local fixed matrices
               ($\hat{A}_{ii}$, $\hat{B}_{i}$) and local design
               parameters ($\alpha_{i1}$, $\alpha_{i2}$,
               $\alpha_{i3}$, $\beta_i$, $\delta_i$). It follows that the computation of controller
               $\subss{\CC}{i}$ is completely independent from the
               computation of controllers $\subss{\CC}{j}$ when $j\neq i$
               since, provided that problem $\PP_i$ is feasible,
               controller $\subss\CC i$ can be directly obtained
               through $K_{i}=G_{i}Y_{i}^{-1}$. In addition, it is
               clear that constraints \eqref{eq:Gcostr} and
               \eqref{eq:Ycostr} affect only the magnitude of control
               variables as stated in Lemma
               \ref{lem:optProbl}. Finally, since all assumptions in Proposition \ref{prop:ctrldec} are verified, the overall closed-loop QSL-ImG is asymptotically stable.
\subsection{Enhancements of local controllers for improving performances}
               In the previous section we have shown how to design
               decentralized controllers $\subss{\CC}{i}$ guaranteeing
               asymptotic stability for the overall closed-loop system
               \eqref{eq:sysaugoverallclosed}. In order to improve transient
               performances of controllers $\subss{\CC}{i}$, we
               enhance them with feed-forward terms for
               \begin{enumerate}[(i)]
               \item pre-filtering reference signals;
               \item compensating measurable disturbances.
               \end{enumerate}

               \subsubsection{Pre-filtering of the reference signal}
                    \label{sec:pre-filter}
                    Pre-filtering is well known technique used to
                    widen the bandwidth so as to speed up the response
                    of the system. Consider the transfer function
                    $\subss{F}{i}(s)$, from $\subss{z_{ref}}{i}(t)$ to
                    the controlled variable  $\subss{z}{i}(t)$ 
                    \begin{equation}
                      \label{eq:F(s)}
                      \subss{F}{i}(s)=(\hat{H}_i \hat{C}_i) (sI-(\hat{A}_{ii}+\hat{B}_{i}K_i ))^{-1}\begin{bmatrix}
                        0 \\
                        1
                      \end{bmatrix} 
                    \end{equation}
of each nominal closed-loop subsystem
\eqref{eq:modelDGUgen-aug-closed}. By virtue of a feedforward
compensator $\subss{\tilde{C}}{i}(s)$, it is possible to filter the reference signal $\subss{z_{ref}}{i}(t)$ (see Figure \ref{fig:prefilter}). 
                    \begin{figure}[!htb]
                      \centering
                      \begin{tikzpicture}
  \begin{Large}
  \sbEntree{zrefi}
  \sbBloc{compensatori}{$\subss{\tilde{C}}{i}(s)$}{zrefi}
  \sbBloc{Fi}{$\subss{F}{i}(s)$}{compensatori}
  \sbSortie{zi}{Fi}
  \sbRelier{zrefi}{compensatori}
  \sbRelier{Fi}{zi}
  \end{Large}
  \sbNomLien{zrefi}{$\subss{z_{ref}} i$}
  \sbNomLien{zi}{$\subss z i$}
  \sbRelier[$\subss{z_{ref}^f} i$]{compensatori}{Fi}
\end{tikzpicture}
                      \caption{Block diagram of closed-loop DGU $i$ with prefilter.}
                      \label{fig:prefilter}
                    \end{figure}
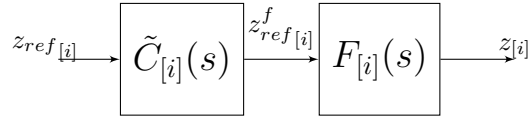
                    Consequently, the new transfer function from $\subss{z_{ref}}{i}(t)$ to $\subss{z}{i}(t)$ becomes
                    \begin{equation}
                      \label{eq:C(s)F(s}
                      \subss{\tilde{F}}{i}(s)=\subss{\tilde{C}}{i}(s)\subss{F}{i}(s)
                    \end{equation}
                    Now, taking a desired transfer function
                    $\subss{\tilde{F}}{i}(s)$ for each subsystem, we
                    can compute, from \eqref{eq:C(s)F(s}, the
                    pre-filter $\subss{\tilde{C}}{i}(s)$ as 
                    \begin{equation}
                      \label{eq:prefilter}
                      \subss{\tilde{C}}{i}(s) = \subss{\tilde{F}}{i}(s)\subss{{F}}{i}(s)^{-1}
                    \end{equation}
                    under the following conditions \cite{Skogestad1996}:
                    \begin{itemize}
                    \item $\subss{{F}}{i}(s)$ must not have Right-Half-Plane (RHP) zeros that would become RHP poles of $\subss{\tilde{C}}{i}(s)$, making it unstable;
                    \item $\subss{{F}}{i}(s)$ must not contain a time delay, otherwise $\subss{\tilde{C}}{i}(s)$ would have a predictive action
                    \item $\subss{\tilde{C}}{i}(s)$ must be realizable, i.e. it must have more poles than zeros.
                    \end{itemize}
                    Hence, if these conditions are fulfilled, the
                    filter $\subss{\tilde{C}}{i}(s)$ given by
                    \eqref{eq:prefilter} is realizable and
                    asymptotically stable (this condition is essential
                    since $\subss{\tilde{C}}{i}(s)$ works in
                    open-loop). Furthermore, since
                    $\subss{\hat{F}}{i}(s)$ is asymptotically stable
                    (controllers
                    $\subss{\CC}{i}$ are, in fact, designed solving the problem
                    $\PP_i$), the closed-loop system including filters
                    $\subss{\tilde{C}}{i}(s)$ is asymptotically stable
                    as well. He highlight that, if some of the
                    previous conditions are not valid, expression \eqref{eq:prefilter} cannot be used. Still, the compensator $\subss{\tilde{C}}{i}(s)$ can be designed for a given bandwidth, as shown in \cite{Skogestad1996}. 
     
 \subsubsection{Compensation of measurable disturbances}
	            \label{sec:compensator}
                    The second enhancement one can introduce regards
                    the compensation of measurable disturbances. We
                    remind that, since we assumed that load dynamics
                    is not known, we have modeled the load
                    currents for each subsystem of the microgrid as a
                    measurable disturbance $\subss{d}{i}(t)$. Let us
                    define new local controllers
                    $\subss{\tilde{\CC}}{i}$ as
                    \begin{equation}
                      \subss{\tilde{\CC}}{i}:\quad \subss{u}{i}=K_{i}\subss{\hat{x}}{i}(t)+\subss{\tilde{u}}{i}(t)
                    \end{equation}
                    Note that $\subss{\tilde{\CC}}{i}$ are obtained by adding term
                    $\subss{\tilde{u}}{i}(t)$ to the controllers $\subss{\CC}{i}$ in
                    \eqref{eq:ctrldec}. Hence,
                    \eqref{eq:modelDGUgen-aug-closed} can be rewritten
                    as follows
                    \begin{equation}
                      \label{eq:subsysDGi-thAUGCLComp}
                      \subss{\tilde{\Sigma}}{i}^{DGU} :
                      \left\lbrace
                        \begin{aligned}
                          \subss{\dot{\hat{x}}}{i}(t) &= (\hat{A}_{ii}+ \hat{B}_{i}K_{i})\subss{\hat{x}}{i}(t)+\hat{M}_{i}\subss{\hat{d}}{i}(t)+\hat{B_{i}}\subss{\tilde{u}}{i}(t)\\
                          \subss{\hat{y}}{i}(t)       &= \hat{C}_{i}\subss{\hat{x}}{i}(t)\\
                          \subss{z}{i}(t)       &= \hat{H}_{i}\subss{\hat{y}}{i}(t)
                        \end{aligned}
                      \right..
                    \end{equation}
                    We now use the new input $\subss{\tilde{u}}{i}(t)$ to compensate the measurable disturbance $\subss{d}{i}(t)$ (recall that $\subss{\hat d}{i} = [\subss{d^T}{i}~\subss{z_{ref}^T}{i} ]^T$). From \eqref{eq:subsysDGi-thAUGCLComp}, the transfer function from the disturbance $\subss{d}{i}(t)$ to the controlled variable $\subss{z}{i}(t)$ is
                    \begin{equation}
                      \label{eq:Gd(s)}
                      G^{d}_{i}(s)=(\hat{H}_i \hat{C}_i)(sI-(\hat{A}_{ii}+\hat{B}_{i}K_i))^{-1}			\begin{bmatrix}
                        M_{i}\\
                        0
                      \end{bmatrix}. 
                    \end{equation}
                    Moreover, the transfer function from the new input $\subss{\tilde{u}}{i}(t)$ to the controlled variable $\subss{z}{i}(t)$ is
                    \begin{equation}
                      \label{eq:G(s)}
                      G_{i}(s)=(\hat{H}_i \hat{C}_i)(sI-(\hat{A}_{ii}+\hat{B}_{i}K_i ))^{-1}\hat{B}_i . 
                    \end{equation}
                    If we combine \eqref{eq:Gd(s)} and \eqref{eq:G(s)}, we obtain
                    \begin{equation}
                      \subss{z}{i}(s)=G_{i}(s)\subss{\tilde{u}}{i}(s)+G^{d}_{i}(s)\subss{d}{i}(s).
                    \end{equation}
                    In order to zero the effect of the disturbance on the controlled variable, we set
                    \begin{equation}
                      \subss{\tilde{u}}{i}(s)=N_{i}(s)\subss{d}{i}(s)
                    \end{equation}
                    where 
                    \begin{equation}
                      \label{eq:compensator}
                      \subss{N}{i}(s)=-G_{i}(s)^{-1}G^{d}_{i}(s)
                    \end{equation}
                    is the transfer function of the compensator. Note that $\subss{N}{i}(s)$ is well defined under the following conditions \cite{Skogestad1996}:
                    \begin{itemize}
                    \item $\subss{G}{i}(s)$ must not have RHP zeros that would become RHP poles of $\subss{N}{i}(s)$;
                    \item $\subss{G}{i}(s)$ must not contain a time delay, otherwise $\subss{N}{i}(s)$ would have a predictive action
                    \item $\subss{N}{i}(s)$ must be realizable, i.e. it must have more poles than zeros.
                    \end{itemize}
                    In this way, we can ensure that the compensator $\subss{N}{i}(s)$ is asymptotically stable, hence preserving asymptotic stability of the system. When some of the previous conditions do not hold, formula \eqref{eq:compensator} cannot be used and perfect compensation cannot be achieved. Still, the compensator $\subss{N}{i}(s)$ can be designed to reject disturbances within a given bandwidth, as shown in \cite{Skogestad1996}. The overall control scheme with the addition of the compensators is shown in Figure \ref{fig:compensator}.
                    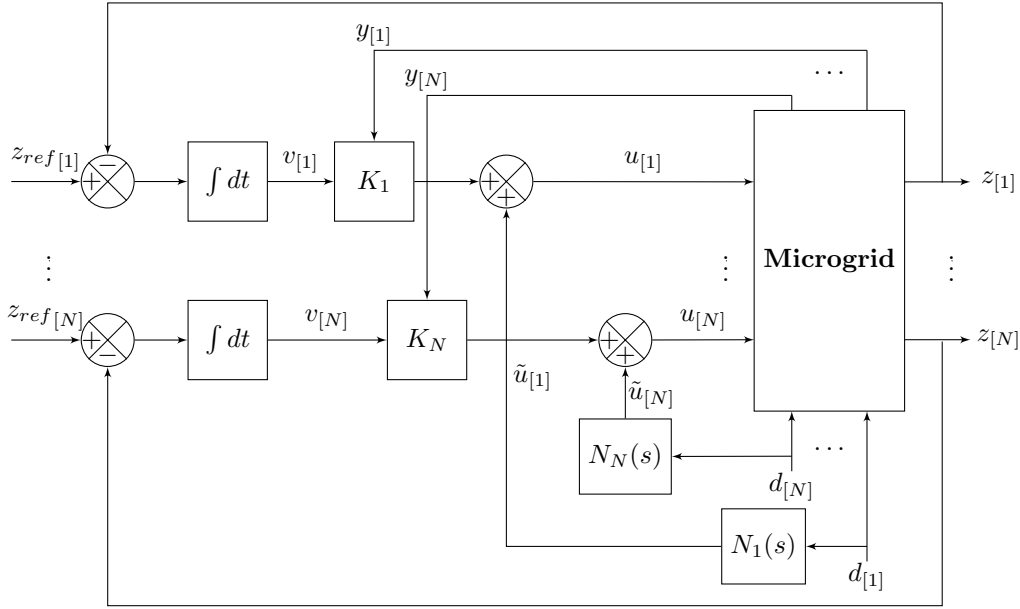
\begin{figure}[!htb]
                      \centering
                      \tikzstyle{input} = [coordinate]
\tikzstyle{output} = [coordinate]
\tikzstyle{guide} = []
\tikzstyle{block} = [draw, rectangle, minimum height=1cm]

\begin{tikzpicture}
  \sbEntree{zref1}
  \sbDecaleNoeudy[3]{zref1}{zrefj}
  \sbDecaleNoeudy[3]{zrefj}{zrefN}
  \node [block, right of=zrefj,node distance=11cm,minimum height=4cm, minimum width=2cm] (microgrid) {\textbf{Microgrid}};
  \node [guide, right of=zrefj,xshift=-0.5cm] (zrefjline) {};
  \draw [draw] (zrefjline) -| node{$\vdots$} (zrefjline);
  
  \sbComph{sumret1}{zref1}
  \sbBloc{integrator1}{$\int dt$}{sumret1}
  \sbBloc[2.5]{controller1}{$K_1$}{integrator1}
  \sbSumb[5]{sumdist1}{controller1}
  \sbRelier[$\subss{z_{ref}}{1}$]{zref1}{sumret1}
  \sbRelier{sumret1}{integrator1}
  \sbRelier[$\subss v 1$]{integrator1}{controller1}
  \sbRelier{controller1}{sumdist1}
  \node [guide, left of=microgrid,yshift=1.05cm,xshift=0.125cm] (u1end) {};
  \sbRelier[$\subss u 1$]{sumdist1}{u1end}
  
  \sbComp{sumretN}{zrefN}
  \sbBloc{integratorN}{$\int dt$}{sumretN}
  \sbBloc[4.5]{controllerN}{$K_N$}{integratorN}
  \sbSumb[7.5]{sumdistN}{controllerN}
  \sbRelier[$\subss{z_{ref}}{N}$]{zrefN}{sumretN}
  \sbRelier{sumretN}{integratorN}
  \sbRelier[$\subss v N$]{integratorN}{controllerN}
  \sbRelier{controllerN}{sumdistN}
  \node [guide, left of=microgrid,yshift=-1.05cm,xshift=0.125cm] (uNend) {};
  \sbRelier[$\subss u N$]{sumdistN}{uNend}
  
  \node [output, below of=microgrid,yshift=-3.0cm,xshift=0.5cm] (d1out) {};
  \node [output, below of=microgrid,yshift=-2.8cm,xshift=0.5cm] (d1outnear) {};
  \node [output, below of=microgrid,yshift=-1.0cm,xshift=0.5cm] (d1outstart) {};
  \draw [draw,->,>=latex'] (d1out) -| node[yshift=-0.2cm]{$\subss d 1$} (d1outstart);
  \node [block, below of=uNend,yshift=-1.75cm,xshift=0.0cm] (N1) {$N_1(s)$};
  \draw [draw,->,>=latex',near end,swap] (d1outnear) -- (N1);
  \draw [draw,->,>=latex',near end,swap] (N1) -| node[xshift=0.35cm] {$\subss{\tilde u} 1$} (sumdist1);
  
  \node [output, below of=microgrid,yshift=-1.5cm] (djout) {};
  \draw [draw] (djout) -| node {$\ldots$} (djout);
  
  \node [output, below of=microgrid,yshift=-1.8cm,xshift=-0.5cm] (dNout) {};
  \node [output, below of=microgrid,yshift=-1.6cm,xshift=-0.5cm] (dNoutnear) {};
  \node [output, below of=microgrid,yshift=-1.0cm,xshift=-0.5cm] (dNoutstart) {};
  \draw [draw,->,>=latex'] (dNout) -| node[yshift=-0.2cm]{$\subss d N$} (dNoutstart);
  \node [block, below of=sumdistN,yshift=-0.55cm,xshift=0.0cm] (NN) {$N_N(s)$};
  \draw [draw,->,>=latex',near end,swap] (dNoutnear) -- (NN);
  \draw [draw,->,>=latex',near end,swap] (NN) -- node[xshift=0.35cm,yshift=-0.2cm] {$\subss{\tilde u} N$} (sumdistN);

  \node [output, above of=microgrid,yshift=1.8cm,xshift=0.5cm] (y1out) {};
  \node [output, above of=microgrid,yshift=1.0cm,xshift=0.5cm] (y1outstart) {};
  \node [output, above of=microgrid,yshift=1.5cm] (yjout) {};
  \node [output, above of=microgrid,yshift=1.2cm,xshift=-0.5cm] (yNout) {};
  \node [output, above of=microgrid,yshift=1.0cm,xshift=-0.5cm] (yNoutstart) {};
  \draw [draw] (y1outstart) -- node {} (y1out);
  \draw [draw,->,>=latex'] (y1out) -| node[yshift=0.2cm]{$\subss y 1$} (controller1);
  \draw [draw] (yjout) -| node {$\ldots$} (yjout);
  \draw [draw] (yNoutstart) -- node {} (yNout);
  \draw [draw,->,>=latex'] (yNout) -| node[yshift=0.2cm]{$\subss y N$} (controllerN);
  
  \node [guide, right of=microgrid,yshift=1.05cm,xshift=-.125cm] (z1) {};
  \node [guide, right of=microgrid,yshift=1.165cm,xshift=0.5cm] (z1near) {};
  \node [guide, right of=microgrid,yshift=1.05cm,xshift=1cm] (z1end) {};
  \draw [draw,->,>=latex',near end,swap] (z1) -- node[xshift=0.6cm] {$\subss{z} 1$} (z1end);
  \sbRenvoi[-6.8]{z1near}{sumret1}{}

  \node [guide, right of=microgrid,yshift=-1.05cm,xshift=-.125cm] (zN) {};
  \node [guide, right of=microgrid,yshift=-0.95cm,xshift=0.5cm] (zNnear) {};
  \node [guide, right of=microgrid,yshift=-1.05cm,xshift=1cm] (zNend) {};
  \draw [draw,->,>=latex',near end,swap] (zN) -- node[xshift=0.6cm] {$\subss{z} N$} (zNend);
  \sbRenvoi[10]{zNnear}{sumretN}{}
  
  \node [guide, right of=microgrid,xshift=0.5cm] (zj) {};
  \draw [draw] (zj) -| node {$\vdots$} (zj);
  
  \node [guide, left of=microgrid,xshift=-0.5cm] (uj) {};
  \draw [draw] (uj) -| node {$\vdots$} (uj);
  
\end{tikzpicture}
                      \caption{Overall microgrid control scheme with compensation of measurable disturbances $\subss{d}{i}(s)$.}
                      \label{fig:compensator}
                    \end{figure}

	  \subsection{Algorithm for the design of local controllers}
Algorithm
               \ref{alg:ctrl_design} collects the steps of the overall design procedure.
               \begin{algorithm}[!htb]
                 \caption{Design of controller $\subss{\CC}{i}$ and compensators $\subss{\tilde{C}}{i}$ and $\subss{N}{i}$ for subsystem $\subss{\hat{\Sigma}}{i}^{DGU}$}
                 \label{alg:ctrl_design}
                 \textbf{Input:} DGU $\subss{\hat{\Sigma}}{i}^{DGU}$ as in \eqref{eq:modelDGUgen-aug} \\
                 \textbf{Output:} Controller $\subss{\CC}{i}$ and, optionally, pre-filter $\subss{\tilde{C}}{i}$ and compensator $\subss{N}{i}$\\
                 \begin{enumerate}[(A)]
                 \item\label{enu:stepAalgCtrl} Find $K_i$ solving the LMI problem \eqref{eq:optproblem}. If it is not feasible \textbf{stop} (the controller $\subss\CC i$ cannot be designed).\\
                   \textbf{Optional steps}
                 \item\label{enu:stepBalgCtrl} Design the
                   asymptotically stable local pre-filter
                   $\subss{\tilde{C}}{i}$ and compensator
                   $\subss{N}{i}$ as in (\ref{eq:compensator}).
                 \end{enumerate}
               \end{algorithm}

          \subsection{PnP operations}
               \label{sec:PnP}
              In the following section, the operations for updating the controllers when DGUs are
              added to or removed from an ImG are presented. We remind
              that all these operations are performed with the
              aim of preserving stability of the new closed-loop
              system. Consider, as a starting point, a microgrid composed of subsystems
              $\subss{\hat{\Sigma}}{i}^{DGU}, i\in\DD$ equipped with
              local controllers $\subss{\CC}{i}$ and compensators
              $\subss{\tilde{C}}{i}$ and $\subss{N}{i}$, $i\in\DD$
              produced by Algorithm \ref{alg:ctrl_design}.
\begin{rmk}
\label{rmk:bumpless}
In order to avoid jumps in the control variable when local regulator
are switched, we embedded each local regulator into a bumpless control
scheme \cite{aastrom2006advanced} shown in Appendix
\ref{sec:AppBumpless}.
\end{rmk}

\textbf{Plugging-in operation}
                    Assume that the plug-in of a new DGU
                    $\subss{\hat{\Sigma}}{N+1}^{DGU}$ described by
                    matrices, $\hat{A}_{N+1\:N+1}$, $\hat{B}_{N+1}$,
                    $\hat{C}_{N+1}$, $\hat{M}_{N+1}$, $\hat{H}_{N+1}$
                    and $\{\hat{A}_{N+1\:j}\}_{j\in\NN_{N+1}}$ needs
                    to be performed. Let $\NN_{N+1}$ be the
                    set of DGUs that are directly coupled to
                    $\subss{\hat{\Sigma}}{N+1}^{DGU}$ through
                    transmission lines and
                    let $\{\hat{A}_{N+1\:j}\}_{j\in\NN_{N+1}}$ be the
                    matrices containing the corresponding coupling
                    terms. According to our method, the
                    design of controller $\subss{\CC}{N+1}$ and compensators
                    $\subss{\tilde{C}}{N+1}$ and $\subss{N}{N+i}$
                    requires Algorithm \ref{alg:ctrl_design} to be
                    executed. Since DGUs
                    $\subss{\hat{\Sigma}}{j}^{DGU}$, $j\in\NN_{N+1}$,
                    have the new neighbour
                    $\subss{\hat{\Sigma}}{N+1}^{DGU}$, we need to
                    redesign controllers $\subss{\CC}{j}$ and
                    compensators $\subss{\tilde{C}}{j}$ and
                    $\subss{N}{j}$, $\forall j\in\NN_{N+1}$ because
                    matrices $\hat{A}_{jj}$, $j\in\NN_{N+1}$ change. 

Only if Algorithm \ref{alg:ctrl_design} does not stop in Step
\ref{enu:stepAalgCtrl} when computing controllers $\subss{\CC}{k}$ for
all $k\in\NN_{N+1}\cup\{N+1\}$, we have that the plug-in of
$\subss{\hat{\Sigma}}{N+1}^{DGU}$ is allowed. Moreover, we stress that the redesign is not propagated further in the
network and therefore the asymptotic stability of the new overall closed-loop QSL-ImG model is preserved even without changing controllers $\subss{\CC}{i}$, $\subss{\tilde{C}}{i}$ and $\subss{N}{i}$, $i\not\in \{N+1\}\cup\NN_{N+1}$.

Prior to real-time plugging-in
               operation (\textit{hot plugging-in}), it is recommended
               to keep set points constant for a sufficient amount of
               time so as to guarantee control variable in the
               bumpless control scheme (see Remark \ref{rmk:bumpless})
               is in steady state. This ensures smooth behaviours of
               the electrical variables.

               \textbf{Unplugging operation}
                    Let us now examine the unplugging of DGU
                    $\subss{\hat{\Sigma}}{k}^{DGU}$, $k\in\DD$. The
                    disconnection of $\subss{\hat{\Sigma}}{k}^{DGU}$
                    from the network leads to a change in matrix
                    $\hat{A}_{jj}$ of each
                    $\subss{\hat{\Sigma}}{j}^{DGU}$, $
                    j\in\NN_{k}$. Consequently, for each $
                    j\in\NN_{k}$, we have to redesign
                    controllers $\subss{\CC}{j}$ and compensators
                    $\subss{\tilde{C}}{j}$ and $\subss{N}{j}$. As for the plug-in operation, we run
                    Algorithm \ref{alg:ctrl_design}. If all operations can be
                    successfully terminated, then 
                    the unplugging of $\subss{\hat{\Sigma}}{k}^{DGU}$
                    is allowed and stability is preserved without redesigning the local controllers $\subss\CC j$, $j\notin\NN_{k}$.

                    When an unplugging operation is scheduled in advance, it is advisable to follow an \textit{hot unplugging} protocol similar to the one introduced for the plugging-in operation.

     \section{Simulation results}
          \label{sec:Simresults}
          In this section, we study performance brought about by PnP
          controllers described in Section \ref{sec:PnPctrl}. As a
          starting point, we consider the ImG depicted in Figure
          \ref{fig:schemairandist_DC} with only two DGUs (Scenario 1)
          and we evaluate performance in terms of tracking step references as
          well as hot plugging-in of the two DGUs and robustness to unknown load dynamic. Then, we
          extend the analysis to an ImG with 6 DGUs (Scenario 2) and we show that stability of the whole microgrid is guaranteed.

         Simulations have been performed in PSCAD, a simulation environment for
         electric systems which allows to implement the microgrid model with realistic electric components.

          For both scenarios, we run a simulation from time 0 s up to time 10 s. Each simulation has been split into subparts that are discussed next.

          \subsection{Scenario 1}
               \label{sec:scenario_1}
               In this Scenario, we consider the ImG shown in Figure
               \ref{fig:schemairandist_DC} composed of two DC DGUs
               connected through high resistive-inductive lines
               supporting 10 $\Omega$ and 6 $\Omega$ loads,
               respectively. For the sake of simplicity, we set $i=1$
               and $j=2$. The output voltage reference $v^\star_{MG}$
               has been selected at $48\mbox{ }V$ and it is equal for both
               DGUs. Parameters values for all DGUs are given in
               Table \ref{tbl:electrical_setup} in Appendix
               \ref{sec:AppElectrPar}. Notice that that they are comparable to those used in \cite{shafiee2014modeling} and \cite{shafiee2014hierarchical}.

               \subsubsection{Voltage reference tracking at the startup}
We assume that at the beginning of the simulation ($t = 0$ s),
subsystems $\subss{\hat{\Sigma}}{1}^{DGU}$ and
$\subss{\hat{\Sigma}}{2}^{DGU}$ are not interconnected. Therefore,
stabilizing controllers $\mathcal{C}_{i}$, $i={1,2}$ are designed
neglecting coupling among DGUs. Moreover, in order to widen the
bandwidth of each closed-loop subsystem, we use local pre-filters $\subss{\tilde{C}}{i}$, $i = 1,2$of reference signals. The desired closed-loop transfer functions $\tilde F_{i}(s)$, $i = {1, 2}$ have been chosen as low-pass filters with DC gain equal to $0$ dB and bandwidth equal to 100 Hz. 
The eigenvalues of the two decoupled closed-loop QSL subsystems are
shown in Figure \ref{fig:closedLoop2Areas_0}. Moreover, by running Step B of
Algorithm \ref{alg:ctrl_design} we obtain two asymptotically stable
local pre-filters $\tilde
C_{i}$, $i = {1,2}$ whose Bode magnitude plots are depicted in Figure
\ref{fig:singularValuePreFilter2Areas_0}. Notice that through the addition of the
pre-filters, the frequency response of the two closed-loop transfer
functions $F_i(s)$, $i = 1,2$ coincide with the frequency response of the
desired transfer functions $\tilde F_{i}(s)$, $i = {1, 2}$ (see the green line in Figure \ref{fig:singularValueClosedLoop2Areas_0}).
\begin{figure}[!htb]
                 \centering
                 \begin{subfigure}[!htb]{0.48\textwidth}
                   \centering
                   \includegraphics[width=1.1\textwidth, height=150pt]{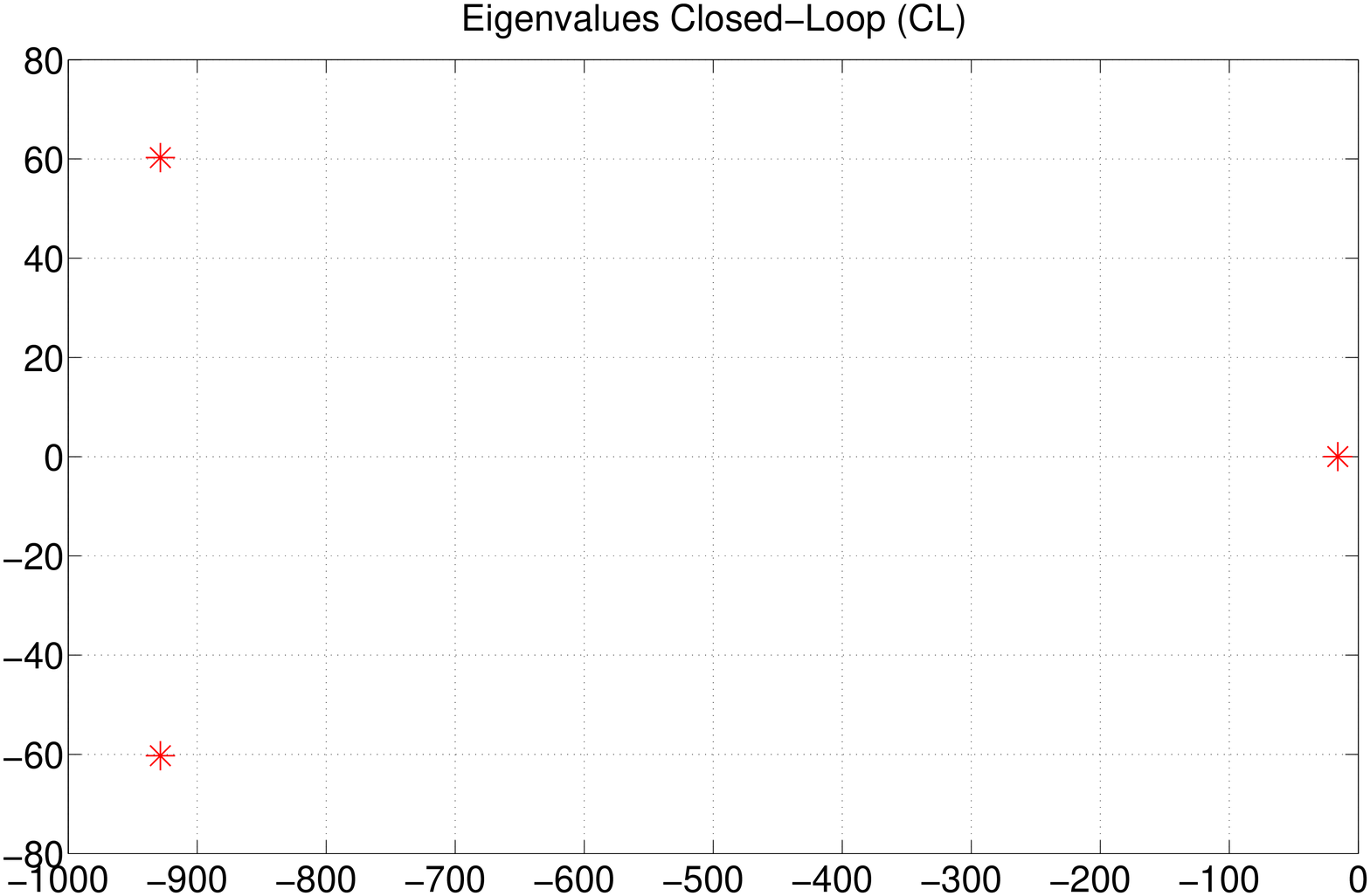}
                   \caption{Eigenvalues of the two decoupled
                     closed-loop QSL subsystems.}
                   \label{fig:closedLoop2Areas_0}
                 \end{subfigure}
                 \begin{subfigure}[!htb]{0.48\textwidth}
                   \centering
                   \includegraphics[width=1.1\textwidth, height=150pt]{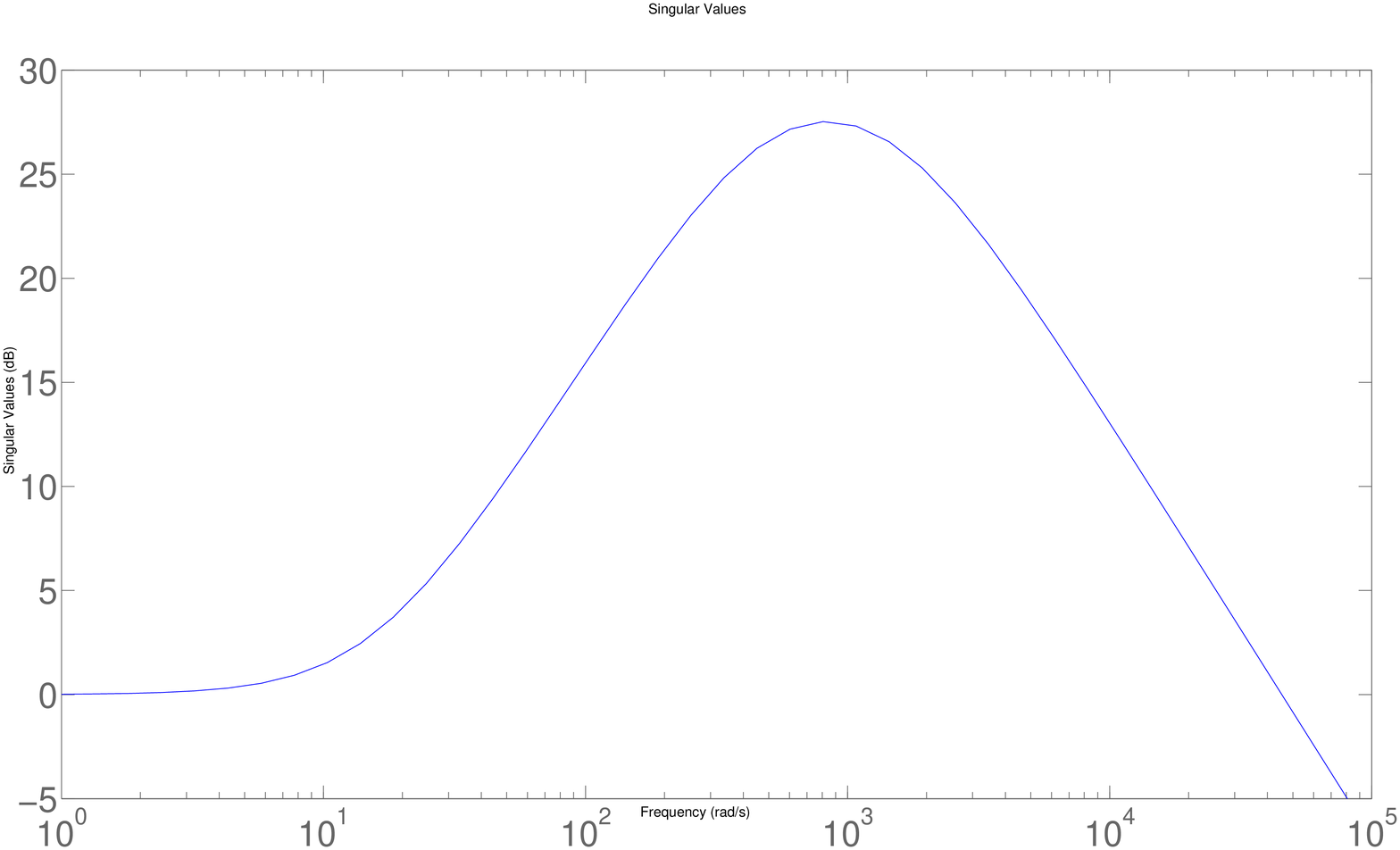}
                   \caption{Bode magnitude plot of pre-filters $\subss{\tilde{C}}{i}$, $i = 1,2$.}
                   \label{fig:singularValuePreFilter2Areas_0}
                 \end{subfigure}
                 \begin{subfigure}[!htb]{0.48\textwidth}
                   \centering
                   \includegraphics[width=1.1\textwidth, height=150pt]{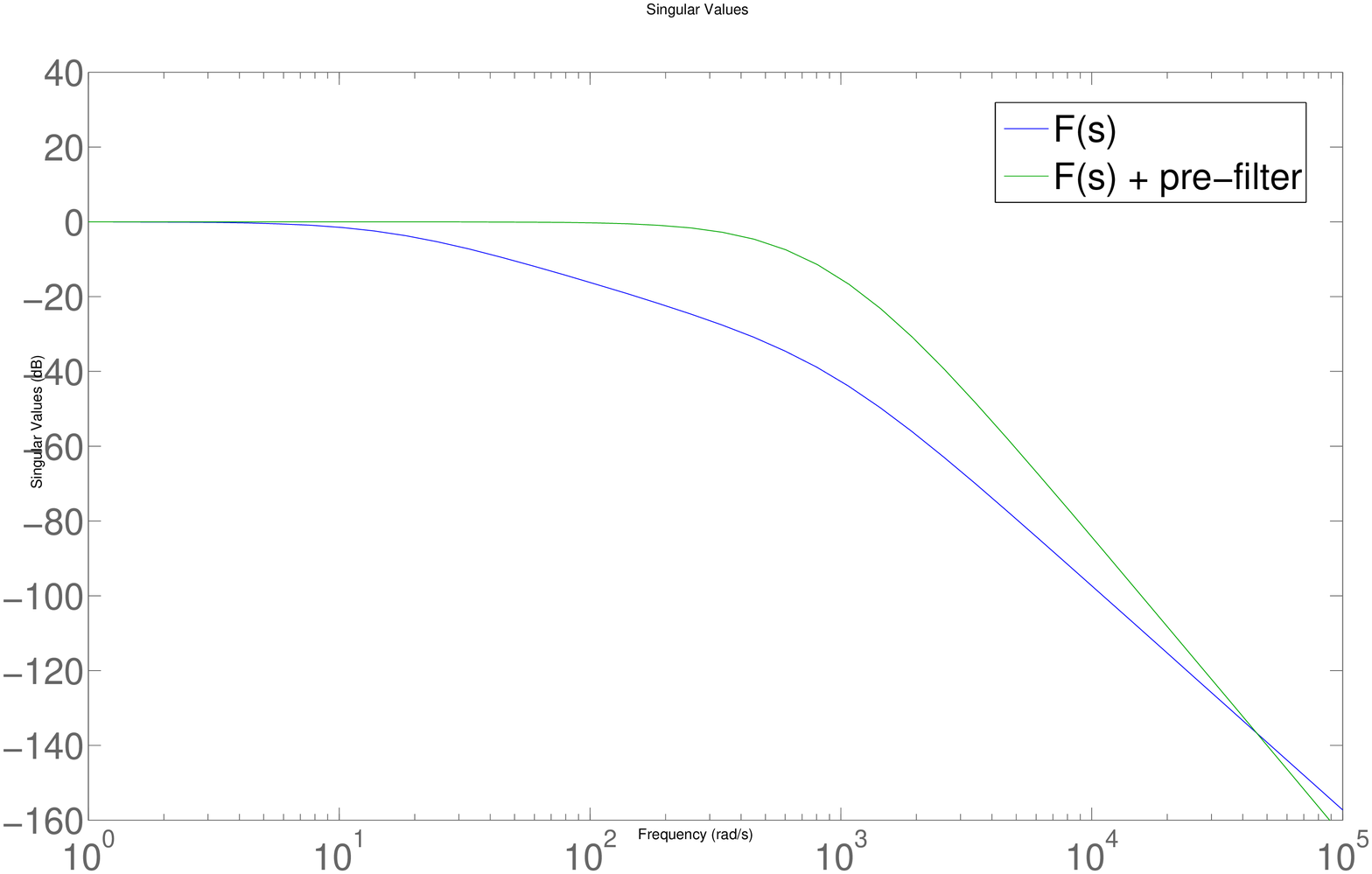}
                   \caption{Bode magnitude plot of $F_i(s)$, $i = 1,2$ with (green) and without (blue) pre-filters.}
                   \label{fig:singularValueClosedLoop2Areas_0}
                 \end{subfigure}
                 \caption{Features of PnP controllers for Scenario 1
                   when the DGUs are not interconnected.}
                 \label{fig:closedLoop2areas_0}
               \end{figure}
Figures \ref{fig:startupDGU1} and \ref{fig:startupDGU2} show the voltages at $PCC_{1}$ and $PCC_{2}$. Note that the controllers ensure an excellent tracking of the reference signals at the startup in a very short time.
    
  \begin{figure}[!htb]
                      \centering
                      \begin{subfigure}[!htb]{0.48\textwidth}
                        \centering
                        \includegraphics[width=1\textwidth, height=130pt]{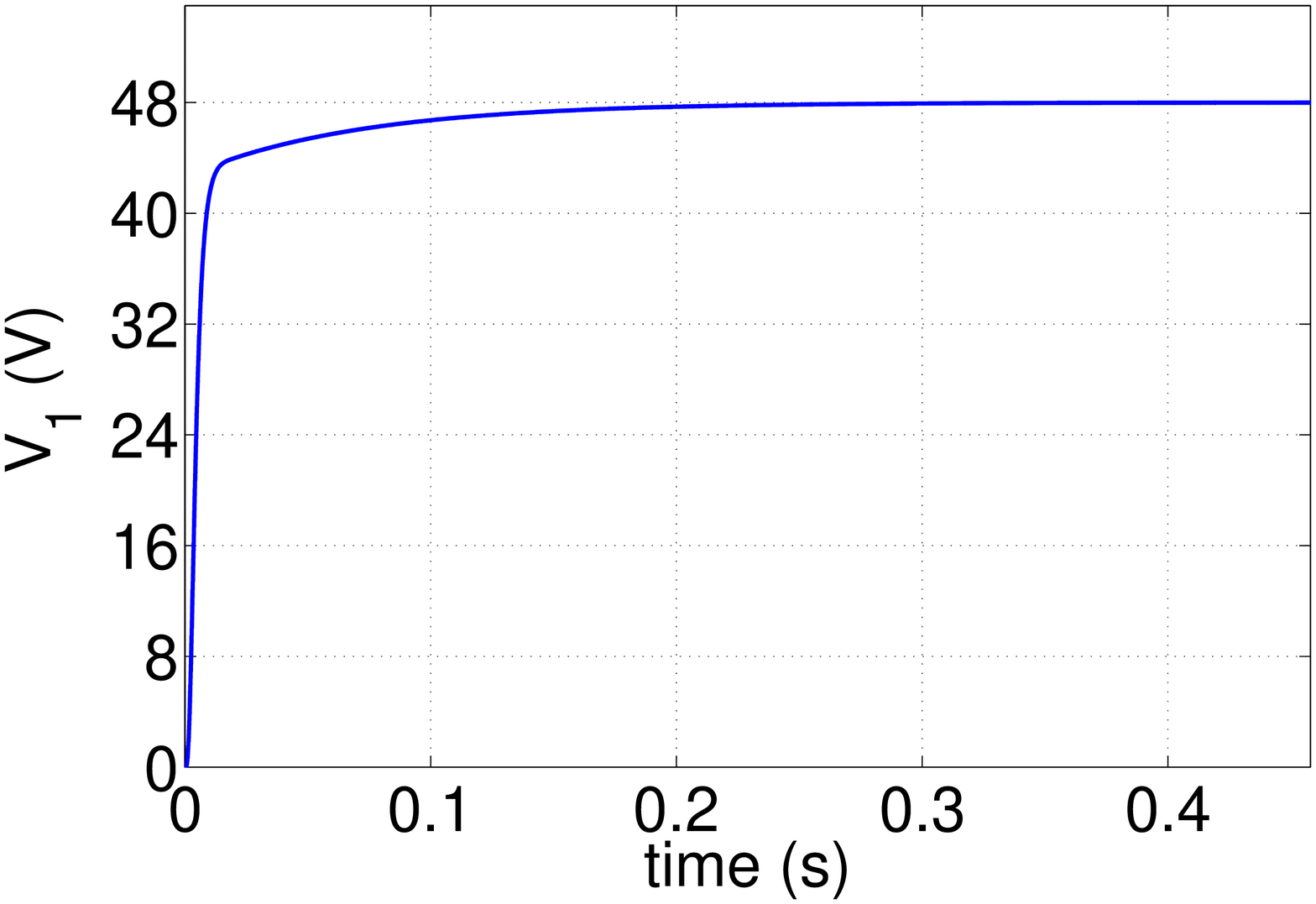}
                        \caption{Load voltage at $PCC_1$.}
                        \label{fig:startupDGU1}
                      \end{subfigure}
                      \begin{subfigure}[!htb]{0.48\textwidth}
                        \centering
                        \includegraphics[width=1\textwidth, height=130pt]{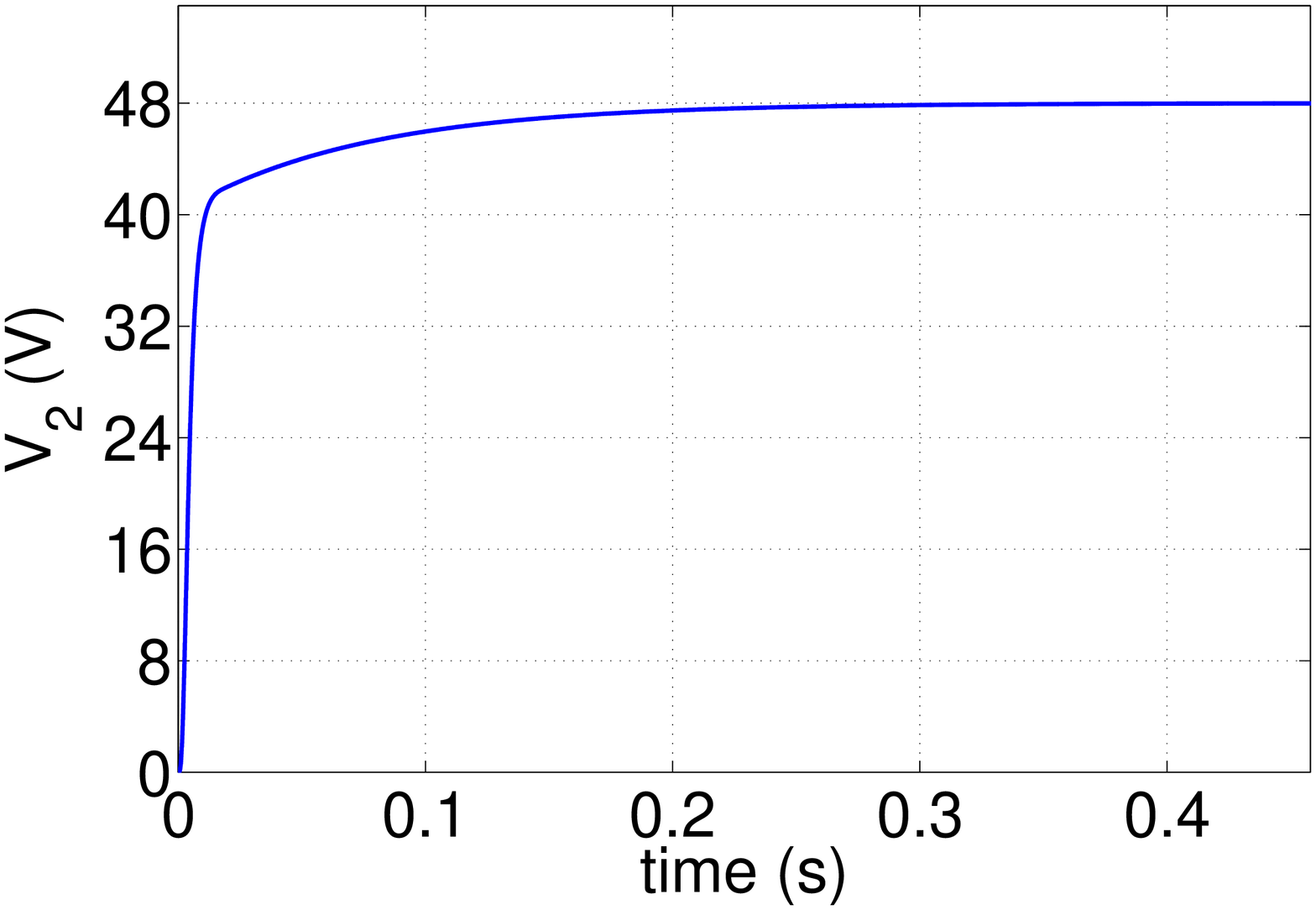}
                        \caption{Load voltage at $PCC_2$.}
                        \label{fig:startupDGU2}
                      \end{subfigure}
                        \caption{Scenario 1 - Voltage reference tracking at the startup.}
                      \label{fig:startupDGU}

                    \end{figure}
               \subsubsection{Hot plugging-in of DGUs 1 and 2}
                    \label{sec:scenario1voltagetrack1}
At time $t = 2$ s, we connect DGUs 1 and 2 together. This requires
real-time switch of the local controllers which translates into two
hot plugging-in operations as described in Section \ref{sec:PnP}. The
new decentralized controllers for subsystems
$\subss{\hat{\Sigma}}{1}^{DGU}$  and $\subss{\hat{\Sigma}}{2}^{DGU}$
are designed running Algorithm~\ref{alg:ctrl_design}. Notice that the
interconnection of the two subsystems lead to a variation of each DGU
dynamics, therefore even compensators $\subss{\tilde{C}}{i}$ and
$\subss{{N}}{i}$, $i = 1,2$ need to be updated. In particular, the new
desired closed-loop transfer functions $\tilde F_{i}(s)$, $i = {1, 2}$
have been chosen as low-pass filters with DC gain equal to 0 dB and bandwidth equal to 100 Hz. 

Since Algorithm 1 never stops in Step A, the hot plug-in of the DGUs
is allowed and local controllers get replaced by the new ones at $t =
2$ s. 
Figure \ref{fig:closedLoop2Areas_1} shows the closed-loop eigenvalues of the
overall QSL ImG composed of two interconnected DGUs.
The Bode magnitude plots of compensators $\subss{\tilde{C}}{i}$ and
$\subss{{N}}{i}$, $i = 1,2$ are depicted in Figure \ref{fig:singularValuePreFilter2Areas_1}
and \ref{fig:singularValueCmpAreas_1}, respectively, while the
singular values of the overall
closed-loop transfer function $F(s)$ with inputs
$[z_{ref_{[1]}},z_{ref_{[2]}}]^T$ and outputs $[z_{[1]},z_{[2]}]^T$ are shown in Figure
\ref{fig:singularValueClosedLoop2Areas_1}.
\begin{figure}[!htb]
                 \centering
                 \begin{subfigure}[!htb]{0.48\textwidth}
                   \centering
                   \includegraphics[width=1.1\textwidth, height=150pt]{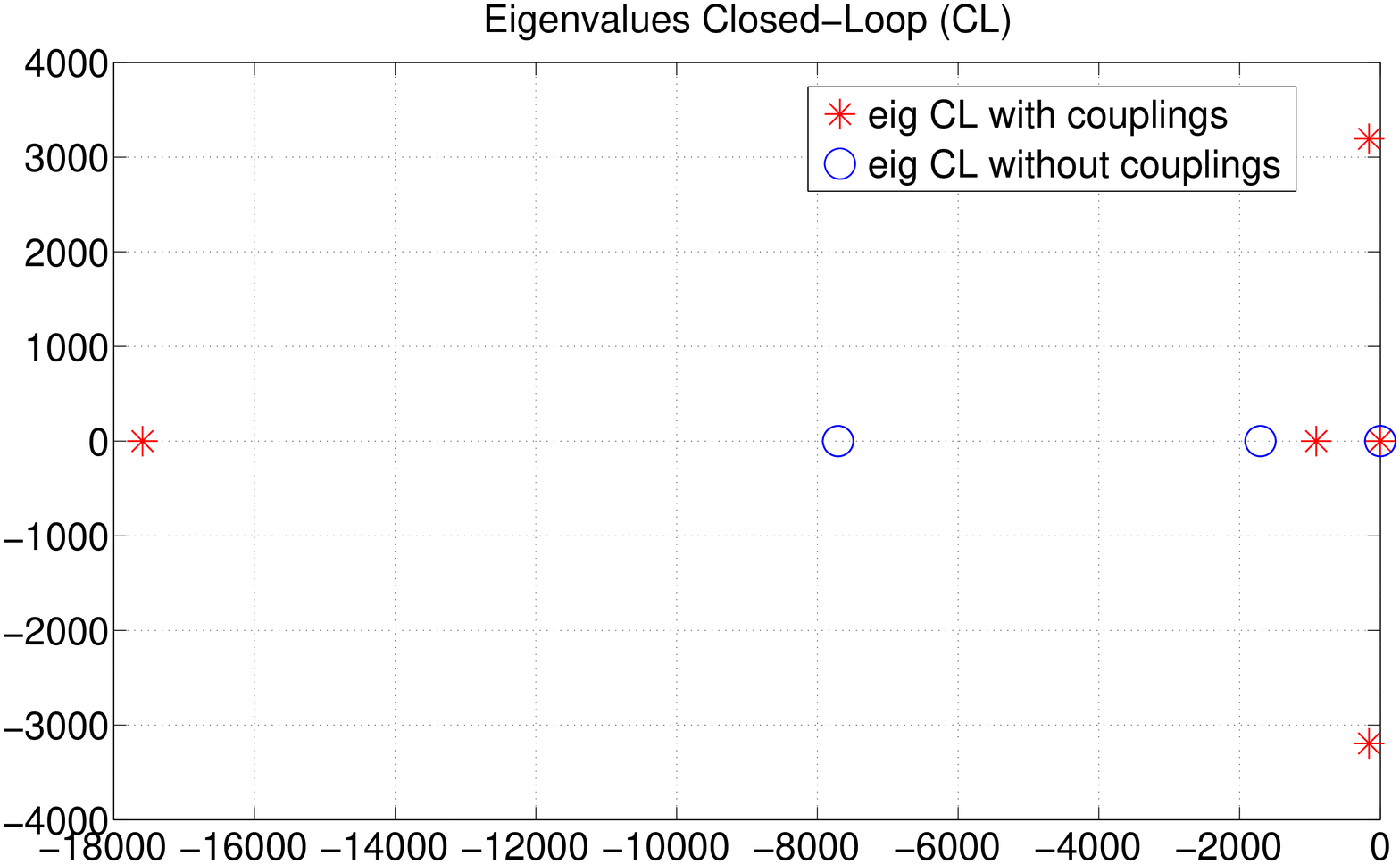}
                   \caption{Eigenvalues of the closed-loop QSL microgrid with (red) and without (blue) couplings.}
                   \label{fig:closedLoop2Areas_1}
                 \end{subfigure}
                 \begin{subfigure}[!htb]{0.48\textwidth}
                   \centering
                   \includegraphics[width=1.1\textwidth, height=150pt]{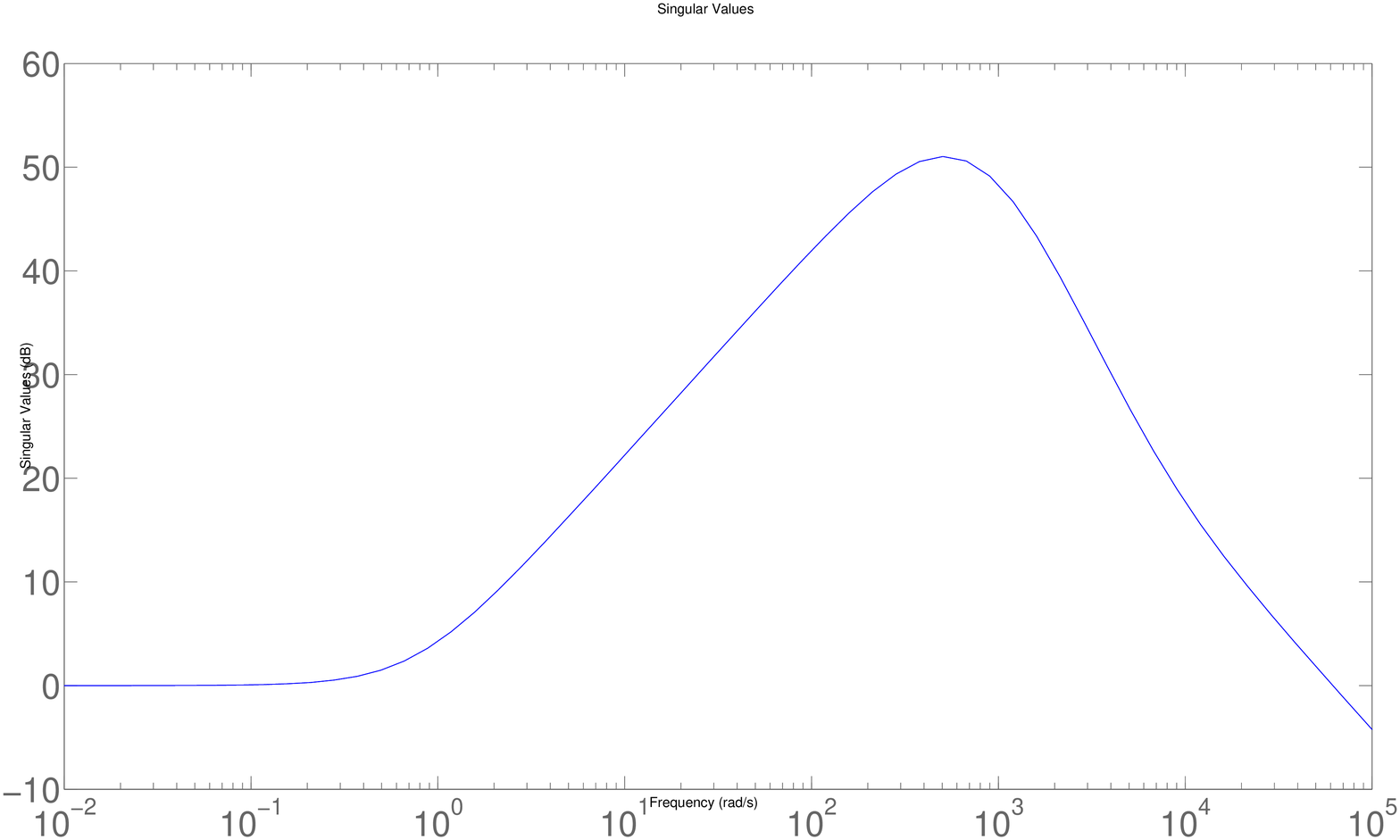}
                   \caption{Bode magnitude plot of pre-filters $\subss{\tilde{C}}{i}$, $i = 1,2$.}
                   \label{fig:singularValuePreFilter2Areas_1}
                 \end{subfigure}
                    \begin{subfigure}[!htb]{0.48\textwidth}
                   \centering
                   \includegraphics[width=1.1\textwidth, height=150pt]{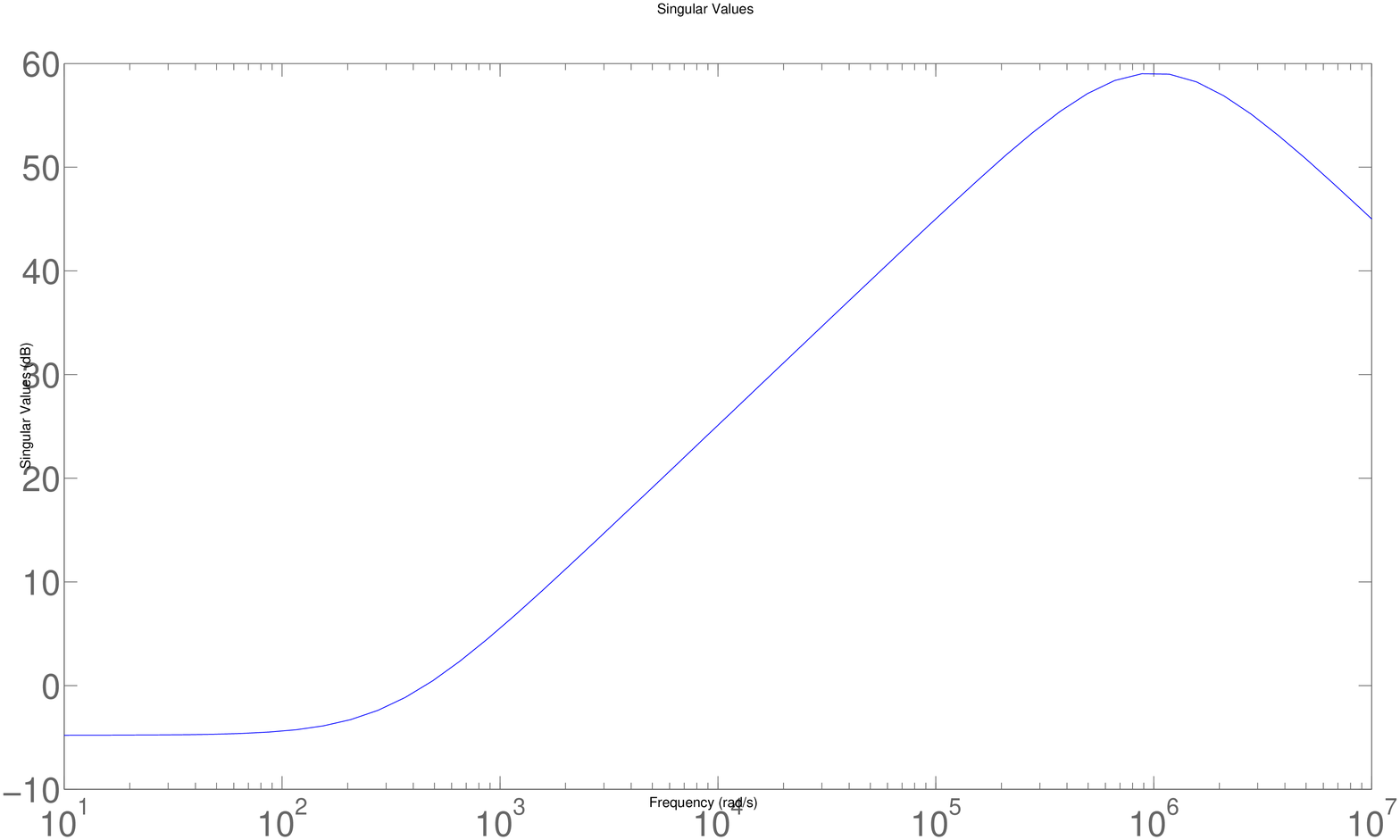}
                   \caption{Bode magnitude plot of disturbances
                     compensators $\subss{{N}}{i}$, $i = 1,2$. }
                   \label{fig:singularValueCmpAreas_1}
                 \end{subfigure}
                 \begin{subfigure}[!htb]{0.48\textwidth}
                   \centering
                   \includegraphics[width=1.1\textwidth, height=150pt]{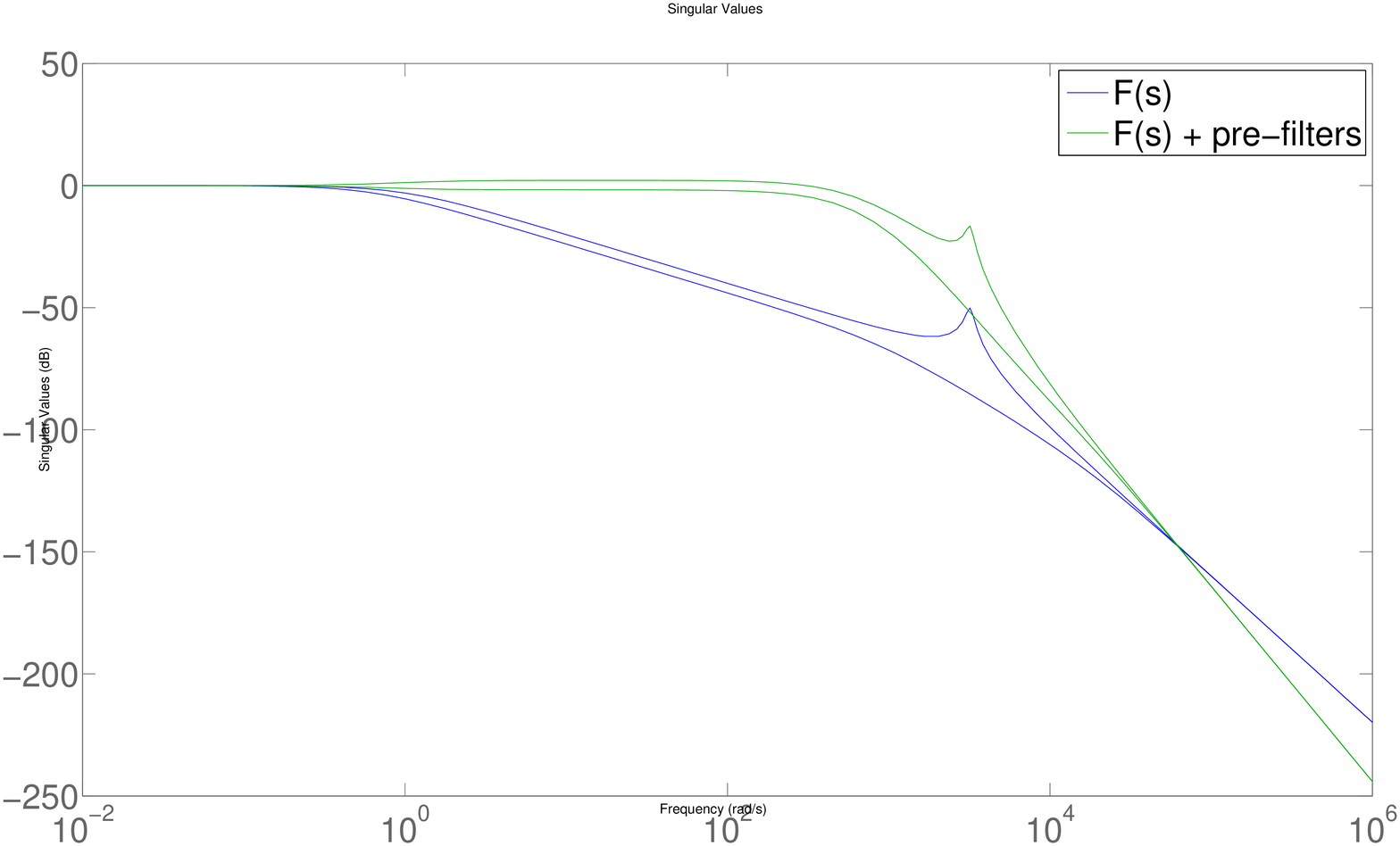}
                   \caption{Singular values of $F(s)$ with (green) and without (blue) pre-filters.}
                   \label{fig:singularValueClosedLoop2Areas_1}
                 \end{subfigure}
                 \caption{Features of PnP controllers for Scenario 1
                   when the DGUs are connected together.}
                 \label{fig:closedLoop2areas_1}
               \end{figure}

Figure \ref{fig:hotpluginDGU} shows the dynamic responses of the
voltages at $PCC_{1}$ and $PCC_{2}$ when the subsystems are connected
together. We highlight that the bumpless control transfer schemes
ensure no significant deviations in the output signals when the
controller switch is performed. Moreover, through the proposed
decentralized control strategy, voltage regulation is excellent.

\begin{figure}[!htb]
                      \centering
                      \begin{subfigure}[!htb]{0.48\textwidth}
                        \centering
                        \includegraphics[width=1\textwidth, height=130pt]{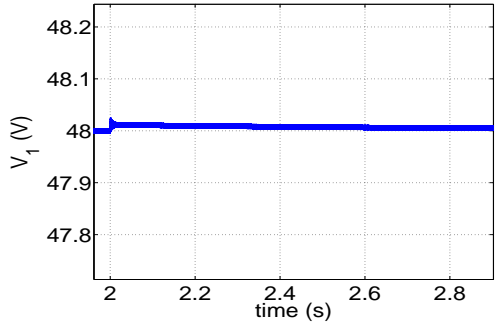}
                        \caption{Load voltage at $PCC_1$.}
                      \end{subfigure}
                      \begin{subfigure}[!htb]{0.48\textwidth}
                        \centering
                        \includegraphics[width=1\textwidth, height=130pt]{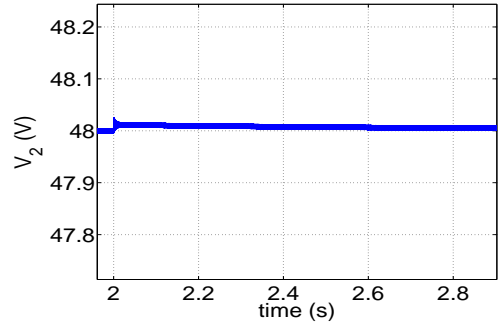}
                        \caption{Load voltage at $PCC_2$.}
                      \end{subfigure}
                        \caption{Scenario 1 - Impact of bumpless
                          control transfer on the hot plug-in at time
                          $t=2$ s.}
                      \label{fig:hotpluginDGU}

                    \end{figure}

                    
               \subsubsection{Robustness to unknown load dynamics}
                    Next, we assess the performance of PnP controllers
                    when loads suddenly change at a certain time. To this purpose, at $t = 3$ s we
                    decrease the load resistances at $PCC_1$ and
                    $PCC_2$ to half of their initial values. Figure
                    \ref{fig:Sc1switch} shows the response of the
                    ImGs. Figures \ref{fig:Sc1switchV1} and
                    \ref{fig:Sc1switchV2} show the load voltage at
                    $PCC_1$ and $PCC_2$ which confirm very good
                    compensation of the current disturbances produced by load changes. We
                    notice small oscillations of the voltage signals
                    due to the presence of complex coniugate poles in
                    the transfer function of the closed-loop overall
                    system including couplings (as shown in Figure \ref{fig:closedLoop2Areas_1}). However, these
                    oscillations disappear after a short transient. We
                    recall that load currents (see Figures
                    \ref{fig:Sc1switchI1} and \ref{fig:Sc1switchI2})
                    are treated as measurable disturbances in our
                    model. Varying the load resistance, induces
                    step-like changes in the disturbances.
                    \begin{figure}[!htb]
                      \centering
                      \begin{subfigure}[!htb]{0.48\textwidth}
                        \centering
                        \includegraphics[width=1\textwidth, height=130pt]{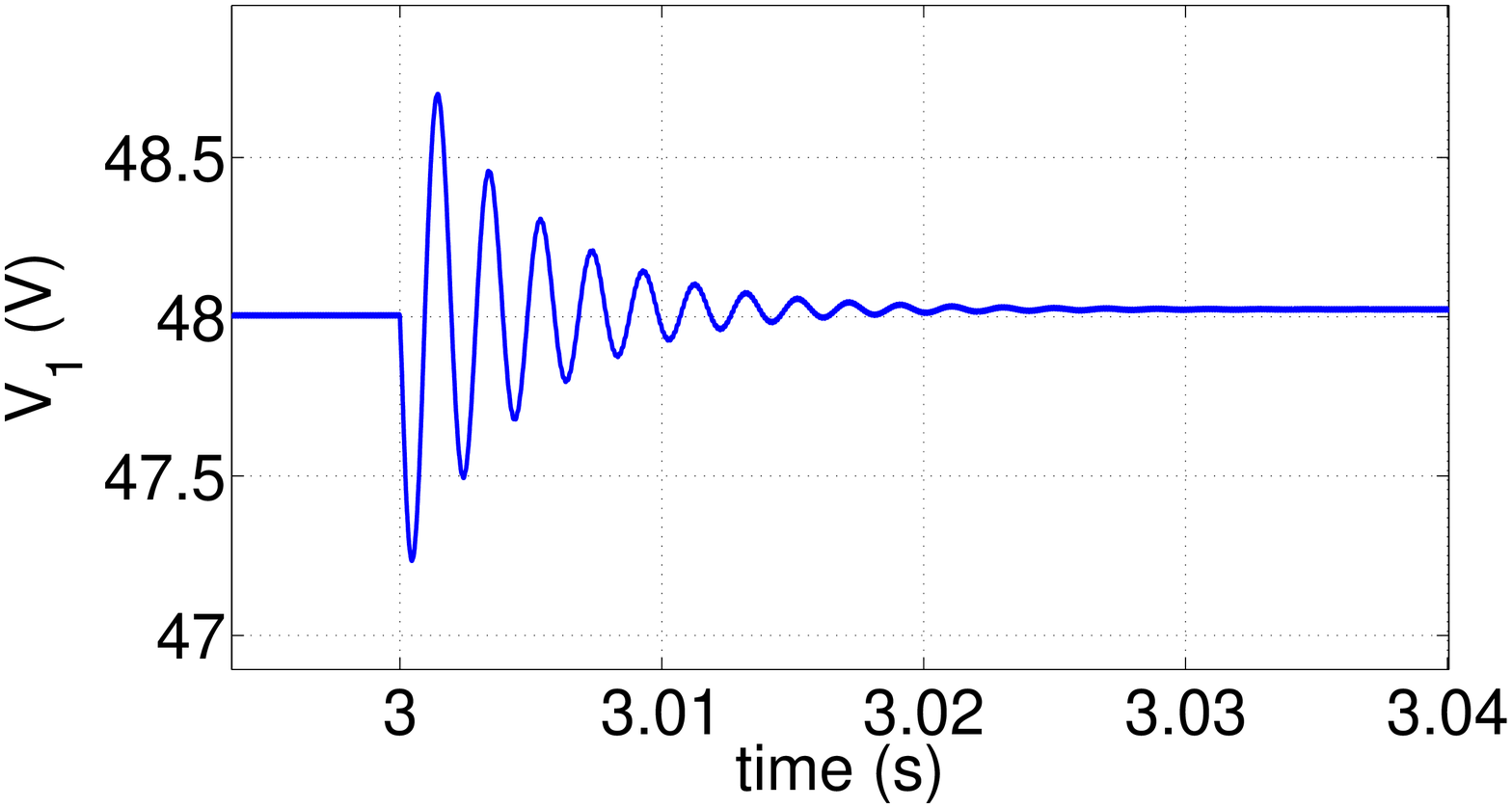}
                        \caption{Voltage at $PCC_1$.}
                        \label{fig:Sc1switchV1}
                      \end{subfigure}
                      \begin{subfigure}[!htb]{0.48\textwidth}
                        \centering
                        \includegraphics[width=1\textwidth, height=130pt]{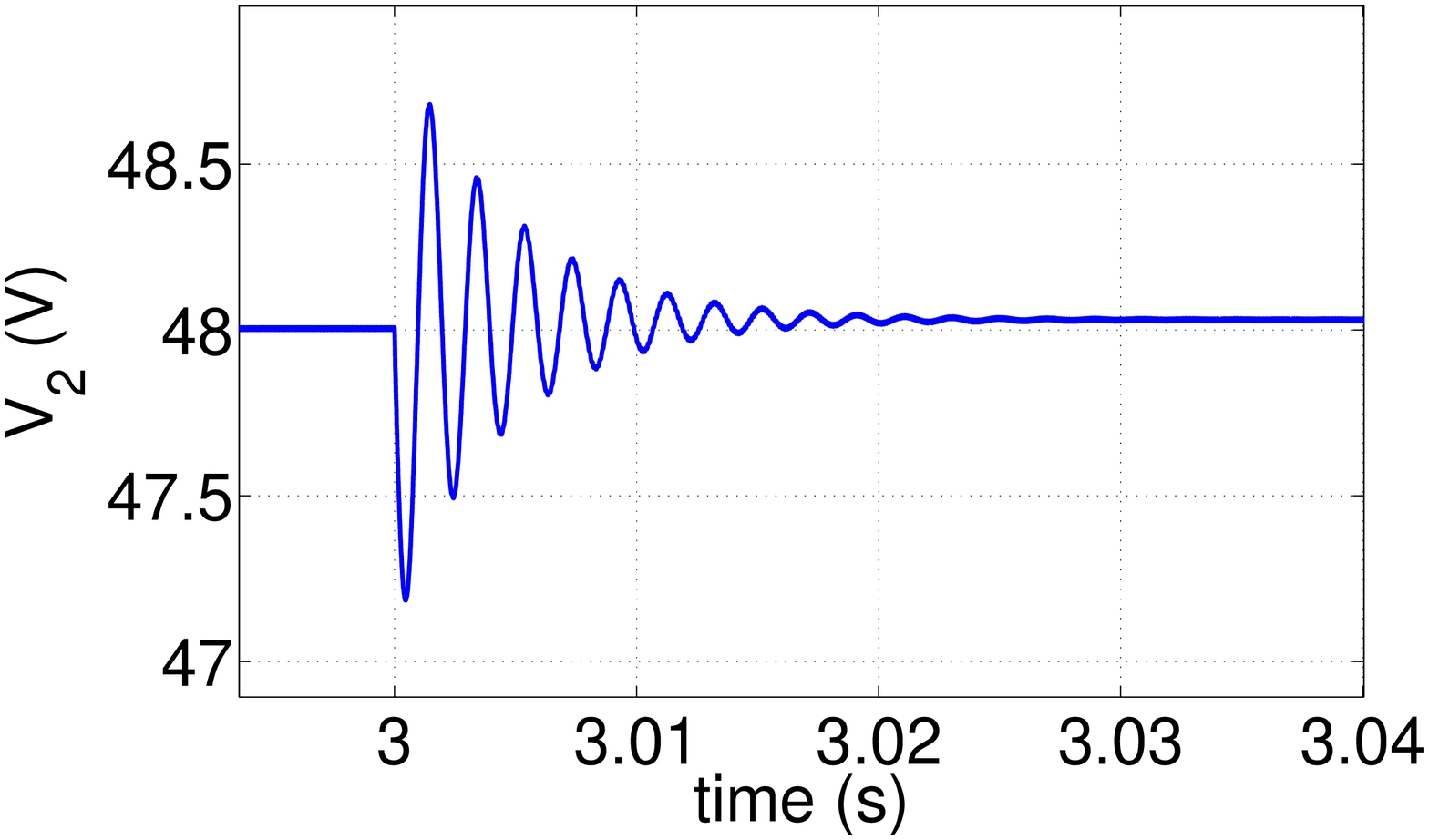}
                        \caption{Voltage at $PCC_2$.}
                        \label{fig:Sc1switchV2}
                      \end{subfigure}
                      \begin{subfigure}[!htb]{0.48\textwidth}
                        \centering
                        \includegraphics[width=1\textwidth, height=130pt]{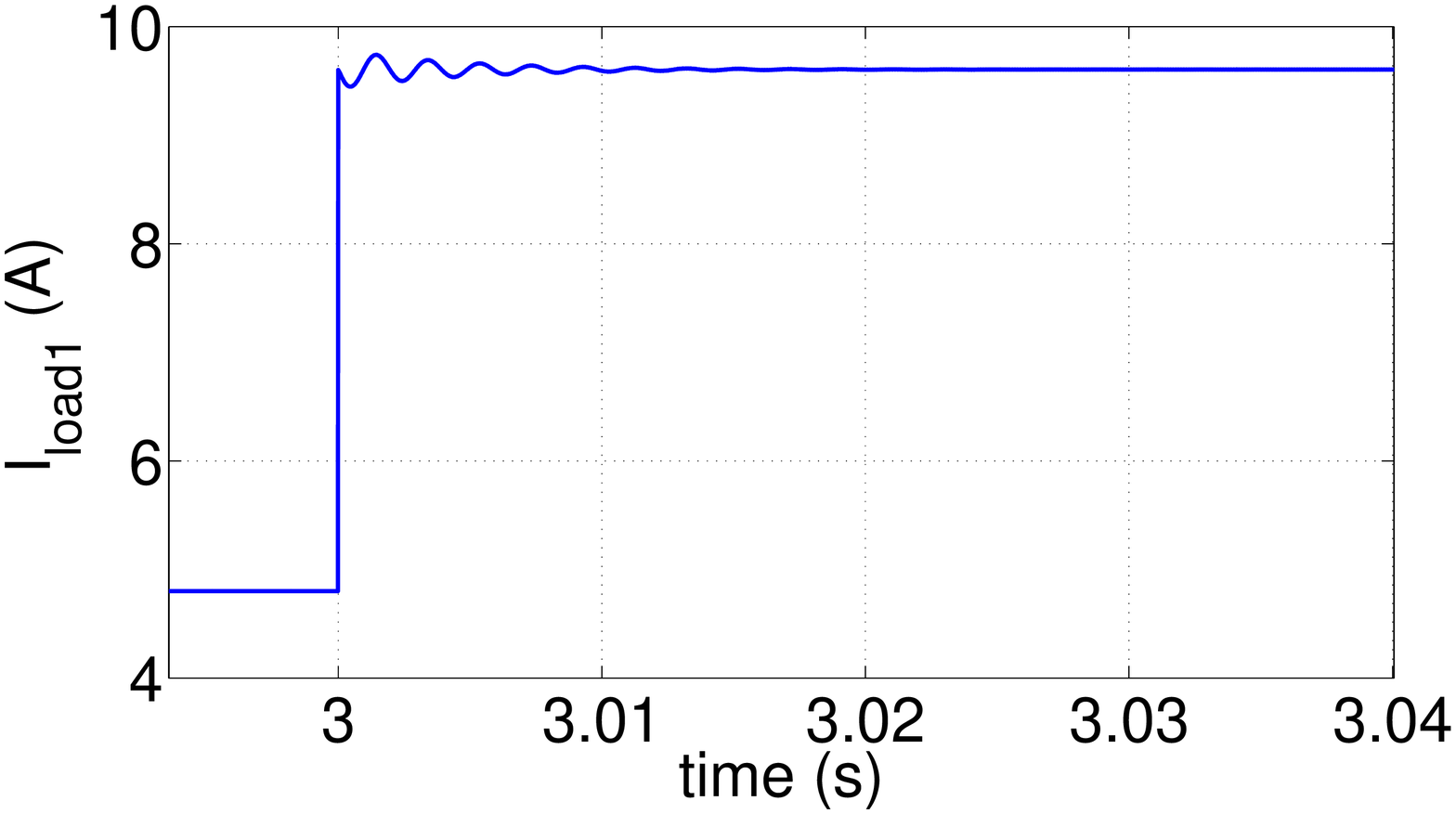}
                        \caption{Instantaneous load current $I_{L1}$.}
                        \label{fig:Sc1switchI1}
                      \end{subfigure}
                      \begin{subfigure}[!htb]{0.48\textwidth}
                        \centering
                        \includegraphics[width=1\textwidth, height=130pt]{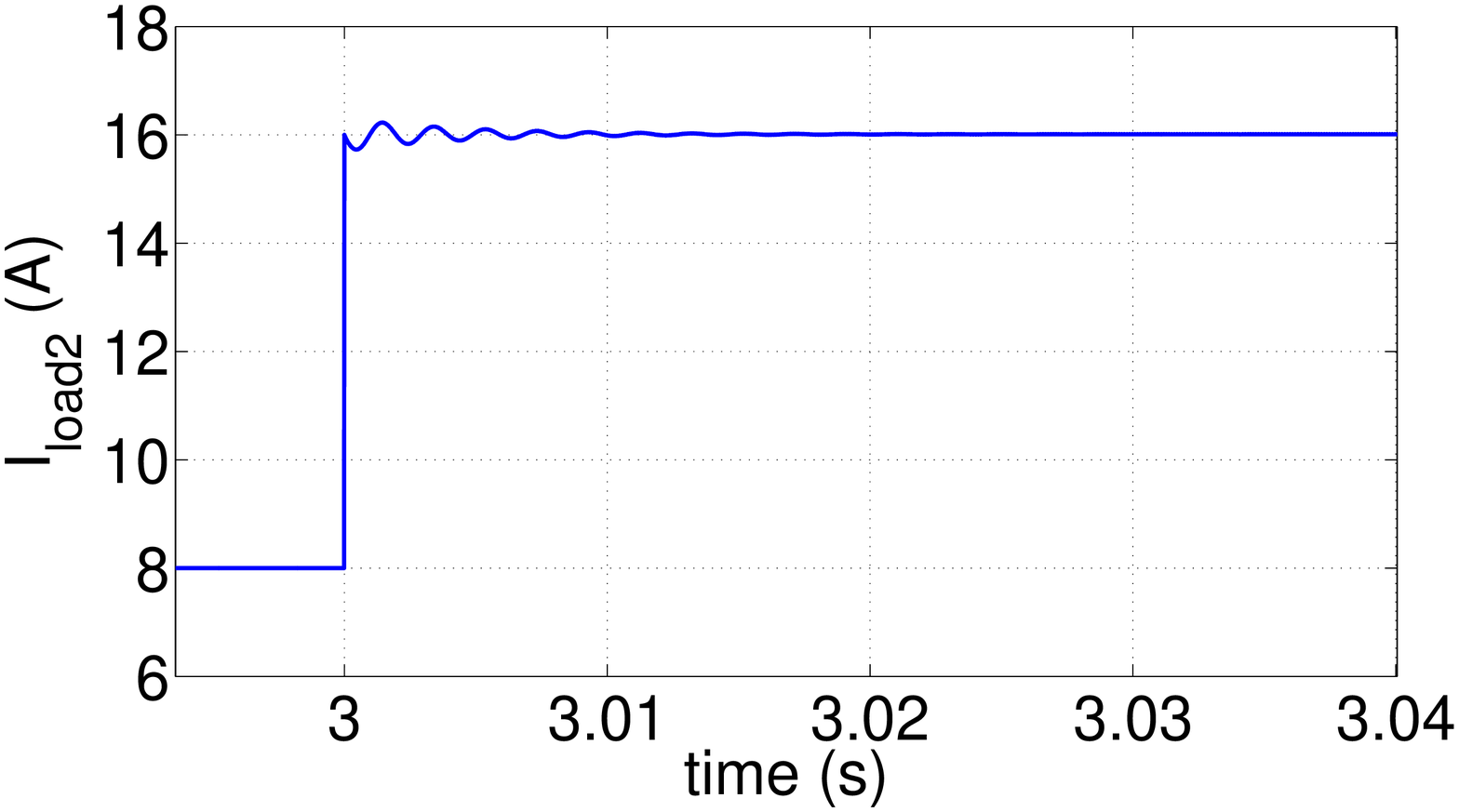}
                        \caption{Instantaneous load current $I_{L2}$.}
                        \label{fig:Sc1switchI2}
                      \end{subfigure}
                      \caption{Scenario 1 - Performance of PnP
                        decentralized voltage control in presence of
                        load switches at time $t = 3$ s.}
                      \label{fig:Sc1switch}
                    \end{figure}

               \subsubsection{Voltage tracking for DGU 1}
                    \label{sec:differentC}
                    Finally, we evaluate the performance in tracking step changes in the voltage reference at one $PCC$ (e.g. $PCC_1$) when the DGUs are connected together. This test is of particular concern if we look at the concrete implementation of islanded DC microgrids. In fact, changes in the voltage references can be required in order to regulate power flow among the DGUs, or to control the state-of-charge of possible batteries embedded in the ImG.

                    To this purpose, at $t = 4$ s we let the reference
                    signal of DGU 1, $v^\star_{1, MG}$, step down to
                    $47.6\mbox{ }V$. Notice that this small variation
                    of the voltage reference at $PCC_1$ is sufficient
                    to let an appreciable amount of current flow
                    through the line, since the line impedance is
                    quite small. The dynamic responses of the overall
                    microgrid to this change are shown in Figure
                    \ref{fig:Vtrack}. As one can see, controllers guarantees good tracking performances in a reasonable time with small interactions between the two DGUs. 

   \begin{figure}[!htb]
                      \centering
                      \begin{subfigure}[!htb]{0.48\textwidth}
                        \centering
                        \includegraphics[width=1\textwidth, height=130pt]{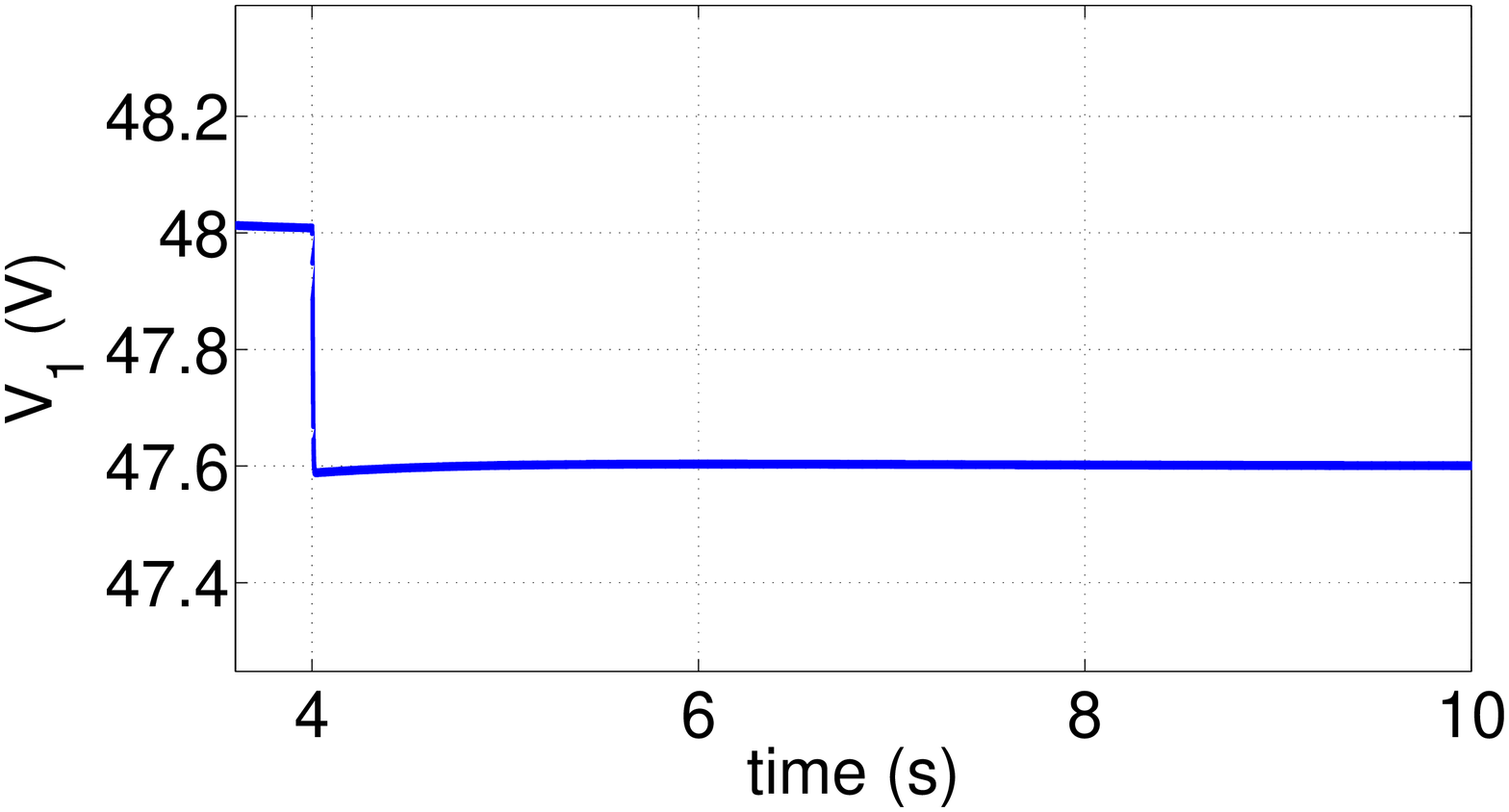}
                        \caption{Voltage at $PCC_1$.}
                        \label{fig:V1track}
                      \end{subfigure}
                      \begin{subfigure}[!htb]{0.48\textwidth}
                        \centering
                        \includegraphics[width=1\textwidth, height=130pt]{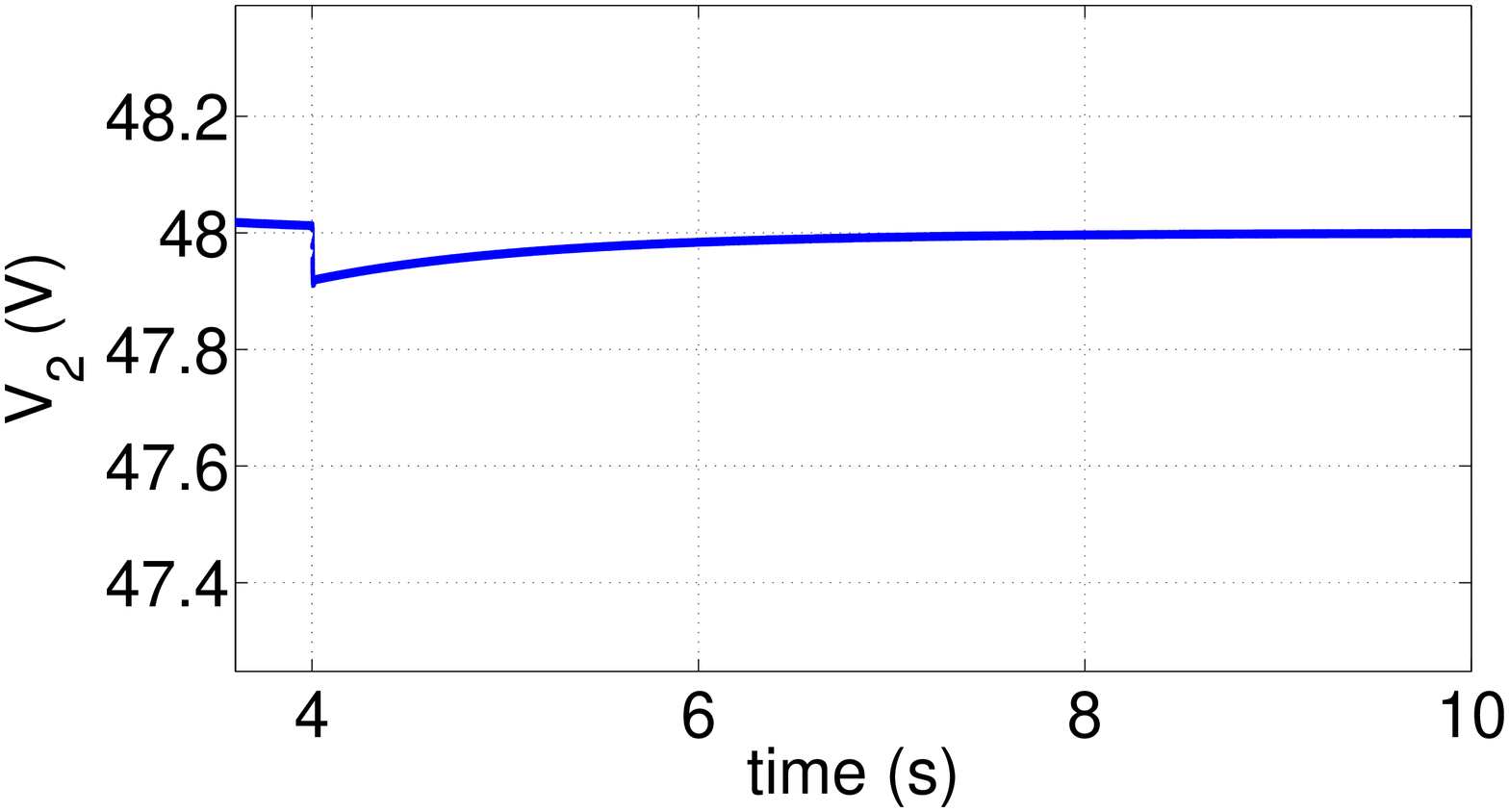}
                        \caption{Voltage at $PCC_2$.}
                        \label{fig:V2track}
                      \end{subfigure}
                     \caption{Scenario 1 - Performance of PnP decentralized voltage controllers in terms of set-point tracking for DGU 1.}
                      \label{fig:Vtrack}                 
                    \end{figure}

            \subsection{Scenario 2}
             \label{sec:scenario_2}
               In this second scenario, we consider the meshed ImG shown in Figure \ref{fig:5area_1}. As
               one can notice, the
               main difference with respect to Scenario 1 is that some
               DGUs have more than one neighbour. This means
               that the disturbances influencing their dynamics will
               be greater. Moreover, the presence of a loop further
               complicates voltage regulation. 
               To our knowledge, control of loop-interconnected DGUs
               has never been investigated for
               DC microgrids. 
    
      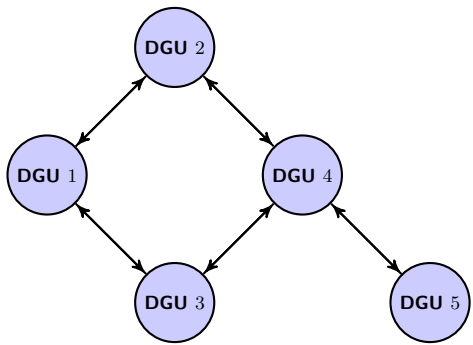
\begin{figure}[!htb]
                 \centering
                 \begin{tikzpicture}[scale=0.8,transform shape,->,>=stealth',shorten >=1pt,auto,node distance=3cm, thick,main node/.style={circle,fill=blue!20,draw,font=\sffamily\bfseries}]
					  
  \node[main node] (1) {\small{DGU $1$}};
  \node[main node] (2) [above right of=1] {\small{DGU $2$}};
  \node[main node] (3) [below right of=1] {\small{DGU $3$}};
  \node[main node] (4) [above right of=3] {\small{DGU $4$}};
  \node[main node] (5) [below right of=4] {\small{DGU $5$}};
 
  \path[every node/.style={font=\sffamily\small}]
  (1) edge node [left] {} (3)
  (3) edge node [right] {} (1)

  (1) edge node [left] {} (2)
  (2) edge node [right] {} (1)
  
  (3) edge node [left] {} (4)
  (4) edge node [right] {} (3)
  
  (2) edge node [left] {} (4)
  (4) edge node [right] {} (2)
  
  (4) edge node [left] {} (5)
  (5) edge node [right] {} (4);
\end{tikzpicture}
                 \caption{Scenario 2 - Scheme of the ImG
                   composed of 5 DGUs, connected through transmission
                   lines (black arrows).}
                 \label{fig:5area_1}
               \end{figure}

In order to assess the capability of the proposed decentralized
approach to cope with heterogeneous dynamics, we
consider an ImG composed of DGUs with non-identical electrical parameters. They are listed in Tables \ref{tbl:diffpar5}, \ref{tbl:linespar5}
and \ref{tbl:commpar5} in Appendix \ref{sec:AppElectrPar}.

We also assume that DGUs 1-5 supply  $10\mbox{ } \Omega$,
$6\mbox{ } \Omega$, $20\mbox{ } \Omega$, $2 \mbox{ }\Omega$ and
$4\mbox{ } \Omega$ loads, respectively. Moreover, we highlight that, for this Scenario, no compensators $\tilde C_{i}$ and $N_i$ have been used. 
               At the beginning of the simulation, all the DGUs are
               assumed to be isolated and not connected to each
               other. However, we choose to equip each subsystem
               $\subss{\hat{\Sigma}}{i}^{DGU}$,
               $i\in\DD=\{1,\dots,5\}$, with controllers
               $\subss{\CC}{i}$ designed by running Algorithm
               \ref{alg:ctrl_design} and taking into account couplings
               among DGUs. This is possible because, as shown in
               Section \ref{sec:riferimento}, local controllers
               stabilize the ImG
               also in absence of couplings. Because of this choice of local controllers
               in the startup phase, when the five subsystems are
               connected together at time $t = 1.5$ s, no bumpless
               control scheme is required since no real-time switch of controllers is performed. The
               closed-loop eigenvalues of the overall QSL ImG are depicted in Figure
               \ref{fig:eig_5Areas_0} while Figure \ref{fig:Fs_5Areas_0} shows the closed-loop
               transfer function of the whole microgrid.
\begin{figure}[!htb]
                 \centering
                  \begin{subfigure}[!htb]{0.48\textwidth}
                   \centering
                   \includegraphics[width=1.1\textwidth, height=150pt]{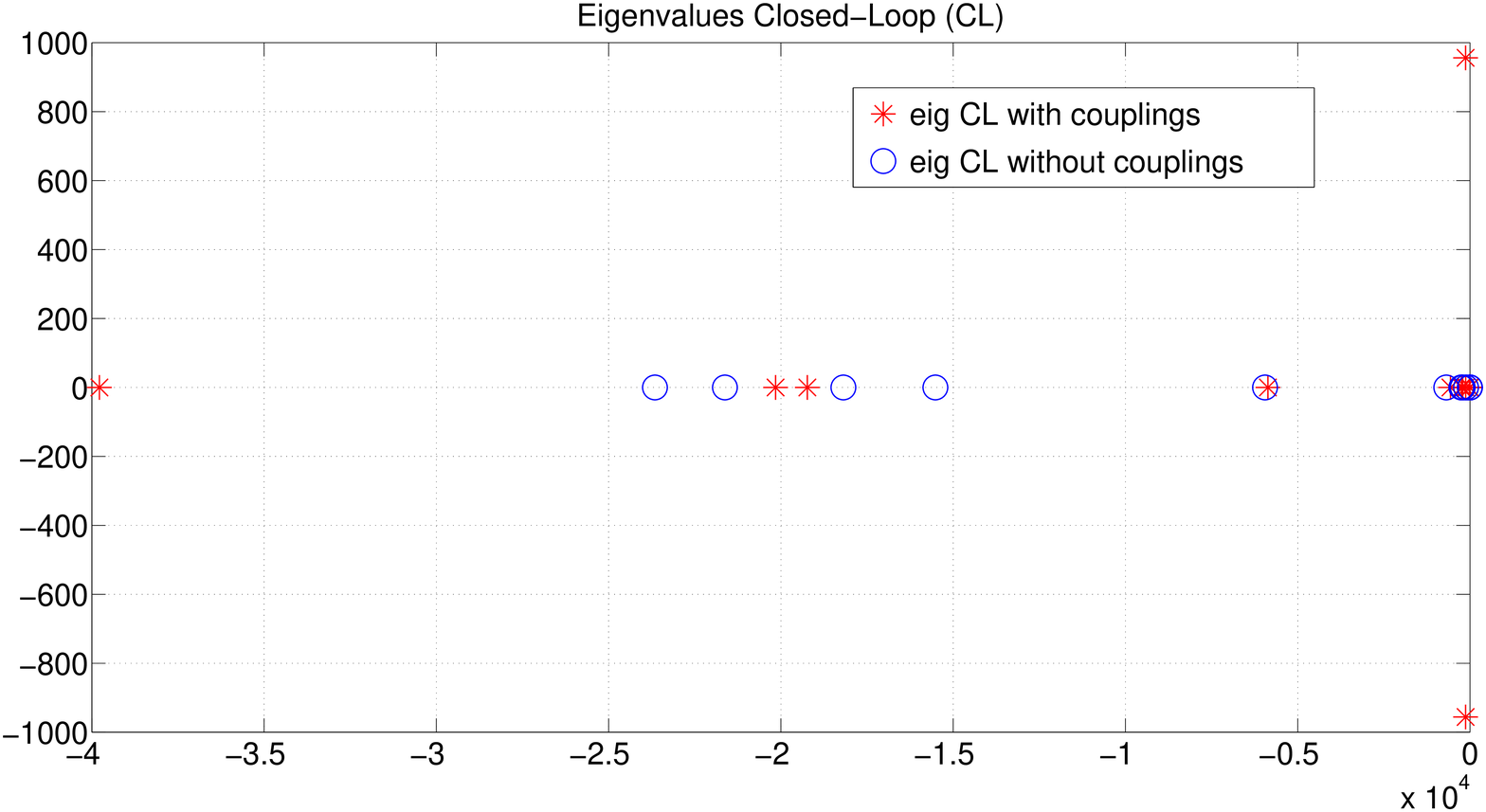}
                   \caption{Eigenvalues of the closed-loop QSL microgrid with (red) and without (blue) couplings.}
                   \label{fig:eig_5Areas_0}
                 \end{subfigure}
                 \begin{subfigure}[!htb]{0.48\textwidth}
                   \centering
                   \includegraphics[width=1.1\textwidth, height=150pt]{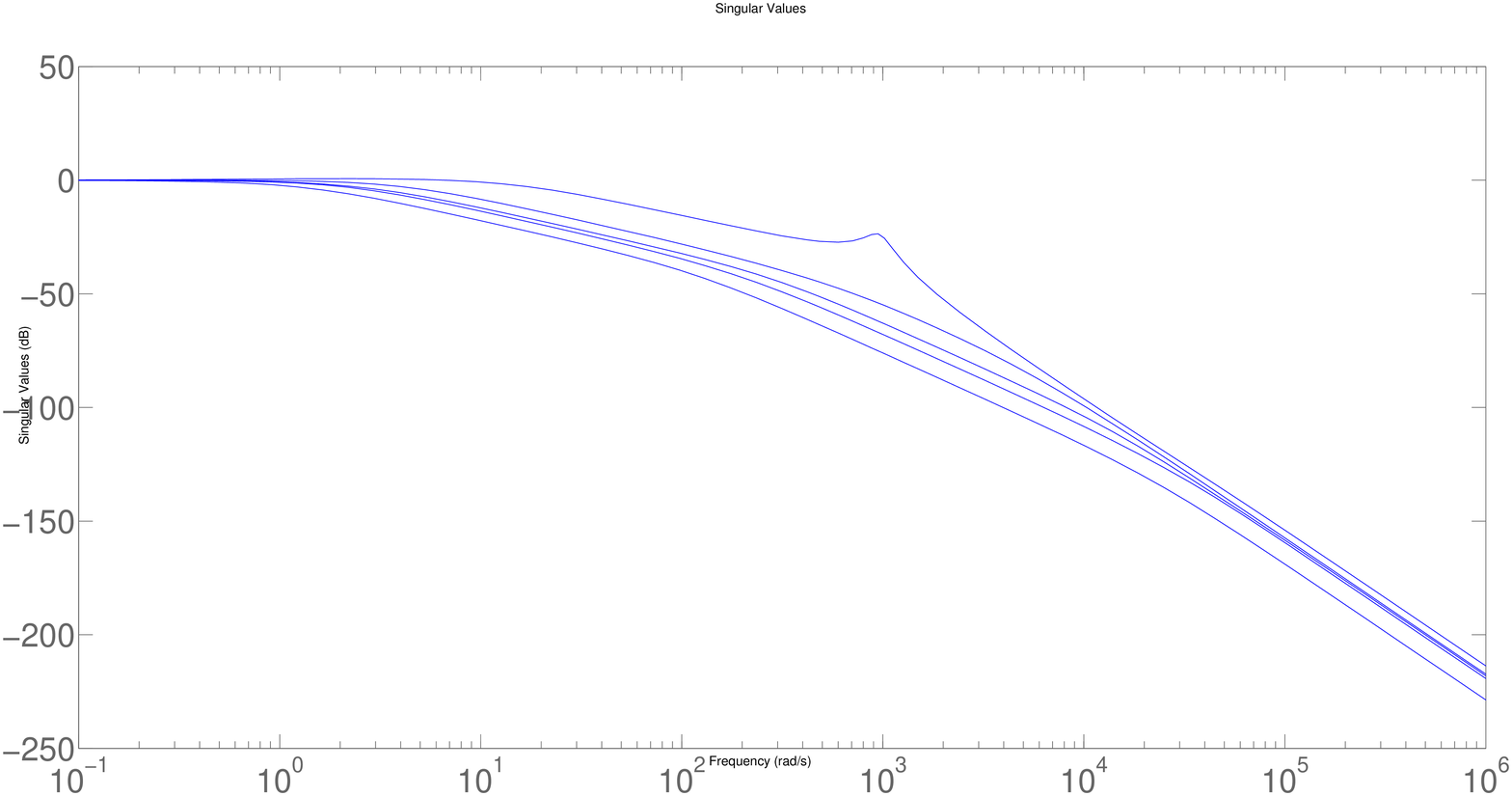}
                   \caption{ Singular values of $F(s)$.}
                   \label{fig:Fs_5Areas_0}
                 \end{subfigure}
                 \caption{Features of PnP controllers for Scenario 2
                   with 5 interconnected DGUs.}
                 \label{fig:5areas}
               \end{figure}

\subsubsection{Plug-in of a new DGU}
               For evaluating the PnP capabilities of our control
               approach, we simulate the connection of DGU
               $\subss{\hat{\Sigma}}{6}^{DGU}$ with
               $\subss{\hat{\Sigma}}{1}^{DGU}$ and
               $\subss{\hat{\Sigma}}{5}^{DGU}$, as shown in Figure
               \ref{fig:5areasplug}. Therefore, we have
               $\NN_{6}=\{1,5\}$. 
\begin{figure}[!htb]
                 \centering
                 \begin{tikzpicture}[scale=0.8,transform shape,->,>=stealth',shorten >=1pt,auto,node distance=3cm, thick,main node/.style={circle,fill=blue!20,draw,font=\sffamily\bfseries}]
					  
  \node[main node] (1) {\small{DGU $1$}};
  \node[main node] (2) [above right of=1] {\small{DGU $2$}};
  \node[main node] (3) [below right of=1] {\small{DGU $3$}};
  \node[main node] (4) [above right of=3] {\small{DGU $4$}};
  \node[main node] (5) [below right of=4] {\small{DGU $5$}};
  \node[main node] (6) [below of=3] {\small{DGU $6$}};
 
  \path[every node/.style={font=\sffamily\small}]
  (1) edge node [left] {} (3)
  (3) edge node [right] {} (1)

  (1) edge node [left] {} (2)
  (2) edge node [right] {} (1)
  
  (3) edge node [left] {} (4)
  (4) edge node [right] {} (3)
  
  (2) edge node [left] {} (4)
  (4) edge node [right] {} (2)
  
  (4) edge node [left] {} (5)
  (5) edge node [right] {} (4);

  \draw[red] (1) to (6);
  \draw[red] (6) to (1);
  \draw[red] (5) to (6);
  \draw[red] (6) to (5);

\end{tikzpicture}
                 \caption{Scenario 2 - Scheme of the ImG
                   composed of 6 DGUs after the plugging-in of $\subss{\hat{\Sigma}}{6}^{DGU}$.}
                 \label{fig:5areasplug}
               \end{figure}
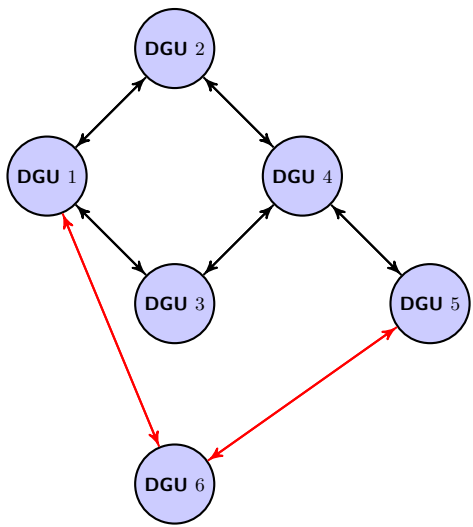
In principle, subsystems $\subss{\hat{\Sigma}}{j}^{DGU}$,
               $j\in\NN_{6}$ must update their controllers
               $\subss{\CC}{j}$  (see Section \ref{sec:PnP}). However,
               we highlight that previous controllers for DGUs
               $\subss{\hat{\Sigma}}{1}^{DGU}$ and
               $\subss{\hat{\Sigma}}{5}^{DGU}$ can be also maintained,
               provided that the already computed matrices $K_j$, $j\in\NN_{6}$ still
               fulfill all constraints in \eqref{eq:optproblem}
               for the new ImG topology. Since this test
               succeeds, we proceed by executing  Algorithm
               \ref{alg:ctrl_design} for synthesizing $\subss{\CC}{6}$
               for the new DGU only. Algorithm \ref{alg:ctrl_design} never
               stops in Step \ref{enu:stepAalgCtrl} and therefore the addition of $\subss{\hat{\Sigma}}{6}^{DGU}$ is allowed.
      The real-time plugging-in of
               $\subss{\hat{\Sigma}}{6}^{DGU}$ is executed at time
               $t=2$ s. Until the plug-in of
               $\subss{\hat{\Sigma}}{6}^{DGU}$, common reference
               $v^\star_{MG}$ for DGUs 1-5 is the same as for DGUs 1-2
               in Scenario 1 and the subsystem
               $\subss{\hat{\Sigma}}{6}^{DGU}$ is assumed to work
               isolated, tracking the reference voltage
               $v^\star_{MG}$. Figures \ref{fig:eig_2} and \ref{fig:F_s_2} show
               respectively the closed-loop eigenvalues and the singular
               values of the
               closed-loop $F(s)$ of the overall QSL ImG in Figure \ref{fig:5areasplug} equipped with the
               controllers described above. From Figure
               \ref{fig:Sc2_V_0}, we note that right after the hot
               plug-in of $\subss{\hat{\Sigma}}{6}^{DGU}$  at $t=2$ s, load voltages of $\subss{\hat{\Sigma}}{1}^{DGU}$ and $\subss{\hat{\Sigma}}{5}^{DGU}$ do not deviate from the respective reference signals. 
               \begin{figure}[!htb]
                 \centering
                  \begin{subfigure}[!htb]{0.48\textwidth}
                   \centering
                   \includegraphics[width=1.1\textwidth, height=150pt]{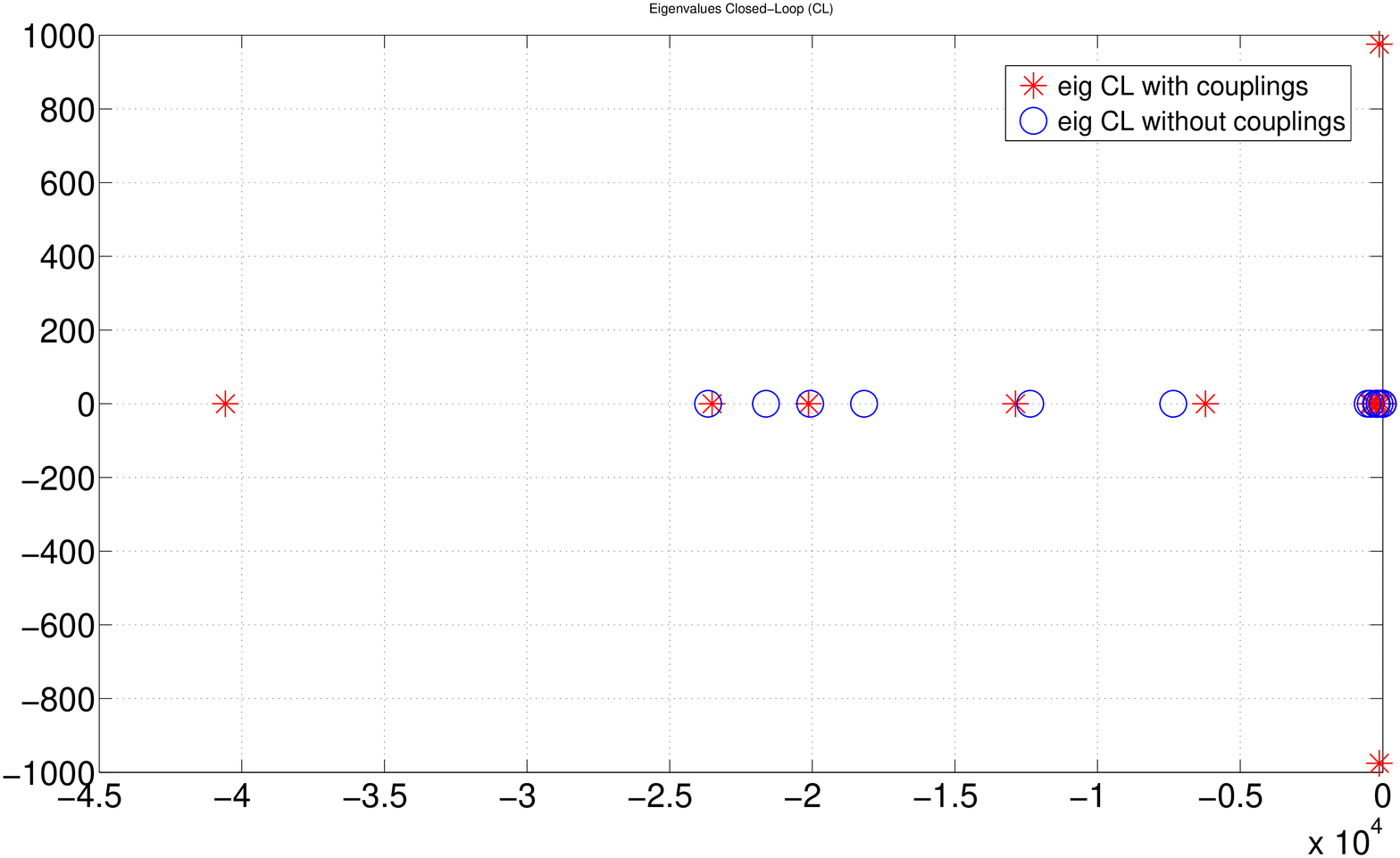}
                   \caption{Eigenvalues of the closed-loop QSL microgrid with (red) and without (blue) couplings.}
                   \label{fig:eig_2}
                 \end{subfigure}
                 \begin{subfigure}[!htb]{0.48\textwidth}
                   \centering
                   \includegraphics[width=1.1\textwidth, height=150pt]{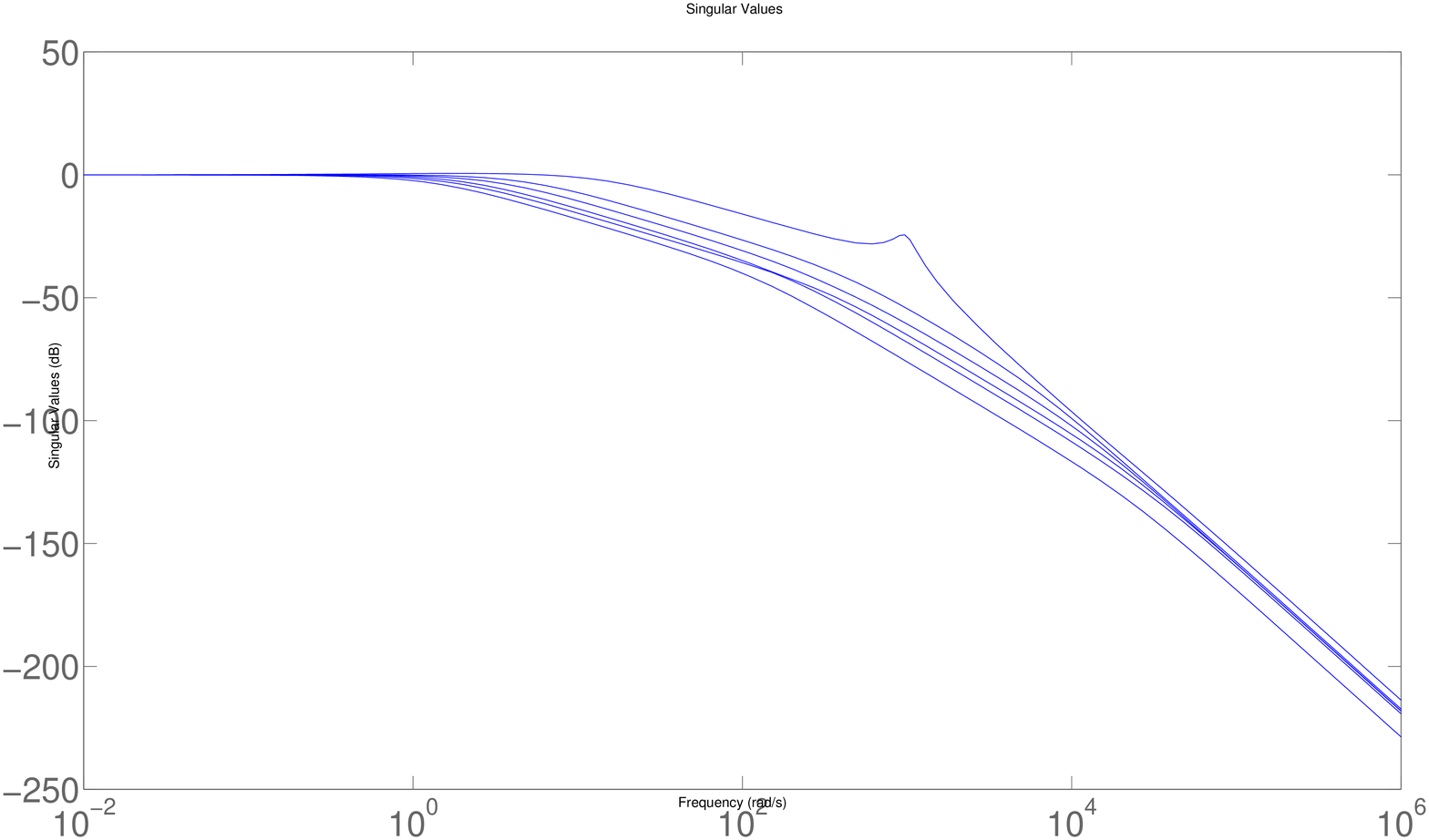}
                   \caption{ Singular values of $F(s)$.}
                   \label{fig:F_s_2}
                 \end{subfigure}
                 \caption{Features of PnP controllers for Scenario 2
                   with 6 interconnected DGUs}
                 \label{fig:closedLoop5areas}
               \end{figure}

\subsubsection{Robustness to unknown load dynamics}
In order to test the robustness of the overall ImG to unknown load
dynamics, at $t=3$ s we halve the load of DGU 6, which was equal to $8 \mbox{ }\Omega$ for $t<3$ s.

  \begin{figure}[!htb]
                      \centering
                      \begin{subfigure}[!htb]{0.48\textwidth}
                        \centering
                        \includegraphics[width=1\textwidth, height=130pt]{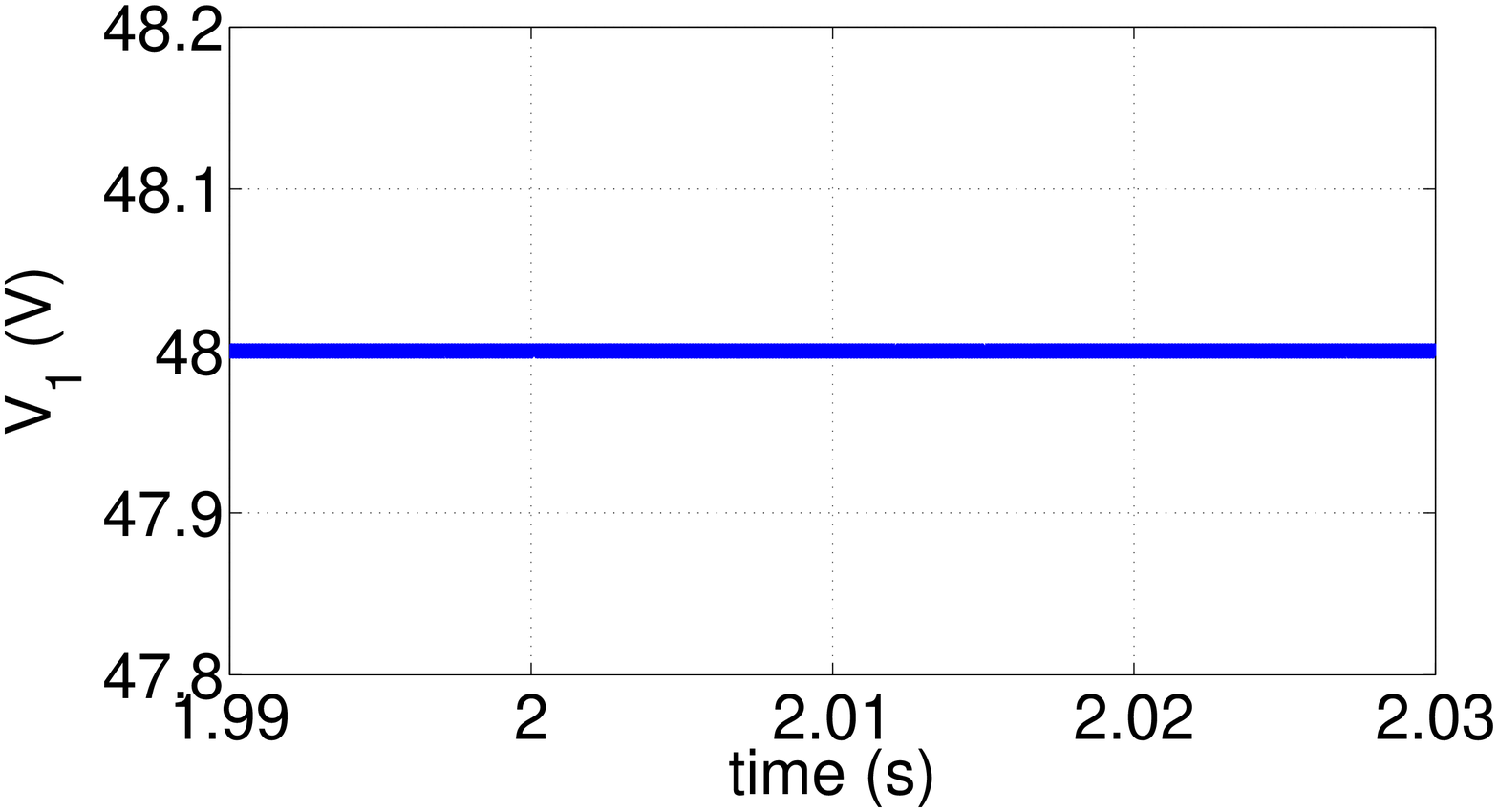}
                        \caption{Voltage at $PCC_1$.}
                        \label{fig:Sc2_V1}
                      \end{subfigure}
                      \begin{subfigure}[!htb]{0.48\textwidth}
                        \centering
                        \includegraphics[width=1\textwidth, height=130pt]{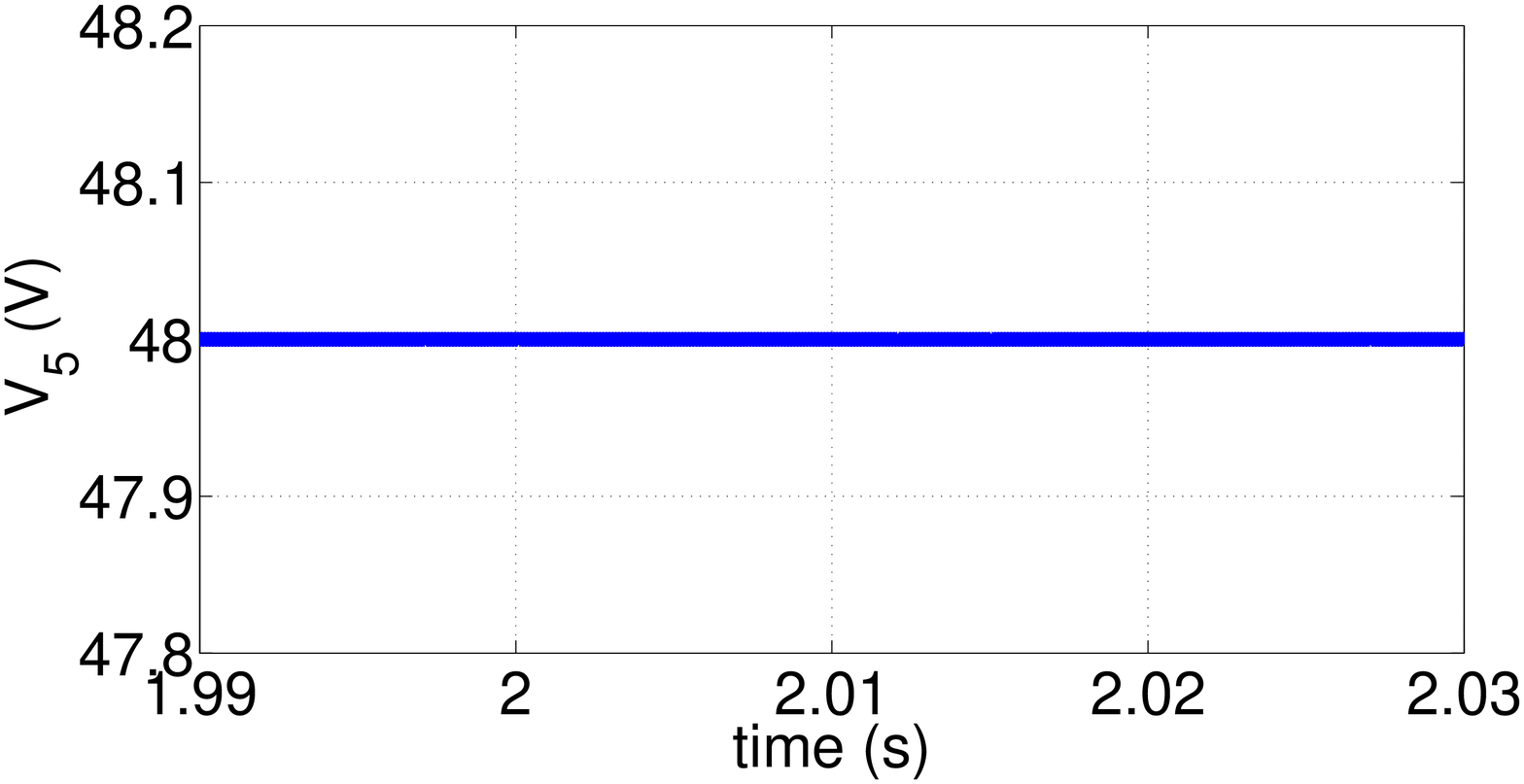}
                        \caption{Voltage at $PCC_5$.}
                        \label{fig:Sc2_V5}
                      \end{subfigure}
                      \begin{subfigure}[!htb]{0.48\textwidth}
                        \centering
                        \includegraphics[width=1\textwidth, height=130pt]{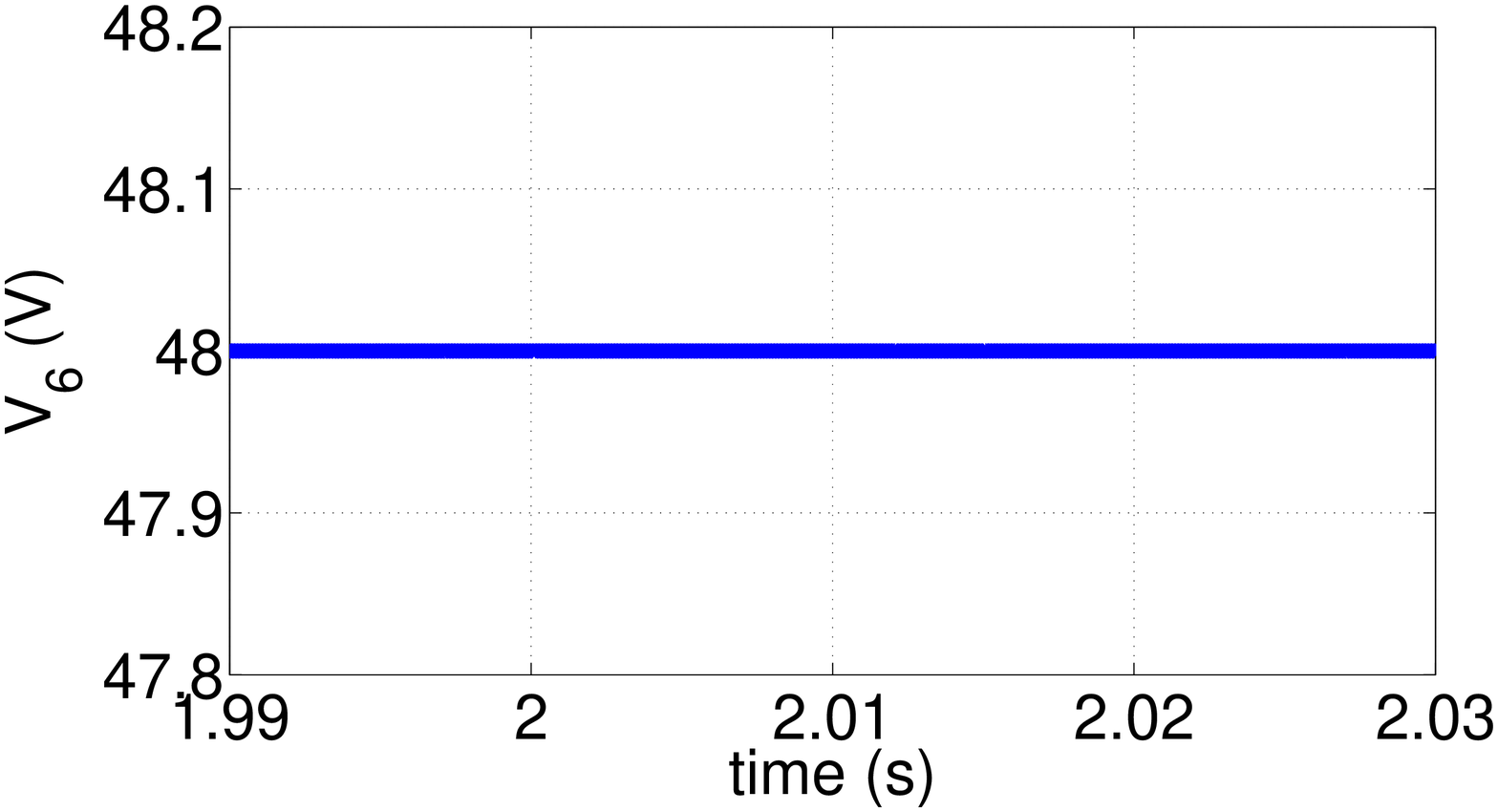}
                        \caption{Voltage at $PCC_6$.}
                        \label{fig:Sc2_V6}
                      \end{subfigure}
                     \caption{Scenario 2 - Performance of PnP
                       decentralized voltage controllers during the
                       hot plug-in of DGU 6 at time $t=2$ s.}
                      \label{fig:Sc2_V_0}                 
                    \end{figure}
Figures \ref{fig:Sc2_V1} and \ref{fig:Sc2_V5} show that, when the load
change of $\subss{\hat{\Sigma}}{6}^{DGU}$ occurs, the voltages at
$PCC_1$ and $PCC_5$ exhibit very small variations which last for a
short time. Then, load voltages of $\subss{\hat{\Sigma}}{1}^{DGU}$
and $\subss{\hat{\Sigma}}{5}^{DGU}$ converge to their reference
values. Similar remarks can be done for the new DGU
$\subss{\hat{\Sigma}}{6}^{DGU}$: as shown in Figure \ref{fig:Sc2_V6},
there is a short transient at the time of the load change, that is
effectively compensated by the control action. These experiments
highlight that controllers $\subss{\CC}{i}$, $i = 1,\dots,6$ may
ensure very good tracking of the reference signal and robustness to unknown load dynamics even without using compensators $\subss{\tilde{C}}{6}$ and $\subss{N}{6}$.
\begin{figure}[!htb]
                      \centering
                      \begin{subfigure}[!htb]{0.48\textwidth}
                        \centering
                        \includegraphics[width=1\textwidth, height=130pt]{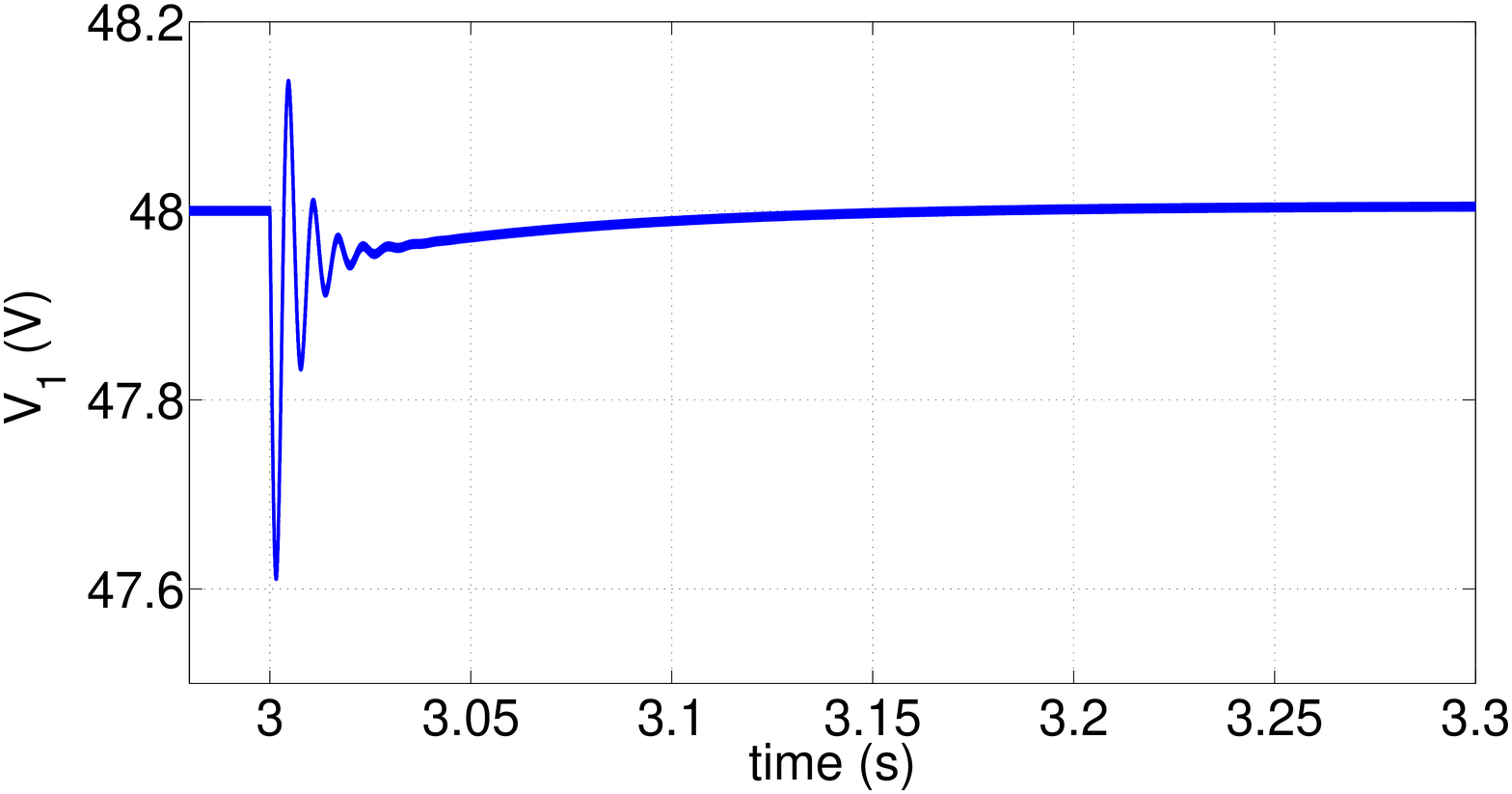}
                        \caption{Voltage at $PCC_1$.}
                        \label{fig:Sc2_V1}
                      \end{subfigure}
                      \begin{subfigure}[!htb]{0.48\textwidth}
                        \centering
                        \includegraphics[width=1\textwidth, height=130pt]{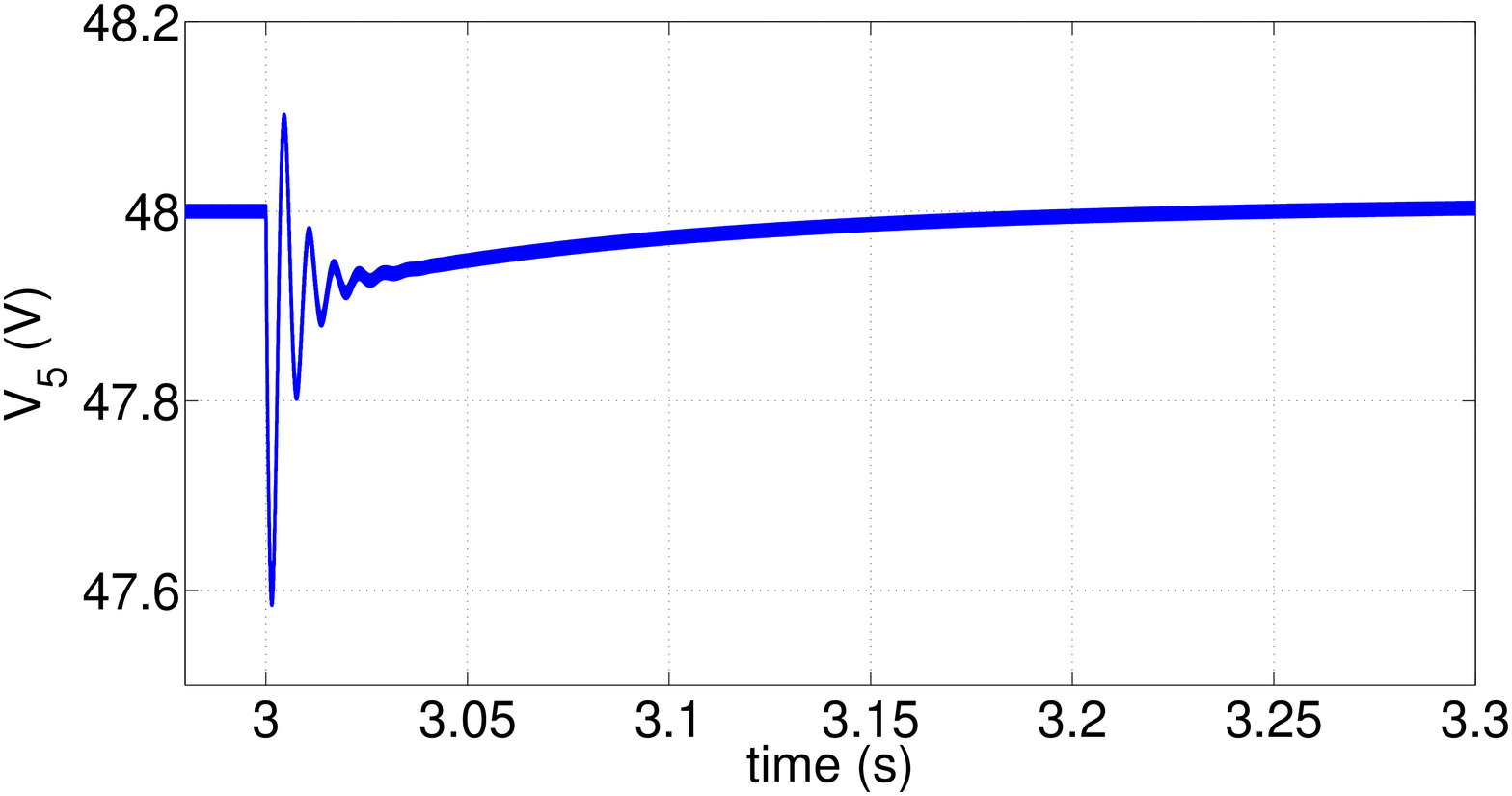}
                        \caption{Voltage at $PCC_5$.}
                        \label{fig:Sc2_V5}
                      \end{subfigure}
                      \begin{subfigure}[!htb]{0.48\textwidth}
                        \centering
                        \includegraphics[width=1\textwidth, height=130pt]{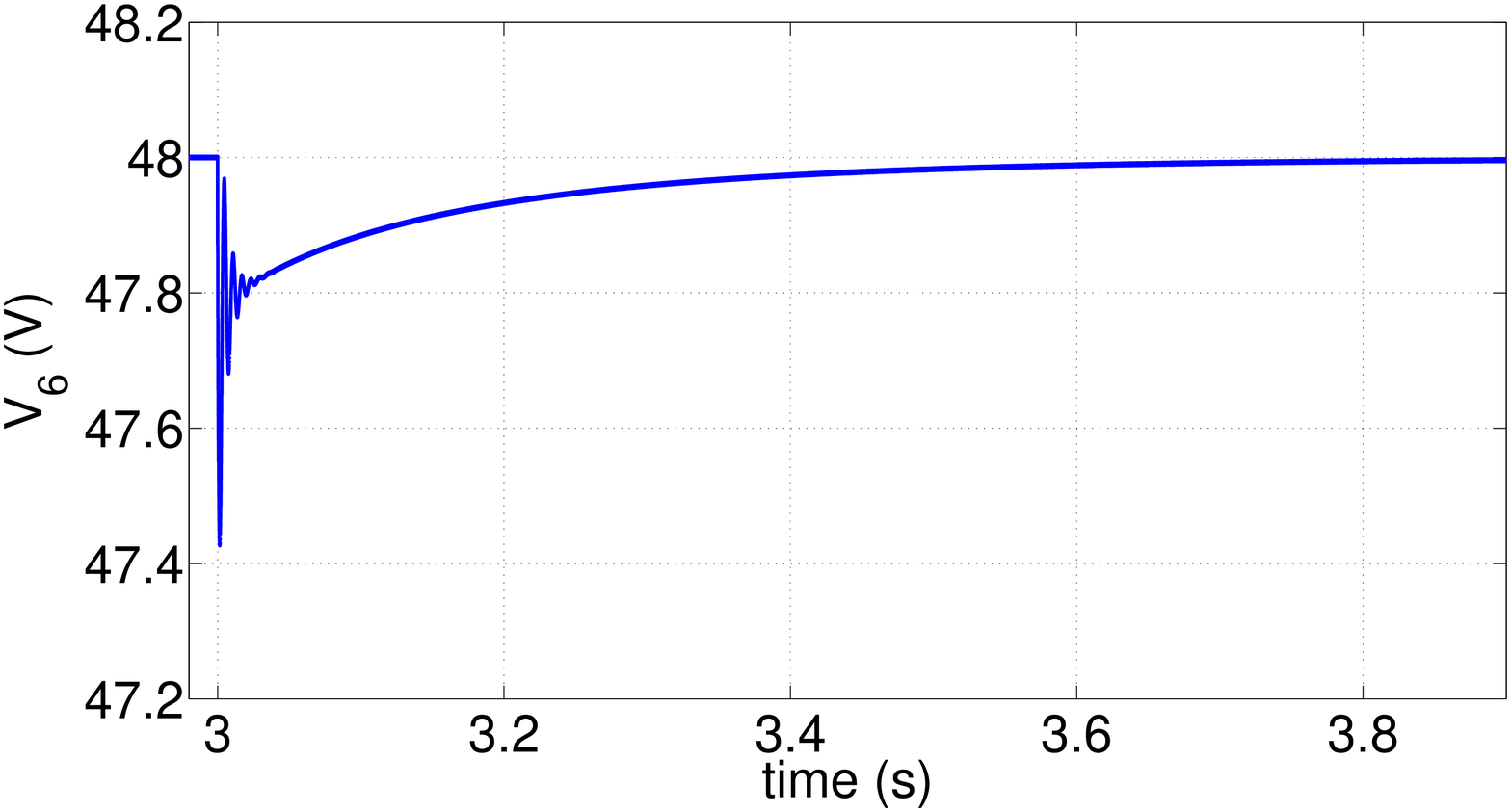}
                        \caption{Voltage at $PCC_6$.}
                        \label{fig:Sc2_V6}
                      \end{subfigure}
                     \caption{Scenario 2 - Performance of PnP
                       decentralized voltage controllers in terms of
                       robustness to an abrupt change of load
                       resistances at time $t = 3$ s.}
                      \label{fig:Sc2_V_1}                 
                    \end{figure}
\subsubsection{Unplugging of a DGU}
Next, we simulate the disconnection of $\subss{\hat{\Sigma}}{3}^{DGU}$
so that the considered ImG assumes the topology shown in Figure
\ref{fig:5areas_unplug}. The set of neighbours of DGU 3 is $\NN_{3}=\{1,4\}$.

 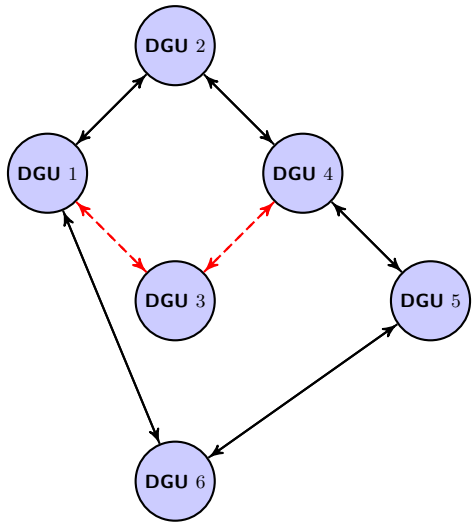
\begin{figure}[!htb]
                 \centering
                 \begin{tikzpicture}[scale=0.8,transform shape,->,>=stealth',shorten >=1pt,auto,node distance=3cm, thick,main node/.style={circle,fill=blue!20,draw,font=\sffamily\bfseries}]
					  
 \node[main node] (1) {\small{DGU $1$}};
  \node[main node] (2) [above right of=1] {\small{DGU $2$}};
  \node[main node] (3) [below right of=1] {\small{DGU $3$}};
  \node[main node] (4) [above right of=3] {\small{DGU $4$}};
  \node[main node] (5) [below right of=4] {\small{DGU $5$}};
  \node[main node] (6) [below of=3] {\small{DGU $6$}};
 
 \path[every node/.style={font=\sffamily\small}]
  (1) edge node [left] {} (6)
  (6) edge node [right] {} (1)

  (1) edge node [left] {} (2)
  (2) edge node [right] {} (1)
  
  (5) edge node [left] {} (6)
  (6) edge node [right] {} (5)
  
  (2) edge node [left] {} (4)
  (4) edge node [right] {} (2)
  
  (4) edge node [left] {} (5)
  (5) edge node [right] {} (4);

  \draw[red,dashed] (1) to (3);
  \draw[red,dashed] (3) to (1);
  \draw[red,dashed] (3) to (4);
  \draw[red,dashed] (4) to (3);

\end{tikzpicture}
                 \caption{Scenario 2 - Scheme of the ImG
                   composed of 5 DGUs after the unplugging of $\subss{\hat{\Sigma}}{3}^{DGU}$.}
                 \label{fig:5areas_unplug}
               \end{figure}

Because of the disconnection, there is a change in
the local dynamics $\hat{A}_{jj}$ of DGUs
$\subss{\hat{\Sigma}}{j}^{DGU}$, $j\in\NN_3$. Then, in theory, each controller
$\subss{\CC}{j}$, $j\in\NN_3$ must be redesigned (see Section \ref{sec:PnP}). As for the plugging-in
operation, we decide to maintain the previous controller for DGUs 1
and 4, after checking that the already computed matrices $K_j$, $j\in\NN_{3}$
fulfill all constraints in \eqref{eq:optproblem} even when DGU 3 is
removed. Since this test ends successfully, the disconnection of
$\subss{\hat{\Sigma}}{3}^{DGU}$ is allowed. Figure \ref{fig:eig_posteriori_unpl} shows that the
closed-loop model of the new QSL microgrid is still asymptotically
stable in spite of the unplugging operation while Figure \ref{fig:F_s_posteriori_unpl} shows
the closed-loop transfer function $F(s)$ of the ImG.
Hot-unplugging of $\subss{\hat{\Sigma}}{3}^{DGU}$ is performed at time $t=7$ s. As shown in Figure
\ref{fig:Sc2_V_2}, the load voltages of DGU
$\subss{\hat{\Sigma}}{j}^{DGU}$, $j\in\NN_3$ do not deviate from the
respective reference signals. We stress again that stability of the microgrid is
preserved despite the disconnection of 
$\subss{\hat{\Sigma}}{3}^{DGU}$. 
           \begin{figure}[!htb]
                 \centering
                  \begin{subfigure}[!htb]{0.48\textwidth}
                   \centering
                   \includegraphics[width=1.1\textwidth, height=150pt]{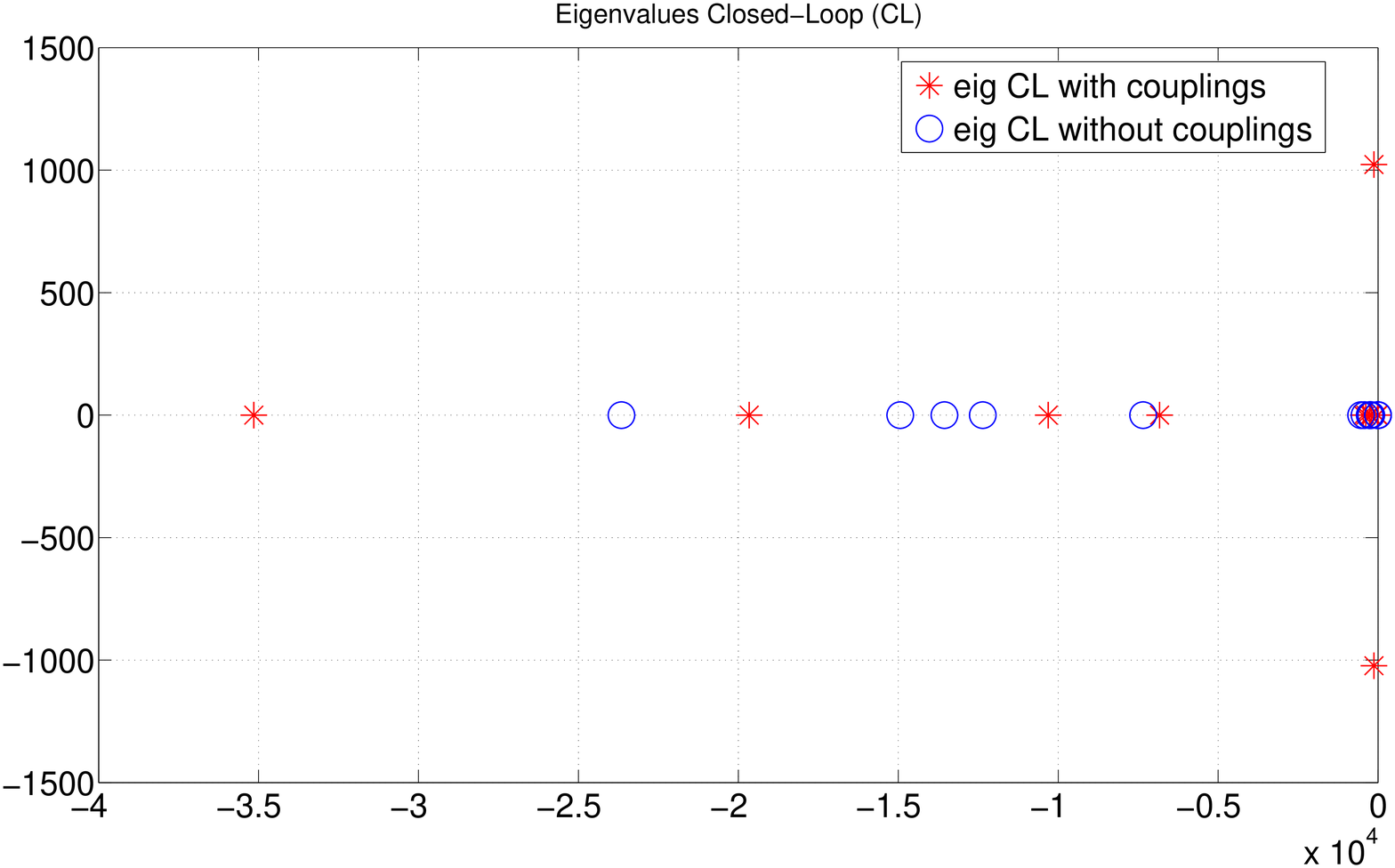}
                   \caption{Eigenvalues of the closed-loop QSL microgrid with (red) and without (blue) couplings.}
                   \label{fig:eig_posteriori_unpl}
                 \end{subfigure}
                 \begin{subfigure}[!htb]{0.48\textwidth}
                   \centering
                   \includegraphics[width=1.1\textwidth, height=150pt]{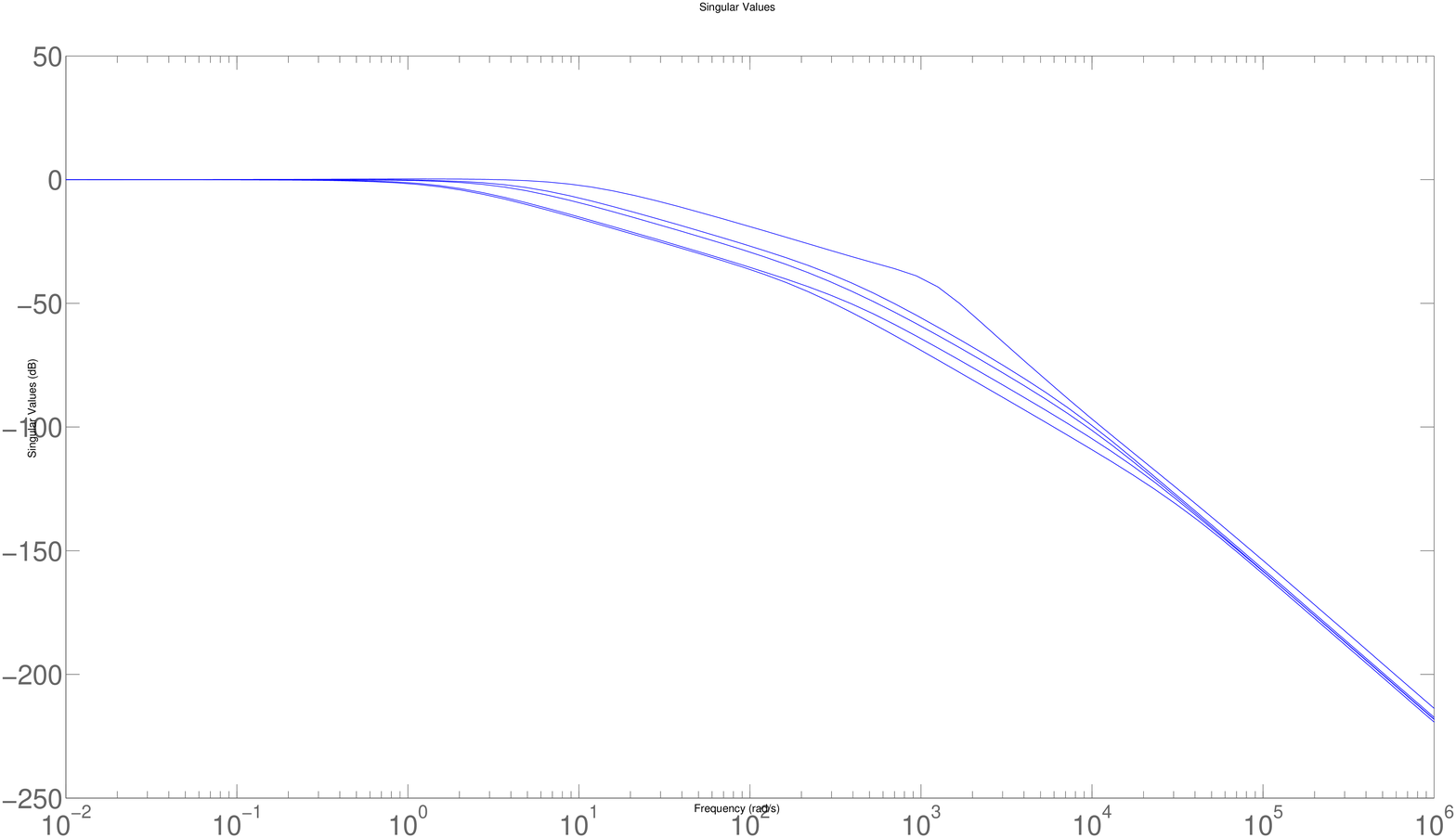}
                   \caption{ Singular values of $F(s)$.}
                   \label{fig:F_s_posteriori_unpl}
                 \end{subfigure}
                 \caption{Features of PnP controllers for Scenario 2
                   after the unplugging of DGU 3.}
                 \label{fig:posteriori_unplugging}
               \end{figure}    

 \begin{figure}[!htb]
                      \centering
                      \begin{subfigure}[!htb]{0.48\textwidth}
                        \centering
                        \includegraphics[width=1\textwidth, height=130pt]{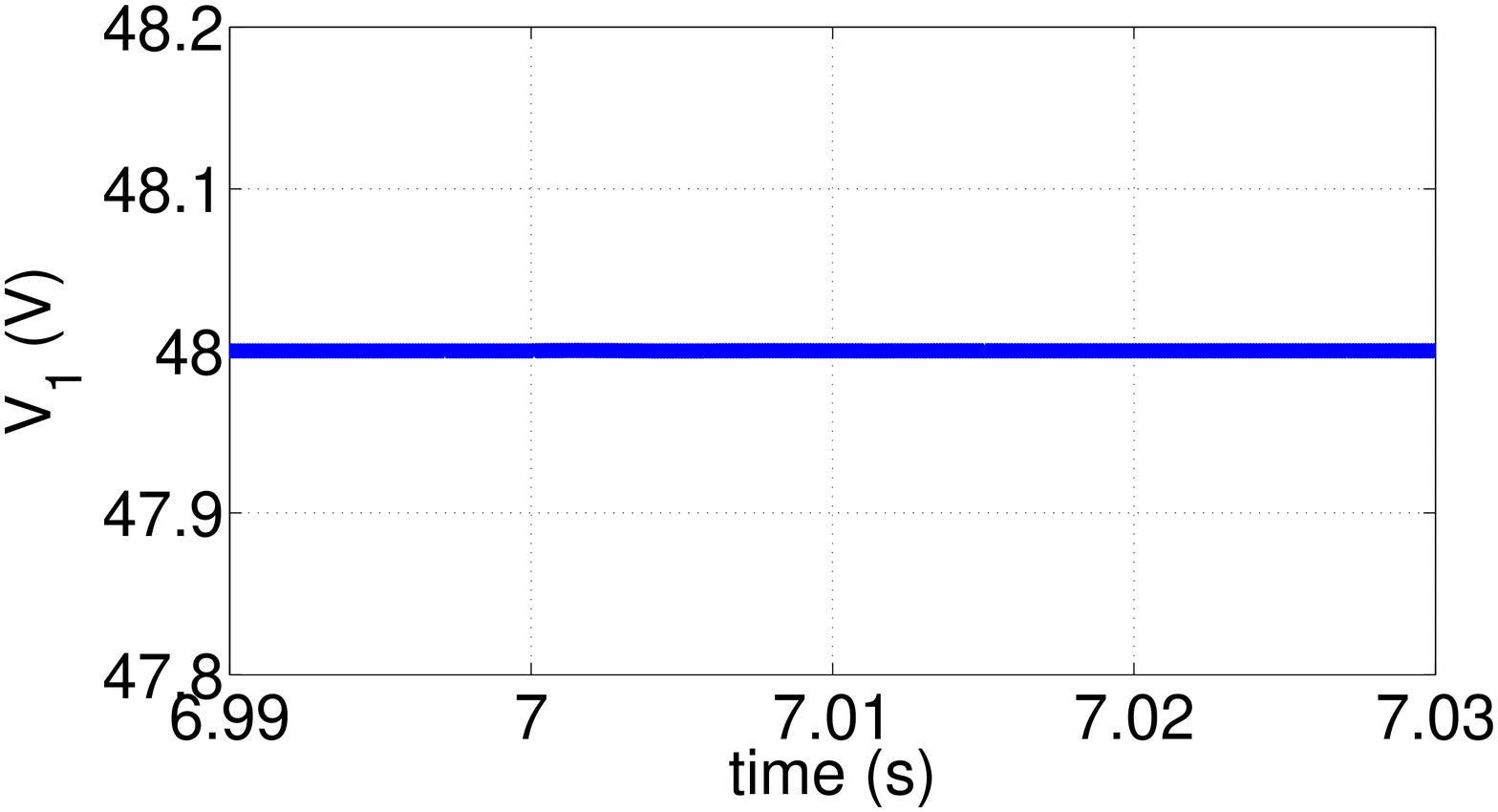}
                        \caption{Voltage at $PCC_1$.}
                        \label{fig:Sc2_V1_2}
                      \end{subfigure}
                      \begin{subfigure}[!htb]{0.48\textwidth}
                        \centering
                        \includegraphics[width=1\textwidth, height=130pt]{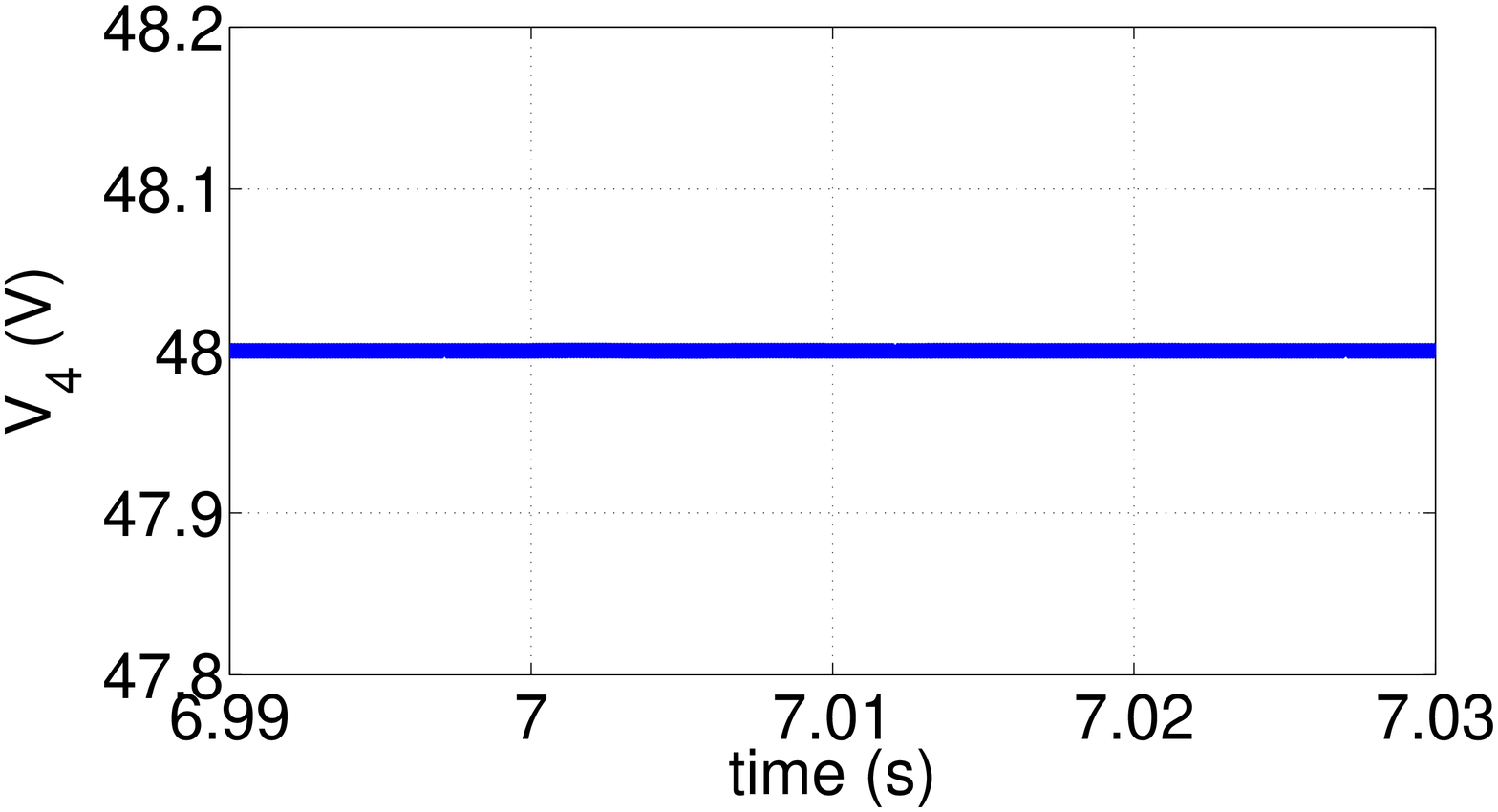}
                        \caption{Voltage at $PCC_4$.}
                        \label{fig:Sc2_V4_2}
                      \end{subfigure}
                     \caption{Scenario 2 - Performance of PnP
                       decentralized voltage controllers during the
                       hot-unplugging of DGU 3 at $t=7$ s.}
                      \label{fig:Sc2_V_2}                 
                    \end{figure}

               \section{Conclusions}
          \label{sec:conclusions}
          In this paper, a decentralized control scheme for
          guaranteeing voltage stability in DC ImGs was presented. The
          main feature of the proposed approach is that, whenever a
          plugging-in or -out of DGUs is required, only a limited
          number of local controllers must be updated. Moreover, as mentioned in Section ~\ref{sec:differentC}, local voltage controllers should be coupled with a higher control layer devoted to power flow regulation so as to orchestrate mutual help among DGUs. To this purpose, we will study if and how ideas from secondary control of ImGs \cite{shafiee2014hierarchical} can be reappraised in our context.

     \clearpage

     \appendix
     \section{Matrices appearing in microgrid models}
     \label{sec:AppMatrices}

     The following appendix collects all matrices appearing in Section \ref{sec:Model}.
     
     \subsection{Matrices in the model \eqref{eq:sysdistABCDM}}
          \label{sec:AppMasterSlave}
	\begin{equation*}
          \setlength{\arraycolsep}{6pt}
          \setcounter{MaxMatrixCols}{6}
          A=\matr{
            0 & \frac{1}{C_{ti}} & \frac{1}{C_{ti}} & 0 & 0 & 0\\
            -\frac{1}{L_{ti}} & -\frac{R_{ti}}{L_{ti}} & 0 & 0 & 0 & 0\\
            -\frac{1}{L_{ij}} & 0 & -\frac{R_{ij}}{L_{ij}} & 0 & \frac{1}{L_{ij}} & 0 \\
            \frac{1}{L_{ji}} & 0 & 0 &  -\frac{R_{ji}}{L_{ji}} & -\frac{1}{L_{ji}} & 0 \\  
            0 & 0 & 0 & \frac{1}{C_{tj}} & 0 & \frac{1}{C_{tj}} \\
            0 & 0 & 0 & 0 & -\frac{1}{L_{tj}} & -\frac{R_{tj}}{L_{tj}}
          }
        \end{equation*}
        
	\begin{equation*}
          \label{ABCDM}
          B=\matr{
            0 & 0 \\
            \frac{1}{L_{ti}} & 0 \\
            0 & 0 \\
            0 & 0 \\
            0 & 0 \\
            0 & \frac{1}{L_{tj}}
          }\qquad
          C^T =\matr{
            1 & 0 \\
            0 & 0 \\
            0 & 0 \\
            0 & 0 \\
            0 & 1 \\
            0 & 0
          }\qquad
          M=\matr{
            -\frac{1}{C_{ti}} & 0  \\
            0 & 0 \\
            0 & 0 \\
            0 & 0 \\
            0 & -\frac{1}{C_{tj}} \\
            0 & 0 
          }
	\end{equation*}
	
     \subsection{Matrices in the QSL model \eqref{eq:subsysDGUi} and \eqref{eq:subsysLine}}
          \label{sec:AppMasterMaster}
	
          \paragraph{DGU-$i$, $i\in\{1,2\}$}
          \begin{equation*}
            \renewcommand\arraystretch{1.5}
            A_{ii}=\begin{bmatrix}
              -\frac{1}{R_{ij}C_{ti}} & \frac{1}{C_{ti}} \\
              -\frac{1}{L_{ti}} & -\frac{R_{ti}}{L_{ti}} \\
            \end{bmatrix}
          \end{equation*}
          \begin{equation*}
            \renewcommand\arraystretch{1.5}
            A_{ij}=
            \begin{bmatrix}
              \frac{1}{R_{ij}C_{ti}} & 0 \\
              0 & 0 
            \end{bmatrix}
          \end{equation*}
          
          \begin{equation*}
            B_{i}=\begin{bmatrix}
              0 \\
              \frac{1}{L_{ti}}
            \end{bmatrix}
            \qquad
            M_{i}=\begin{bmatrix}
              -\frac{1}{C_{ti}} \\
              0 \\
            \end{bmatrix}
            \qquad
            C_{i}=\begin{bmatrix}
              1&0\\
              0&1
            \end{bmatrix} 
            \qquad
            H_{i}=\begin{bmatrix}
              1 & 0 
            \end{bmatrix}
          \end{equation*}
          
          \paragraph{Line $i\neq j$}
          \begin{equation}
            \renewcommand\arraystretch{1.2}
            \label{matrixss3}
            A_{li,ij}=\begin{bmatrix}
              -\frac{1}{L_{ij}} & 0\\
            \end{bmatrix}\quad	
            A_{lj,ij}=\begin{bmatrix}
              \frac{1}{L_{ij}} & 0\\
            \end{bmatrix}\quad
            A_{ll,ij}=
              -\frac{R_{ij}}{L_{ij}}
          \end{equation}
          
     \subsection{QSL model of microgrid composed of $N$ DGUs}
          \label{sec:AppNDGunit}          
          \paragraph{DGU-$i$, $i\in\DD$}
          \begin{equation}
            \label{eq:Aii}
            \renewcommand\arraystretch{2}
            A_{ii}=\begin{bmatrix}
              \sum_{j\in\NN_i}-\frac{1}{R_{ij}C_{ti}} & \frac{1}{C_{ti}} \\
              -\frac{1}{L_{ti}} & -\frac{R_{ti}}{L_{ti}} \\
            \end{bmatrix}
          \end{equation}          
          \begin{equation}
            \label{eq:Aij}
            \renewcommand\arraystretch{2}
            A_{ij}=
            \begin{bmatrix}
              \frac{1}{R_{ij}C_{ti}} & 0 \\
              0 & 0 
            \end{bmatrix}
	\end{equation}
        We remind that $R_{ij}$ and $L_{ij}$ are the resistance and
        the inductance of the line between DGU $i$ and DGU
        $j$. Moreover, matrices $B_i$, $C_i$, $M_i$ and $H_i$ are
        equal to those appearing in Section \ref{sec:AppMasterMaster}.
	
        \paragraph{Overall model of a microgrid composed by $N$ DGUs}
             \label{sec:AppOverallsys}

             \begin{equation}
               \label{TheSystem}
               \begin{aligned}
                 \begin{bmatrix}
                   \subss{\dx}{1} \\
                   \subss{\dx}{2} \\
                   \subss{\dx}{3} \\
                   \vdots \\
                   \subss{\dx}{N}
                 \end{bmatrix} 
                 &= 
                 \underbrace{\left[\begin{array}{ccccc}
                       A_{11} & A_{12} & A_{13} & \dots  & A_{1N} \\
                       A_{21} & A_{22} & A_{23} & \dots  & A_{2N} \\
                       A_{31} & A_{32} & A_{3l} & \dots  & A_{3N} \\
                       \vdots & \vdots & \vdots & \ddots & \vdots\\
                       A_{N1} & A_{N2} & A_{N3} & \dots  & A_{NN}
                     \end{array}
                   \right]}_{\mbf{A}}
                 \begin{bmatrix}
                   \subss{x}{1} \\
                   \subss{x}{2} \\
                   \subss{x}{3} \\
                   \vdots \\
                   \subss{x}{N}
                 \end{bmatrix} 
                 +\\
                 &+\,
                 \underbrace{\begin{bmatrix}
                     B_{1} & 0 & \dots & 0\\
                     0 & B_{2} & \ddots & \vdots\\
                     \vdots & \ddots & \ddots & 0\\
                     0& \dots & 0  & B_{N}
                   \end{bmatrix}}_{\mbf{B}}
                 \begin{bmatrix}
                   \subss{u}{1}\\
                   \subss{u}{2}\\
                   \vdots\\
                   \subss{u}{N}
                 \end{bmatrix}
                 + \underbrace{\begin{bmatrix}
                     M_{1} & 0 & \dots & 0\\
                     0 & M_{2} & \ddots & \vdots\\
                     \vdots & \ddots & \ddots & 0\\
                     0& \dots & 0  & M_{N}
                   \end{bmatrix}}_{\mbf{M}}
                 \begin{bmatrix}
                   \subss{d}{1}\\
                   \subss{d}{2}\\
                   \vdots\\
                   \subss{d}{N}
                 \end{bmatrix}\\    
                 \begin{bmatrix}
                   \subss{y}{1}\\
                   \subss{y}{2}\\
                   \subss{y}{3}\\
                   \vdots\\
                   \subss{y}{N}
                 \end{bmatrix}
                 &=
                 \underbrace{\left[\begin{array}{ccccc}
                       C_{1} & 0 & 0 & \dots & 0 \\
                       0 & C_{2} & 0 & \ddots & \vdots \\
                       0 & 0 & C_{3} & \ddots & 0 \\
                       \vdots & \ddots & \ddots &\ddots & 0\\
                       0 & \dots & 0 & 0  & C_{N}
                     \end{array}
                   \right]}_{\mbf{C}}
                 \begin{bmatrix}
                   \subss{x}{1} \\
                   \subss{x}{2} \\
                   \subss{x}{3} \\
                   \vdots \\
                   \subss{x}{N}
                 \end{bmatrix}\\
                 \begin{bmatrix}
                   \subss{z}{1}\\
                   \subss{z}{2}\\
                   \subss{z}{3}\\
                   \vdots\\
                   \subss{z}{N}
                 \end{bmatrix}
                 &=
                 \underbrace{\begin{bmatrix}
                     H_{1} & 0 & 0 & \dots & 0 \\
                     0 & H_{2} & 0 & \ddots & \vdots \\
                     0 & 0 & H_{3} & \ddots & 0 \\
                     \vdots & \ddots & \ddots &\ddots & 0\\
                     0& \dots & 0 & 0  & H_{N}
                   \end{bmatrix}}_{\mbf{H}}\begin{bmatrix}
                   \subss{y}{1}\\
                   \subss{y}{2}\\
                   \subss{y}{3}\\
                   \vdots\\
                   \subss{y}{N}
                 \end{bmatrix}.
               \end{aligned} 
             \end{equation}
     \clearpage
     
\section{Bumpless control transfer}
     \label{sec:AppBumpless}
Since the controller is a dynamic system, it is necessary to make sure
that the state of the system is correct when a switch of the
controller (i.e. a plugging-in or unplugging operation) is required. 
Assuming that the control switch is made at a certain point in time $\bar t$, we
call $u_{prec,i}$ the control signal produced by the controller
$\mathcal C_{i}$
up to time $\bar t$. It might happen that the updated controller
will provide a control variable $u_{i}$ different from $u_{prec,i}$. Therefore,
it is necessary to ensure there is no substantial change in the two
outputs at $\bar t$. This is called \textit{bumpless control transfer}
\cite{aastrom2006advanced}. 

A bumpless control transfer implementation of PnP local controller for
system $\hat\Sigma_{i}^{DGU}$ is illustrated in Figure
\ref{fig:bumpless_scheme}. 

\begin{figure}[htb]
                        \centering
                        \includegraphics[width=1\textwidth, height=150pt]{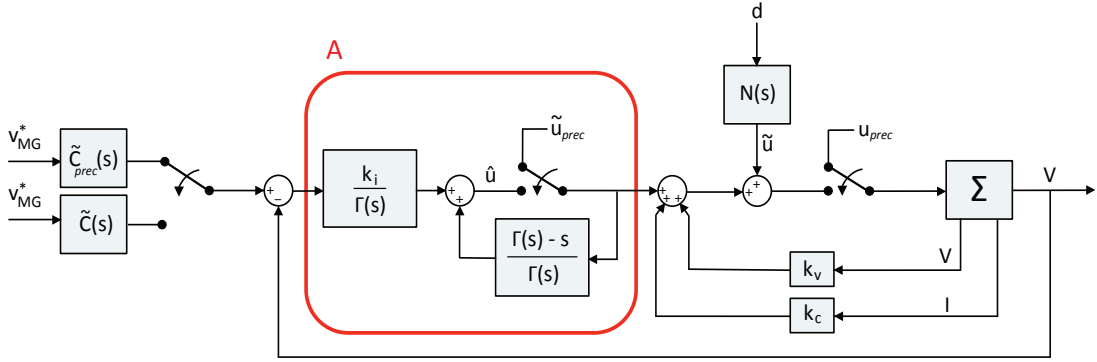}
                        \caption{Bumpless control transfer implementation.}
                        \label{fig:bumpless_scheme}
                      \end{figure}

For the sake of simplicity we drop here the
index $i$ of the subsystem and associated local variables and all
switches in Figure \ref{fig:bumpless_scheme} are assumed to commute at time $\bar t$. In the
Figure, vector
$$ K = [k_v\mbox{ }k_c\mbox{ }k_i]^T$$
contains the parameters of the local controller to be activated at
time $\bar t $. Notice that the integrator embedded in the DGU model
for zeroing the steady-state error is replaced by block A (highlighted
in red in
Figure \ref{fig:bumpless_scheme}), where the
polynomial $\Gamma(s)$ has to be chosen such that $k_i>\Gamma(0)$ and the
transfer function $$\Psi = \frac{\Gamma(s)-s}{\Gamma(s)}$$ is
asymptotically stable and realizable. In block A, a switch is present so that
the signal is either $\tilde{u}_{prec}$ (up to time $\bar
t$) or $\hat u$ (right after $\bar t$). The variable $\tilde{u}_{prec}$ is given by
\begin{eqnarray}
\tilde{u}_{prec} = u_{prec}-k_vV-k_cI_t-\tilde{u}
\end{eqnarray}
where $\tilde{u}$ is the additional input produced by compensator
$N(s)$, computed with respect to the dynamics of the system after the
commutation (set $N(s)=0$ if such a compensation is not implemented). We highlight that since there could be a
transient in the $\hat u$ response to track signal
$\tilde{u}_{prec}$, it is fundamental to wait for the two signals to
become similar before proceeding with the commutation. In this way we
avoid jumps in the control variable.

Furthermore, if an
optional prefilter of the reference is implemented, at time $\bar t$
is also necessary to commute from transfer
function $\tilde{C}_{prec}(s)$ to $\tilde{C}(s)$, since each
plugging-in or unplugging operation of other DGUs in the overall ImG
lead to a variation of the local dynamics of the considered subsystem
$\hat\Sigma_{i}^{DGU}$ (see the term
$\sum_{j\in\NN_i}-\frac{1}{R_{ij}C_{ti}}$ in \eqref{eq:Aii}).
     \clearpage
     \section{Electrical and simulation parameters of Scenario 1 and 2}
     \label{sec:AppElectrPar}
     In this appendix, we provide all the electrical and simulation
     parameters of Scenarios 1 and 2 (which are described in Sections \ref{sec:scenario_1} and \ref{sec:scenario_2}, respectively).
   \begin{table}[!htb]                 
                      \centering
                      \begin{tabular}{*{4}{c}}
                        \toprule
                       Parameter & Symbol & Value \\
                        \midrule
                        DC power supply & $V_{DC}$ & 100 $V$ \\
                        Output capacitance & $C_{t*}$ & 2.2 $mF$\\
                        Converter inductance & $L_{t*}$ & 1.8$\mbox{ }mH$\\
                        Inductor + switch loss resistance & $R_{t*}$ & 0.2 $\Omega$ \\
                        Switching frequency & $f_{sw}$ & 10 kHz\\
                        \midrule
                        Transmissione line inductance & $L_{*\circ}$ &
                                                                       1.8$\mbox{
                                                                       }\mu
                                                                       H$\\
                        Transmission line resistance & $R_{*\circ}$ & 0.05 $\Omega$\\
                        \bottomrule
                      \end{tabular}
                      
                      \caption{Electrical setup and line parameters}	
                      \label{tbl:electrical_setup}
                    \end{table}
                    \begin{table}[!htb]
       \caption{VSC filter parameters for DGUs $\subss{\hat{\Sigma}}{i}^{DGU}$, $i=\{1,\dots,6\}$ in Scenario 2.}	
       \label{tbl:diffpar5}
       \centering
       \begin{tabular}{*{4}{c}}
         \toprule
         DGU & Resistance $R_t (\Omega)$ & Capacitance $C_t(mF$) & Inductance $L_t (mH$)\\
         \midrule
         $\subss{\hat{\Sigma}}{1}^{DGU}$& 0.2 & 2.2 & 1.8\\
         $\subss{\hat{\Sigma}}{2}^{DGU}$& 0.3 & 1.9 & 2.0\\
         $\subss{\hat{\Sigma}}{3}^{DGU}$& 0.1 & 1.7 & 2.2\\
         $\subss{\hat{\Sigma}}{4}^{DGU}$& 0.5 & 2.5 & 3.0\\
         $\subss{\hat{\Sigma}}{5}^{DGU}$& 0.4 & 2.0 & 1.2\\
         \midrule
         $\subss{\hat{\Sigma}}{6}^{DGU}$& 0.6 & 3.0 & 2.5\\
         \bottomrule
       \end{tabular}
     \end{table}

     \begin{table}[!htb]
       \caption{Transmission lines parameters for Scenario 2.}	
       \label{tbl:linespar5}
       \centering
       \begin{tabular}{*{3}{c}}
         \toprule
         Connected DGUs $(i,j)$ & Resistance $R_s (\Omega)$ & Inductance
                                                            $L_s (\mu H)$ \\
         \midrule
         $(1,2)$ & 0.05 & 2.1 \\
         $(1,3)$ & 0.07 & 1.8 \\
         $(3,4)$ & 0.06 & 1.0 \\
         $(2,4)$ & 0.04 & 2.3 \\
         $(4,5)$ & 0.08 & 1.8 \\
         \midrule
         $(1,6)$ & 0.1 & 2.5 \\
         $(5,6)$ & 0.08 & 3.0 \\
         \bottomrule
       \end{tabular}
     \end{table}

     \begin{table}[!htb]
       \caption{Common parameters of DGUs $\subss{\hat{\Sigma}}{i}^{DGU}$, $i=\{1,\dots,6\}$ in Scenario 2.}	
       \label{tbl:commpar5}
       \centering
       \begin{tabular}{*{3}{c}}
         \toprule
         Parameter & Symbol & Value \\
         \midrule
                         \hspace{30mm} Electrical parameters  \\
                        \hline
                        DC power supply & $V_{DC}$ & 100 $V$ \\
                        Output capacitance & $C_{t}$ & 2.2 $mF$\\
                        Converter inductance & $L_{t}$ & 1.8 $mH$\\
                        Inductor+switch loss resistance & $R_{t}$ & 0.2 $\Omega$ \\
                        Switching frequency & $f_{sw}$ & 10 kHz\\
  
         \bottomrule
       \end{tabular}
     \end{table}

     \clearpage

     \bibliographystyle{IEEEtran}
     \bibliography{microgrids-report}

\end{document}